\documentclass{article}[11pt]
\usepackage[latin1]{inputenc}
\usepackage{scrpage2}
\usepackage{amsthm,amsmath,amsfonts,dsfont}
\usepackage{graphicx}
\usepackage{eepic}
\usepackage{epsfig}
\usepackage{multirow}
\usepackage{enumerate}
\usepackage{textcomp}
\usepackage{color}
\usepackage{array}
\usepackage{a4wide}
\usepackage{amssymb}
\usepackage{pdflscape}
\usepackage{srcltx}

\newcommand{\comment}[1]{}

\newtheorem{assumption}{Assumption}[section]

\newtheorem{definition}[assumption]{Definition}

\newtheorem{lemma}[assumption]{Lemma}
\newtheorem{proposition}[assumption]{Proposition}
\newtheorem{theorem}[assumption]{Theorem}
\newtheorem{corollary}[assumption]{Corollary}
\newtheorem{remark}[assumption]{Remark}
\newtheorem{example}[assumption]{Example}

\begin{document}

\title{A Law of Large Numbers for Limit Order Books\thanks{Most of this work was done while the first author was visiting CEREMADE; grateful acknowledgment is made for hospitality and financial support. Both authors acknowledge financial support through the SFB 649 ``Economic Risk'' and Deutsche Bank AG. We thank seminar participants at various institutions for valuable comments and suggestions.}}
\author{Ulrich Horst\thanks{Corresponding author: horst@math.hu-berlin.de}, Michael Paulsen \\Institut f\"{u}r Mathematik\\
Humboldt-Universit\"{a}t zu Berlin\\
Unter den Linden 6\\
10099 Berlin
}
\date{First Version: January 22nd, 2013 \\ This Version: November 3rd, 2014}
\maketitle
\begin{abstract}
\noindent We define a stochastic model of a two-sided limit order book in terms of its key quantities \textit{best bid [ask] price} and the \textit{standing buy [sell] volume density}. For a simple scaling of the discreteness parameters, that keeps the expected volume rate over the considered price interval invariant, we prove a limit theorem. The limit theorem states that, given regularity conditions on the random order flow, the key quantities converge in probability to a tractable continuous limiting model. In the limit model the buy and sell volume densities are given as the unique solution to first-order linear hyperbolic PDEs, specified by the expected order flow parameters.
\end{abstract}

\vspace{4mm}

\noindent {\bf Key words:} limit order book, scaling limit, averaging principle, queueing theory

\vspace{2mm}

\noindent {\bf AMS Subject Classification:} 60B11, 90B22, 91B70

\vspace{6mm}

%\centerline{\sl Preliminary; Comments Welcome}

\thispagestyle{empty}

\renewcommand{\baselinestretch}{1.15}\normalsize

\section*{Introduction}
A significant part of financial transactions is nowadays carried out over an electronic limit order book (LOB) where unexecuted orders are stored and displayed while awaiting execution. From a mathematical perspective, LOBs are high-dimensional complex priority queueing systems. Incoming limit orders can be placed at many different price levels while incoming market orders are matched against standing limit orders according to a set of priority rules. Almost all exchanges give priority to orders submitted at more competitive prices (``price priority'') and displayed orders have priority over hidden orders at the same level (``display priority''). Orders with the same display status and submission price are usually served on a first-come-first-serve basis. The inherent complexity of limit order books renders their mathematical analysis challenging. In this paper, we propose a simple and transparent queueing theoretic LOB model whose dynamics converges to a coupled ODE:PDE system after suitable scaling. The limiting system can be solved in closed form.

There is a significant economic and econometric literature on LOBs including Biais at al. \cite{Biais}, Easley and O'Hara \cite{EasleyOHara}, Foucault et al. \cite{Foucault}, Gloston and Milgrom \cite{GlostenMilgrom}, Parlour \cite{Parlour98}, Rosu \cite{Rosu} and many others that puts a lot of emphasis on the realistic modeling of the working of the LOB, and on its interaction with traders' order submission strategies. There are, however, only few papers that analyze order flows and the resulting LOB dynamics in a mathematically rigorous manner. Among the first was the one by Kruk \cite{Kruk2003}. He studied a queueing theoretic model of a transparent double auction in continuous time. The microstructure may be interpreted as that of a simple limit order book, if one considers a buyer as a buy limit order and a seller as a sell limit order. At the auction, there are $N\in\mathbb{N}$ possible prices of the security and thus 2N different classes of customers (N classes of buyers and N classes of sellers). Kruk established diffusion and fluid limits. The diffusion limit states that for N=2 possible prices, the scaled number of outstanding buy orders at the lower price and the scaled number of outstanding sell orders at the higher price converge weakly to a semimartingale reflected two-dimensional Brownian motion in the first quadrant. The fluid limit is such that various LOB quantities converge weakly to affine functions of time. 

Since information on the best bid and ask price and volume at different price levels is available to all market participants, it is natural to assume that the order arrival dynamics depends on the current state of the order book. The feature of conditional state-dependence was considered by Cont et al. \cite{ContStoikovTalreja}, who proposed a continuous time stochastic model with a finite number of possible securities prices where events (buy/sell market order arrival, buy/sell limit order placement and cancelation) are modeled using independent Poisson processes. The arrival rates of limit orders depend on the distance to the best bid and ask price in a power-law fashion. The authors were able to show that the state of the order book, defined as a vector containing all volumes in the order book at different prices, is an ergodic Markov process. Using this fact, several key quantities such as  the probabilities of a mid price move, a move in the bid price before a move in the ask price, or the probability of volume execution before a price move could be computed and benchmarked against real data without taking scaling limits.

Cont and de Larrard \cite{ContDeLarrard2012a} considered a scaling limit in the diffusion sense for a Markovian limit order market in which the state is represented by the best bid and ask price and the queue length, i.e. the number of orders at the best bid and ask price, respectively. With this reduction of the state space, under symmetry conditions on the spread and stationarity assumptions on the queue lengths, it was shown that the price converges to a Brownian motion with volatility specified by the model parameters in the diffusion limit. Very recently, Cont and de Larrard \cite{ContDeLarrard2012b} studied the reduced state space under weaker conditions and proved a refined diffusion limit by showing that under heavy traffic conditions the bid and ask queue lengths are given by a two-dimensional Brownian motion in the first quadrant with reflection to the interior at the boundaries, similar to the diffusion limit result for $N=2$ prices in Kruk \cite{Kruk2003}.

In the framework analyzed by Abergel and Jedidi \cite{AbergelJedidi2011}, the volumes of the order book at different distances to the best bid and ask were modeled as a finite dimensional continuous time Markov chain and the order flow as independent Poisson processes. Under the assumption that the width of the spread is constant in time, using Foster-Lyapunov stability criteria for the Markov chain, the authors proved ergodicity of the order book and a diffusion limit for the mid price. In the diffusion limit, the mid price is a Brownian motion with constant volatility given by the averaged price impact of the model events on the order book.

In this paper, we prove a law of large numbers result for the whole book (prices and volumes). Specifically, we propose a continuous-time model of a two-sided state-dependent order book with random order flow and cancelation and countably many submission price levels. The buy and the sell side volumes are coupled through the best bid and ask price dynamics.\footnote{The coupling of the buy and sell sides through prices is essential for the limiting volume dynamics to follow a PDE. It is \textit{not} essential for obtaining a scaling limit per se. Our mathematical framework is flexible enough to allow for dependencies of order flows on standing volumes, but that one would lead to a function-valued ODE as the scaling limit, rather than a PDE. The PDE-scaling is so much more transparent that it justifies, in our view, the restriction of order flow dependencies on prices.}
We model the buy and sell side volumes as \textit{density functions in relative price coordinates}, i.e. relative to their distance to the best bid/ask prices.
%the integral of the relative buy [sell] volume density function between two price levels is equal to the buy [sell] volume at the lower of the prices.
Volumes at positive distances are defined as standing limit volume, observable in the order book (``visible book''), and volume at negative price distances corresponds to volume that \textit{would} be placed in the spread if the next order would be a spread limit order (``shadow book''). The state of the model at any point in time is thus a quadruple comprising the best bid price, the best ask price, the relative buy volume density function and the relative sell volume density function. The state dynamics is then defined in terms of a recursive stochastic process taking values in a function space.

In order to establish our scaling limit we use a mathematical framework similar to the recurrently defined semi-Markov processes in Anisimov \cite{Anisimov95,Anisimov2002} and Gikhman and Skorokhod \cite{GikhmanSkorokhodTTSPIII}, but for stochastic processes taking values in Banach spaces. When the analysis of the market is limited to prices as in e.g. Garman \cite{Garman}, Bayraktar et al. \cite{BayraktarHorstSircar} or Horst and Rothe \cite{HorstRothe} or to the joint dynamics of prices and \textit{aggregate volumes} (e.g. at the top of the book) as in Cont and de Larrard  \cite{ContDeLarrard2012b,ContDeLarrard2012a}, then the limiting dynamics can naturally be described by ordinary differential equations or real-valued diffusion processes, depending on the choice of scaling\footnote{We refer to Mandelbaum et al. \cite{MandMasseyPatts1998} for general approximation results for queueing systems.}.

The analysis of the whole book including the \textit{distribution} of standing volume across many price levels is much more complex. Osterrieder \cite{Osterrieder} modeled LOBs using measure-valued diffusions. Our approach is based on an averaging principle for Banach space-valued processes. The key is a uniform law of large numbers for Banach space-valued triangular martingale difference arrays. It allows us to show that the volume densities take values in $L^2$ and that  the noise in the order book models vanishes in the limit with our choice of scaling.  

Our scaling limit requires two time scales: a fast time scale for cancelations and limit order placements outside the spread (events that do not lead to price changes), and a comparably slow time scale for market order arrivals and limit order placements in the spread (events that lead to price changes). The choice of time scales captures the fact that in real world markets significant proportions of orders are never executed. Mathematically, the different time scales imply that aggregate cancelations and limit order placements outside the spread between consecutive price changes can be approximated by their expected values.

Our main result states that when limit buy [sell] orders are placed at random distances from the best ask [bid] price and the price tick tends to zero, order arrival rates tend infinity, and the impact of an individual placement/cancelation on the standing volume tends to zero, then the sequence of scaled order book models converges in probability uniformly over compact time intervals to a deterministic limit. The limiting model is such that the best bid and ask price dynamics can be described in terms of two coupled ODE:s, while the dynamics of the relative buy and sell volume density functions can be described in terms of two linear first-order hyperbolic PDE:s with variable coefficients. The PDE can be given in closed form.

The remainder of  this paper is organized as follows. In the first section we define a sequence of limit order book in terms of four scaling parameters: price tick, expected waiting time between two consecutive orders, volume placed/canceled and the proportion of order arrivals leading to price changes. In Section \ref{chapter-prices} we establish convergence of the bid/ask price dynamics to a 2-dimensional ODE; Section  \ref{chapter-volumes} is devoted to the analysis of the limiting volume dynamics. The uniform law of large numbers for triangular martingale difference arrays as well as useful auxiliary results are proved in an appendix.

%%%%%%%%%%%%%%%%%%%%%%%%%%%
%%%%%%%%%%%%%%%%%%%%%%%%%%%
%%%%%%%%%%%%%%%%%%%%%%%%%%%
%%%%%%%%%%%%%%%%%%%%%%%%%%%
%%%%%%%%%%%%%%%%%%%%%%%%%%%
%%%%%%%%%%%%%%%%%%%%%%%%%%%

\section{A sequence of discrete order book models} \label{chapter-model}
In electronic markets orders can be submitted for prices that are multiples of the \textit{price tick}, the smallest increment by which the price  can move.
In this section, we introduce a sequence of order book models for which we establish a scaling limit when the price tick and impact of a single order on the state of the book tend to zero, while the rate of order arrivals tends to infinity.

The sequence of models is indexed by $n \in {\mathbb N}$. We assume that the set of price levels at which orders can be submitted in the $n$:th models is $\{x^{(n)}_j\}_{i \in \mathbb{Z}}$ where ${\mathbb Z}$ denotes the one-dimensional integer lattice.\footnote{The assumption that there is no minimum price is made for analytical convenience and can easily be relaxed.} We put $x^{(n)}_j := j \cdot \Delta x^{(n)}$ for $j \in {\mathbb Z}$ where $\Delta x^{(n)}$ is the tick size in the $n$:th model.

The \textit{state} of the book changes due to incoming order flow and cancelations of standing volume. The state after $k \in {\mathbb N}$  such \textit{events} will be described by a random variable $S^{(n)}_{k}$ taking values in a suitable \textit{state space} $E$. In the $n$:th model, the $k$:th event occurs at a random point in time $\tau^{(n)}_k$. The time between two consecutive events will be tending to zero sufficiently fast as $n \to \infty$. The state and time dynamics will be defined, respectively, as
\begin{equation} \label{def:OrigStateRecurr}
S^{(n)}_{0}:=s^{(n)}_{0}, \quad S^{(n)}_{k+1}:=S^{(n)}_{k}+\mathcal{D}^{(n)}_{k}(S^{(n)}_{k})
\end{equation}
and
\begin{equation} \label{def:OrigStateRecurr2}
\tau^{(n)}_{0}:=0, \quad \tau^{(n)}_{k+1}:=\tau^{(n)}_{k}+\mathcal{C}^{(n)}_{k}(S^{(n)}_{k}).
\end{equation}
Here $s^{(n)}_0 \in E$ is a deterministic initial state, and $\mathcal{D}^{(n)}_{k}(S^{(n)}_k):E\rightarrow E$ and $\mathcal{C}^{(n)}_{k}(S^{(n)}_{k}): E \to [0, \infty)$ are random operators that will be introduced below. The conditional expected increment of the state sequence given by $S^{(n)}_k$ will be denoted $\mathbb{E} \left[ \mathcal{D}^{(n)}_{k}(S^{(n)}_k) \right]$; the unconditional increment $\mathbb{E} \left[ \mathcal{D}^{(n)}_{k} \right]$.

\subsection{The order book models}
In the sequel we specify the dynamics of our order book models. Throughout all random variables will be defined on a common probability space $\left(\Omega,\mathcal{F},\mathbb{P}\right)$.

%%%%%%%%%%%%%%%%%%%
%%%%%%%%%%%%%%%%%%%
%%%%%%%%%%%%%%%%%%%

\subsubsection{The initial state}

The initial state of the book in the $n$:th model is given by a pair $(B^{(n)}_0,A^{(n)}_0)$ of best bid and ask prices together with the standing buy and sell limit order volumes at the various price levels. It will be convenient to identify the standing volumes with step functions
\begin{equation*}
	v^{(n)}_{b,0}(x):=\sum_{j=0}^{\infty} v^{(n),j}_{b,0} \mathbf{1}_{[x^{(n)}_j, x^{(n)}_{j+1})}( x), \quad v^{(n)}_{s,0}(x):=\sum_{j= 0}^{\infty}v^{(n),j}_{s,0} \mathbf{1}_{[x^{(n)}_j, x^{(n)}_{j+1})}(x) \quad (x \geq 0)
\end{equation*}
that specify the liquidity available for buying and selling \textit{relative} to the best bid and ask price. The liquidity available for buying (sell side of the book) $j \in \mathbb{N}_0$ ticks \textit{above} the best ask price at the price level $x^{(n)}_{A^{(n)}_0 + j}$ is
\[
	\int_{x^{(n)}_j}^{x^{(n)}_{j+1}} v^{(n)}_{s,0}(x) \mathrm{d}x = v^{(n),j}_{s,0} \cdot \Delta x^{(n)}.
\]	
The volume available for selling (buy side of the book\footnote{Notice that the liquidity available for buying is captured by the sell side of the book and vice versa.}) $l \in \mathbb{N}_0$ ticks \textit{below} the best bid price at the price level $x^{(n)}_{B^{(n)}_0 - l}$ is
\[
	\int_{x^{(n)}_l}^{x^{(n)}_{l+1}} v^{(n)}_{b,0}(x) \mathrm{d}x = v^{(n),l}_{b,0} \cdot \Delta x^{(n)}.
\]	
\begin{figure}
\centering
\includegraphics[width=10cm]{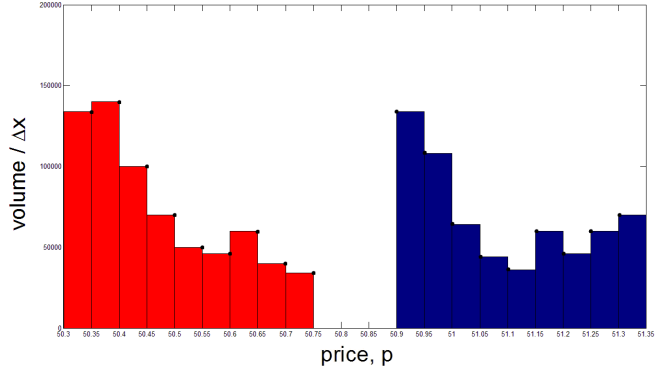}
\caption{Snapshot of the book.}
\end{figure}

In order to conveniently model placements of limit orders into the spread, we extend $v^{(n)}_{b,0}$ and $v^{(n)}_{s,0}$ to the negative half-line. The collection of volumes standing at negative distances from the best bid/ask price is referred to as the \textit{shadow book}. The shadow book will undergo the same dynamics as the standing volume (``visible book''). At any point in time it specifies the volumes that will be placed into the spread should such an event occur next\footnote{One has to specify the volumes placed into the spread somehow. Our choice of shadow books is one such way. The role of the shadow book will be further clarified in the following subsection when we define the impact of order arrivals on the state of the book. Its initial state is part of the model; future states will undergo the analogous dynamics as those of the visible book.}. The role of the shadow will be further illustrated in Section \ref{section-passive} below; see Figures 3 and 4.

\begin{definition}
In the $n$:th model the initial state of the book is given by a quadruple $$S^{(n)}_{0}(\cdot)=\left(B^{(n)}_{0},A^{(n)}_{0},v^{(n)}_{b,0}(\cdot),v^{(n)}_{s,0}(\cdot)\right) $$ where $B^{(n)}_0 \leq A^{(n)}_0$ are the best bid/ask price and the step functions $v^{(n)}_{b,0}, v^{(n)}_{s,0} : {\mathbb R} \to [0,\infty)$ are to be interpreted as follows:
\begin{equation}
v^{(n)}_{b,0}(x) \, \left[v^{(n)}_{s,0}(x)\right]:=
\begin{cases}
\text{standing buy [sell] limit order volume density at price distance $x$}\\
\vspace{3mm}
\text{below [above] the best bid [ask] price, for $x\geq 0$ (visible book)}\\
\text{potential buy [sell] limit order volume density at price distance $x$}\\
\text{above [below] the best bid [ask] price, for $x< 0$ (shadow book).}
\end{cases}
\label{eq:standVol}
\end{equation}
\end{definition}

 Throughout, we shall use the notation $f={\cal O}(g)$ and $f=o(g)$ to indicate that the function $f$ grows asymptotically no faster than $g$, respectively that $|f(x)/g(x)| \to 0$ as $x \to \infty$. With this, we are ready to state our conditions on the initial states. In particular, we assume that the initial volume density functions vanish outside a compact price interval.\footnote{This assumption, which may be generalized, considerably simplifies some of the analysis that follows.}

\begin{assumption}[Convergence of initial states]\label{Assumption-1}
	The initial volume density functions vanish outside a compact interval $[-M,M]$ for some $M > 0$ and and satisfies: 
\[
	| v^{(n)}_{r,0} (\cdot \pm \Delta x^{(n)}) - v^{(n)}_{r,0} (\cdot) |_{\infty} = {\cal O}(\Delta x^{(n)}).
\]
Moreover, 
\[
	\lim_{n \to \infty} || v^{(n)}_{r,0} - v_{r,0} ||_{L^2} = 0 \quad \mbox{and} \quad \lim_{n \to \infty}(B^{(n)}_0,A^{(n)}_0) = (B,A)
\]	
where $\| \cdot \|_{L^2}$ denotes the $L^2$-norm on $\mathbb{R}$ with respect to Lebesgue measure and $v_{r,0} \in L^{2}$ $(r \in \{b,s\})$ is  non-negative, bounded and continuously differentiable.
\end{assumption}

%\hspace{1mm}

\begin{figure}
\begin{minipage}[H]{0.49\textwidth}
\centering
\includegraphics[width=\textwidth]{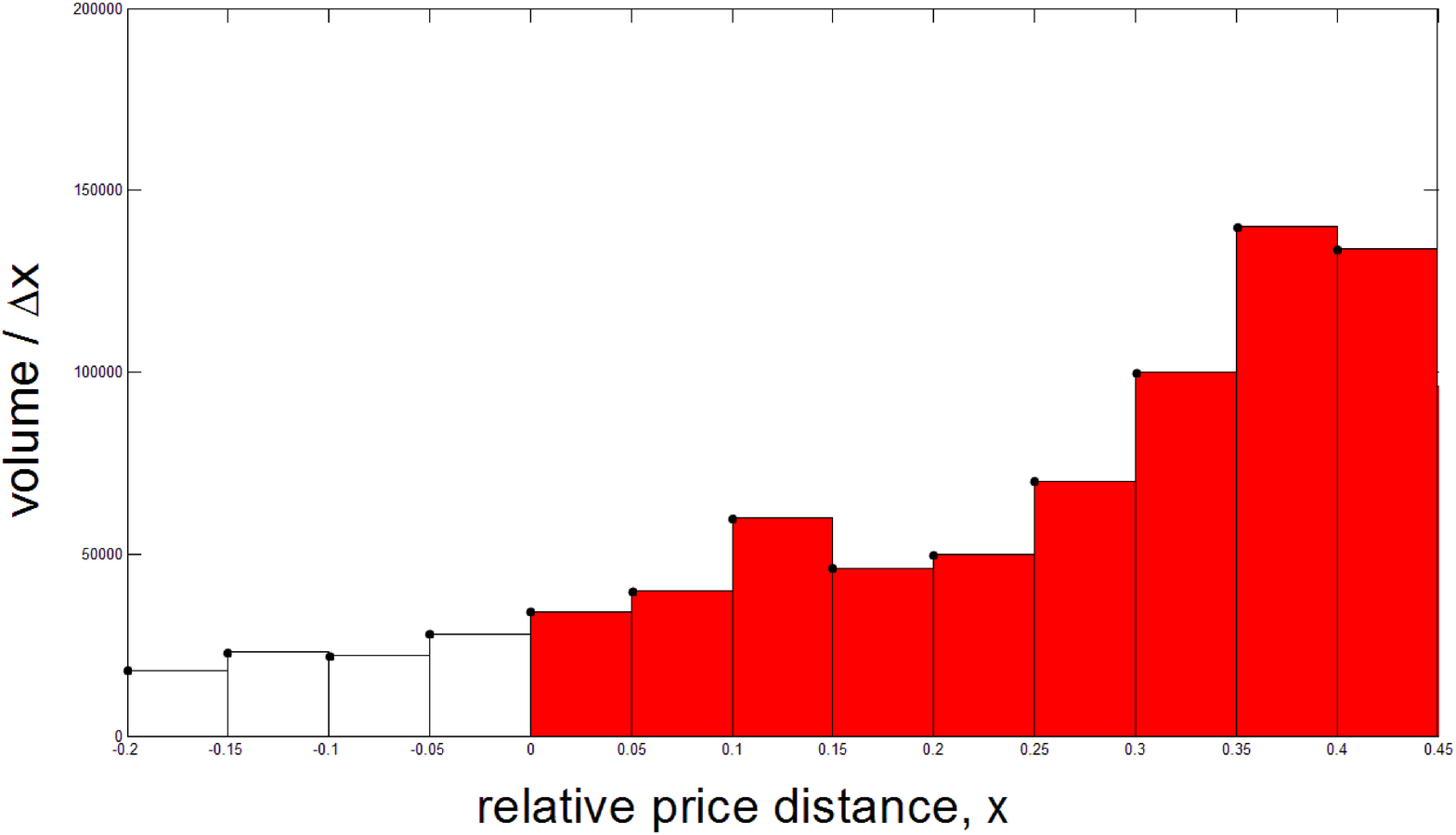}
\end{minipage}
\begin{minipage}[H]{0.49\textwidth}
\centering
	\includegraphics[width=\textwidth]{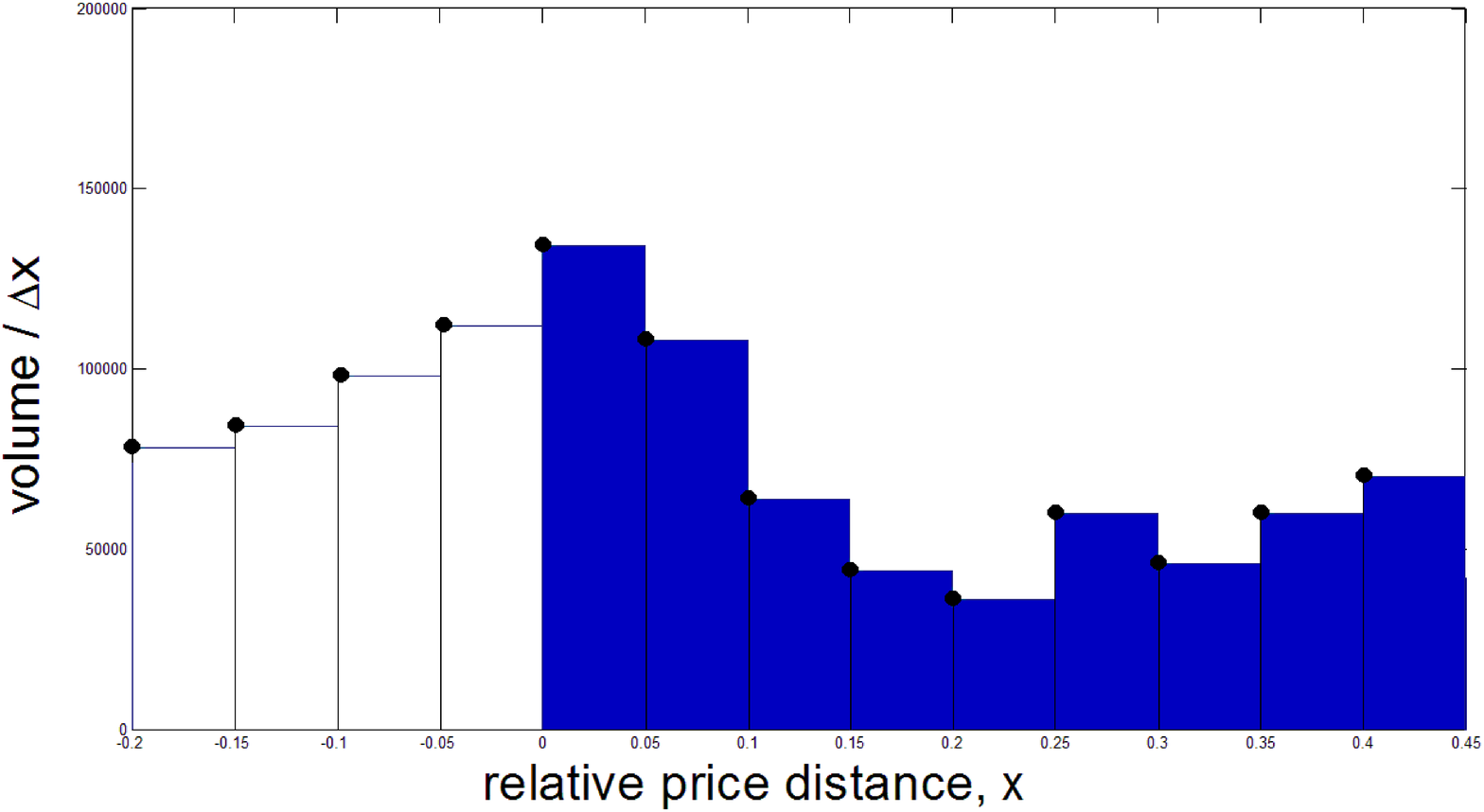}
\end{minipage}
\caption{Left [right]: the buy [sell] volume density function in relative price coordinates. The colored area under the densities is the standing volume; the white is the shadow book.}
\label{Abb:IllustrationVariablesBuySell}
\end{figure}

 The first condition on the volume density functions is intuitive. The second condition will become clear later; it will be used to bound the impact of market orders and limit orders placed into the spread on the state of the book. 
% Under mild smoothness conditions it is not difficult to show that both assumptions are satisfied if the initial volume density functions originate from common benchmark functions $v_{r,0} \in C^2$ in the sense that
%\begin{align}
%v_{r,0}^{(n)}(x):=\sum_{j}v_{r,0}^{(n),j}\mathbf{1}_{[x_{j}^{(n)},x_{j+1}^{(n)})}(x), \quad\text{where}\quad v_{r,0}^{(n),j}:=\frac{1}{\Delta x^{(n)}}\int_{x_{j}^{(n)}}^{x_{j+1}^{(n)}}v_{r,0}(x)\mathrm{d} x.
%\label{def:vnb}
%\end{align}
%One way of motivating this approach is the algorithm proposed in Schmidt et al. \cite{SchmidtBastianMulansky}. It states that one can construct a (not necessarily unique) smooth function $v_{r,0}$ from the initial standing volumes $\{v^{(1),j}_{r,0}\}_{j \in \mathbb{Z}}$ in the benchmark model $n=1$ in such a way that  $v^{(1),j}_{r,0} = \int_{x^{(1)}_j}^{x^{(1)}_{j+1}} v_{r,0}(x)\mathrm{d}x$ where we normalized $\Delta x^{(1)} = 1$.

%%%%%%%%%%%%%%%%%%%
%%%%%%%%%%%%%%%%%%%
%%%%%%%%%%%%%%%%%%%

\subsubsection{Event types}

 There are eight events - labeled {\bf A, ..., H} - that change the state of the book. The events {\bf A, ..., D} affect the buy side of the book:
\begin{align*}
	\textbf{A}&:=\{\text{market sell order}\}& \textbf{B}& :=\{\text{buy limit order placed in the spread}\}\\
	\textbf{C}&:=\{\text{cancelation of buy volume}\}& \textbf{D}& :=\{\text{buy limit order not placed in spread}\}
\end{align*}
The remaining four events affect the sell side of the book:
\begin{align*}
	\textbf{E}&:=\{\text{market buy order}\}& \textbf{F}& :=\{\text{sell limit order placed in the spread}\}\\
	\textbf{G}&:=\{\text{cancelation of sell volume}\}& \textbf{H}& :=\{\text{sell limit order not placed  in the spread}\}. \\
\end{align*}
We will describe the state dynamics of the $n$:th model by a stochastic process $\{S^{(n)}_k\}_{k \in {\mathbb N}}$ that takes values in the Hilbert space
\[
	E:=\mathbb{R}\times\mathbb{R}\times L^{2} \times L^{2}.
\]
The first two components of the vector $S^{(n)}_k$ stand for the best bid and ask price after $k$ events; the third and fourth component refer to the buy and sell volume density functions relative to the best bid and ask price, respectively (visible and shadow book). We define a norm on
$E$ by
\begin{equation}
\|\alpha\|_{E}:=|\alpha_{1}|+|\alpha_{2}|+\|\alpha_{3} \|_{L^{2}}+\|\alpha_{4}\|_{L^{2}}, \quad\alpha=(\alpha_{1},\alpha_{2},\alpha_{3},\alpha_{4})\in E.
\label{def:prodNorm}
\end{equation}
In the sequel we specify how different events change the state of the book and how order arrival times and sizes scale with the parameter $n \in {\mathbb N}$.

%%%%%%%%%%%%%%%%%%%%%%%%%%%%%%%%%%%
%%%%%%%%%%%%%%%%%%%%%%%%%%%%%%%%%%%
%%%%%%%%%%%%%%%%%%%%%%%%%%%%%%%%%%%

\subsubsection {Active orders}  \label{section-active}

Market orders and placements of limit orders in the spread lead to price changes\footnote{A market order that does not lead to a price change can be viewed as a cancelation of standing volume at the best bid/ask price.}. With a slight abuse of terminology we refer to these order types as \textit{active orders}. 
For convenience, we assume that market orders match precisely against the standing volume at the best prices and that limit orders placed in the spread improve prices by one tick. The assumptions that market orders decrease (increase) the best bid (ask) price by one tick while limit orders placed in the spread decrease (increase) prices by the same amount have been made in the literature before and can be generalized without too much effort.\footnote{A market order whose size exceeds the standing volume at the top of the book and would hence move the price by more than one tick, is split by the exchange into a series of consecutively executed orders. The size of each such `suborder', except the last, equals the liquidity at the current best price. Thus, by definition, a single market order cannot move the price by more than on tick.} However, this would unnecessarily complicate the analysis that follows.

Thus, if the $k$:th event is a sell market order (Event \textbf{A}), then the relative buy volume density shifts one price tick to the left (the liquidity that stood $l$ tick into the book now stands $l-1$ ticks into the book), the best bid price decreases by one tick and the relative sell volume density and the best ask price remain unchanged. Since the relative volume density functions are defined on the whole real line, the transition operators
\begin{equation*} \label{def:OrigTranslOperator}
	T^{(n)}_{+}(v)(\cdot)=v(\cdot+\Delta x^{(n)}), \quad T^{(n)}_{-}(v)(\cdot)=v(\cdot-\Delta x^{(n)})
\end{equation*}
are well defined and one has that
\begin{equation*}
	v^{(n)}_{b,k+1}(\cdot)=T^{(n)}_{+}\left(v^{(n)}_{b,k}\right)(\cdot), \quad v^{(n)}_{s,k+1}(\cdot) = v^{(n)}_{s,k}(\cdot)
\end{equation*}
and
\begin{equation*}
B^{(n)}_{k+1}=B^{(n)}_{k}-\Delta x^{(n)}, \quad A^{(n)}_{k+1}=A^{(n)}_{k}.
\end{equation*}
The placement of orders into the spread will be modeled using the shadow book. If the $k$:th event is a buy limit order placement in the spread (Event \textbf{B}), the relative buy volume density shifts one price tick to the right (the liquidity that stood one tick above the best bid in the shadow book now stands at the top of the visible book), the best bid price increases by one tick and the relative buy volume density and the best ask price remain unchanged:
\begin{equation*}
	v^{(n)}_{b,k+1}(\cdot)=T^{(n)}_{-}\left(v^{(n)}_{b,k}\right)(\cdot), \quad v^{(n)}_{s,k+1}(\cdot) = v^{(n)}_{s,k}(\cdot)
\end{equation*}
and
\begin{equation*}
B^{(n)}_{k+1}=B^{(n)}_{k}+\Delta x^{(n)}, \quad A^{(n)}_{k+1}=A^{(n)}_{k}.
\end{equation*}

\begin{remark}
Notice that market order arrivals and limit order placements in the spread are ``inverse operations'': a market sell order arrival followed by a limit buy order placement in the spread (or vice versa) leaves the book unchanged. 
\end{remark}

%%%%%%%%%%%%%%%%%%%%%%%%%
%%%%%%%%%%%%%%%%%%%%%%%%%
%%%%%%%%%%%%%%%%%%%%%%%%%

\subsubsection{Passive orders} \label{section-passive}

Limit order placements outside the spread and cancelations of standing volume do not change prices. With the same minor abuse of terminology as before, we refer to these order types as \textit{passive orders}.
We assume that cancelations of buy volume (Event \textbf{C}) occur for random \textit{proportions} of the standing volume at random price levels while limit buy order placements outside the spread (Event \textbf{D}) occur for random \textit{volumes} at random price levels. The submission and cancelation price levels are chosen {\textsl relative to the best bid price}. To guarantee convergence of the dynamics of volumes to the solution of a PDE, volumes placed and proportions canceled are scaled in a particular way.

\begin{assumption} \label{Assumption-2}
For each $k \in \mathbb{N}$ there exist random variables $\omega^C_k, \omega^D_k$ taking values in $(0,1)$, respectively $[0,M]$ for some $M > 0$ and random variables  $\pi^C_k, \pi^D_k$ taking values in $[-M,M]$ such that, if the $k$:th event is a limit buy order cancelation/placement, then it occurs at the price level $x^{(n)}_{B^{(n)}_k - j}$ $(j \in \mathbb{Z})$ for which
	\[
		\pi^{C,D}_k \in [x^{(n)}_{j}, x^{(n)}_{j+1}). %~~ \mbox{ respectively } ~~ \pi^C_k \in [x^{n}_{j}, x^{(n)}_{j+1})
	\]
The volume canceled, respectively, placed is
\[
	\omega^C_k \cdot \Delta v^{(n)} \cdot v^{(n)}_{b,k}(\pi^C_{k}) ~~ \mbox{ respectively ~~ }
	\omega^D_k \cdot \Delta v^{(n)} .
\]
Here $v^{(n)}_{b,k}(\pi^C_{k})$ is the value of the volume density functions at the cancelation price level and $\Delta v^{(n)}$ is a scaling parameter that describes the impact of an individual limit order arrival (cancelation) on the state of the order book.
\end{assumption}
%
%
%\begin{figure}
%\begin{minipage}[H]{0.49\textwidth}
%\centering
%\includegraphics[width=\textwidth]{OrderbookAbs_3}
%\end{minipage}
%\begin{minipage}[H]{0.49\textwidth}
%\centering
%	\includegraphics[width=\textwidth]{OrderbookAbs_4}
%\end{minipage}
%\caption{Left: New state of the book after buy limit order placement in the spread at time $t+1$. Right: New state of the book at time $t+2$ if (instead) a market buy order arrived at time $t+2$.}
%\label{fig:shadow-book2}
%\end{figure}
%

 Some comments are in order. First, by construction, buy order placements and cancelations take place at random distances (positive or negative) from the best bid price. Since $|\pi^{C,D}_k| \leq M$ no orders are placed/canceled more than $M/\Delta x^{(n)}$ ticks into the book. Volume changes take place in the visible or the shadow book, depending on the sign of of $\pi^I_k$. If $\pi^I_k \geq 0$, then the visible book changes; if $\pi^I_k < 0$ the placement/cancelation takes place in the shadow book and the impact of the event on the state of the visible book (i.e. at price levels \textit{above} the best bid) will be felt only after one or several price increases.

%\hspace{1mm}

\begin{example}
Suppose that the liquidity standing one tick above the best bid price after $(k-1)$ events is $v_{b,k-1}^{(n),-1}$, that the $k$-th event is an order placement of size $\omega^D_k \Delta v^{(n)}$ in the (buy-side) shadow book one tick above the best bid, i.e. $\pi^{D}_k \in [-\Delta x^{(n)},0)$ and that the $(k+1)$st event is a buy limit order placement in the spread. Then $B^{(n)}_{k+1} = B^{(n)}_k + \Delta x^{(n)}$ and the liquidity standing at the top of the book at $B^{(n)}_{k+1}$ (size of the order placed) is $v_{b,k-1}^{(n),-1} + \omega^D_k \Delta v^{(n)}$. This is how the buy order previously placed in the shadow now becomes part of the visible book. The role of the shadow book is further illustrated by Figure \ref{fig:shadow-book2}.
\end{example}

Second, the liquidity available for buying $j \in \mathbb{N}$ ticks into the book after $k$ events is given by by the integral of the volume density function over the interval $[x^{(n)}_j,x^{(n)}_{j+1})$, i.e. equals $v^{(n)}_{b,k}(x^{(n)}_j) \cdot \Delta x^{(n)}$. When a cancelation occurs at this price level, then the new volume is $(\Delta x^{(n)} - \omega^C_{k} \cdot \Delta v^{(n)} ) \cdot v^{(n)}_{b,k}(x^{n}_j) = (1-\omega^C_k \frac{\Delta v^{(n)}}{\Delta x^{(n)}}) \cdot v^{(n)}_{b,k}(x^{n}_j) \cdot \Delta x^{(n)}$. Hence cancelations are proportional to standing volumes. On the level of the volume density functions, Assumption \ref{Assumption-2} implies that
%
%assumption implies that if the $k$:th event is a cancelation of buy orders, then\footnote{The liquidity available for trading $j$ ticks into the book after $k$ events is given by $v^{(n)}_{b,k}(x^{(n)}_j) \cdot \Delta x^{(n)}$, i.e. by the integral of the volume density function over the interval $[x^{(n)}_j,x^{(n)}_{j+1})$. When a cancelation occurs at this price level, then the new volume is $(\Delta x^{(n)} - \omega^C_{k} \cdot \Delta v^{(n)} ) \cdot v^{(n)}_{b,k}(x^{n}_j) = (1-\omega^C_k \frac{\Delta v^{(n)}}{\Delta x^{(n)}}) \cdot v^{(n)}_{b,k}(x^{n}_j) \cdot \Delta x^{(n)}$, hence cancelations are proportional to standing volumes. By contrast, volume placements are additive. This explains the scaling of the volume functions.}
%%
%On the level of the volume density functions the scaling parameter is $\Delta v^{(n)}/\Delta x^{(n)}$, {\sl not} $\Delta v^{(n)}/\sqrt{\Delta x^{(n)}}$. With our choice of scaling, the resulting impact on the book, as measured by the $L^2$-norm of the density functions, is then of the order $\Delta v^{(n)}/\sqrt{\Delta x^{(n)}}$. The reason we work on $L^1$ rather than $L^2$ is that our limit theorem is based on a law of large number for martingale difference arrays taking values in uniformly smooth Banach spaces. The Hilbert space $L^2$ is uniformly smooth; $L^1$ is not.}
\[
	v^{(n)}_{b,k+1}(\cdot) = v^{(n)}_{b,k}(\cdot) - \frac{\Delta v^{(n)}}{\Delta x^{(n)}}\cdot M_{k}^{(n),C}(\cdot) \cdot v^{(n)}_{b,k}(\cdot), \quad \mbox{where} \quad M_{k}^{(n),C}(x)=\omega^C_k \sum_{j=-\infty}^{\infty}\mathbf{1}_{\{\pi_{k}^{C}\in[x^{(n)}_j x^{(n)}_{j+1})\}}(x).
\]
Volume placements are additive. If the order a limit buy order, then the volume density function changes according to
\begin{equation*}
	v^{(n)}_{b,k+1}(\cdot)=v^{(n)}_{b,k}(\cdot)+\frac{\Delta v^{(n)}}{\Delta x^{(n)}} \cdot M_{k}^{(n),D}(\cdot), \quad \mbox{where} \quad
	M_{k}^{(n),D}(x):=\omega^D_k \sum_{j=-\infty}^{\infty}\mathbf{1}_{\{\pi_{k}^{D}\in[x^{(n)}_j, x^{(n)}_{j+1})\}}(x).
\end{equation*}
In either case, the bid/ask price and standing sell side volume of the book remain unchanged:
\[
	v^{(n)}_{s,k+1}(\cdot) = v^{(n)}_{s,k}(\cdot), \quad B^{(n)}_{k+1} = B^{(n)}_k, \quad A^{(n)}_{k+1} = A^{(n)}_k.
\]
Similar considerations apply to the sell side with respective random quantities $\omega^G_k, \omega^H_k$ and $\pi^G_k, \pi^H_k$.

\begin{figure}[h]
\begin{minipage}[h]{0.48\textwidth}
\centering
	\includegraphics[width=\textwidth]{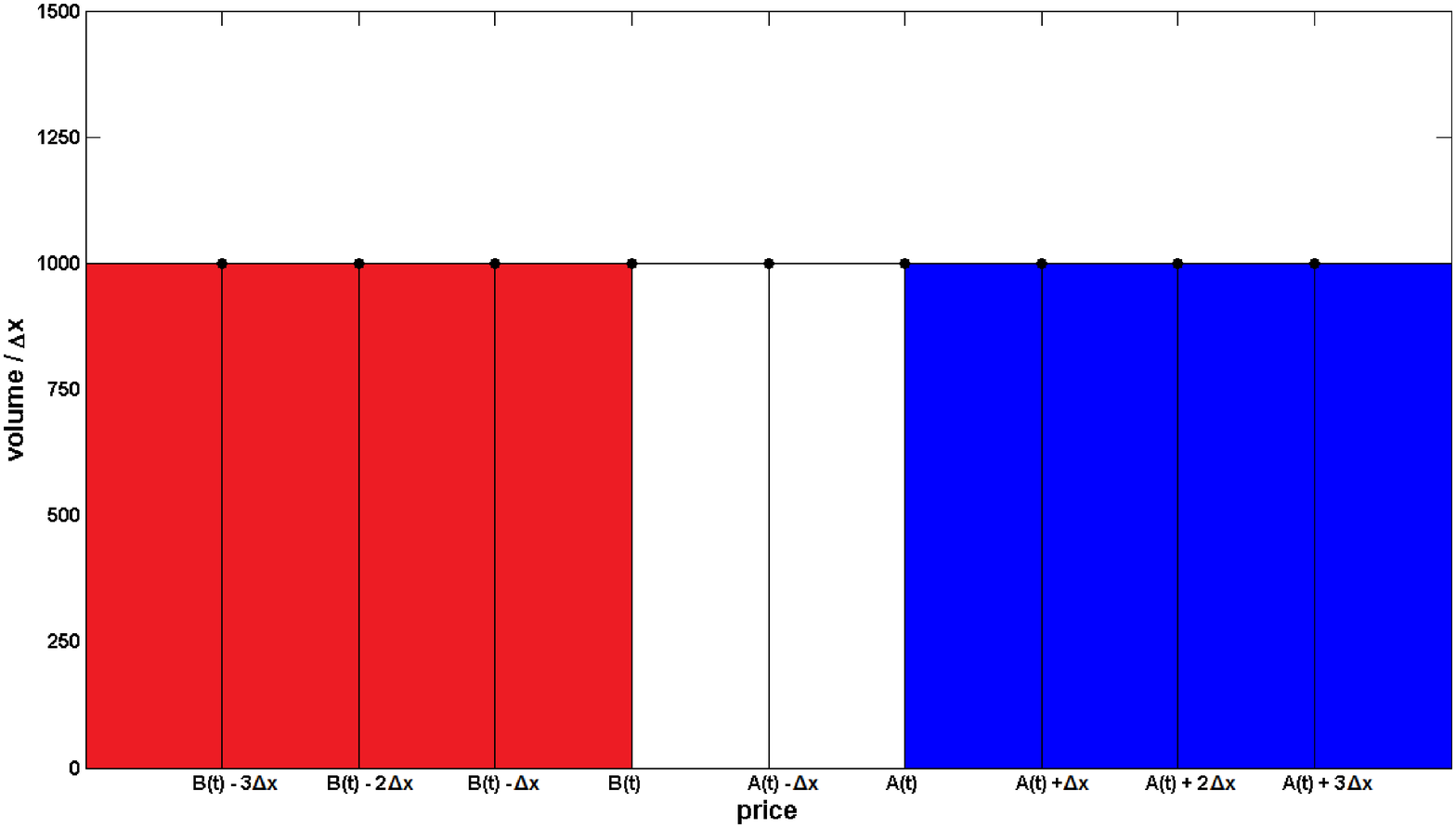}
%\caption{}
\end{minipage}
\begin{minipage}[h]{0.48\textwidth}
\centering
	\includegraphics[width=\textwidth]{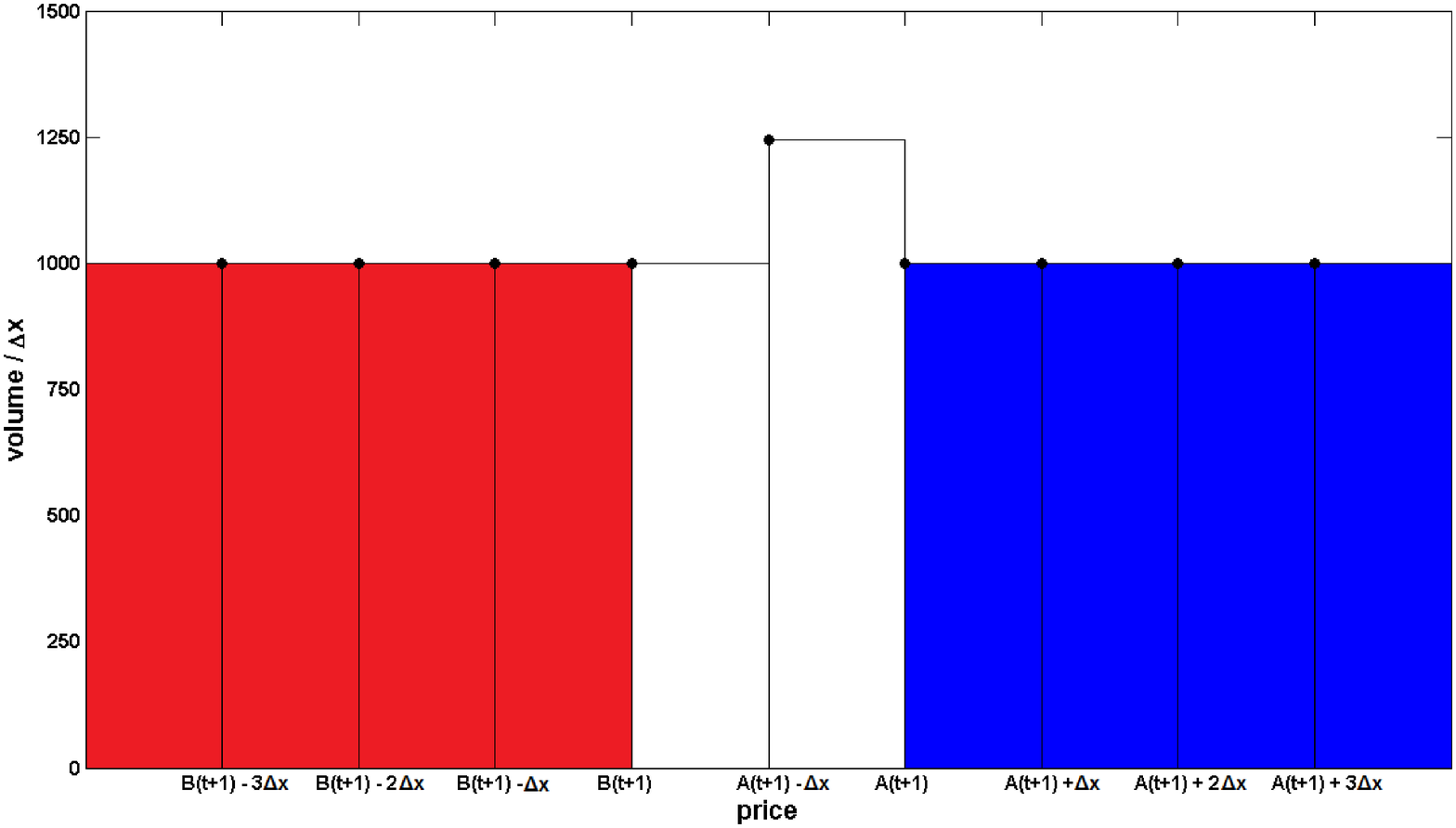}
%\caption{}
\end{minipage}

\begin{minipage}[h]{0.48\textwidth}
\centering
	\includegraphics[width=\textwidth]{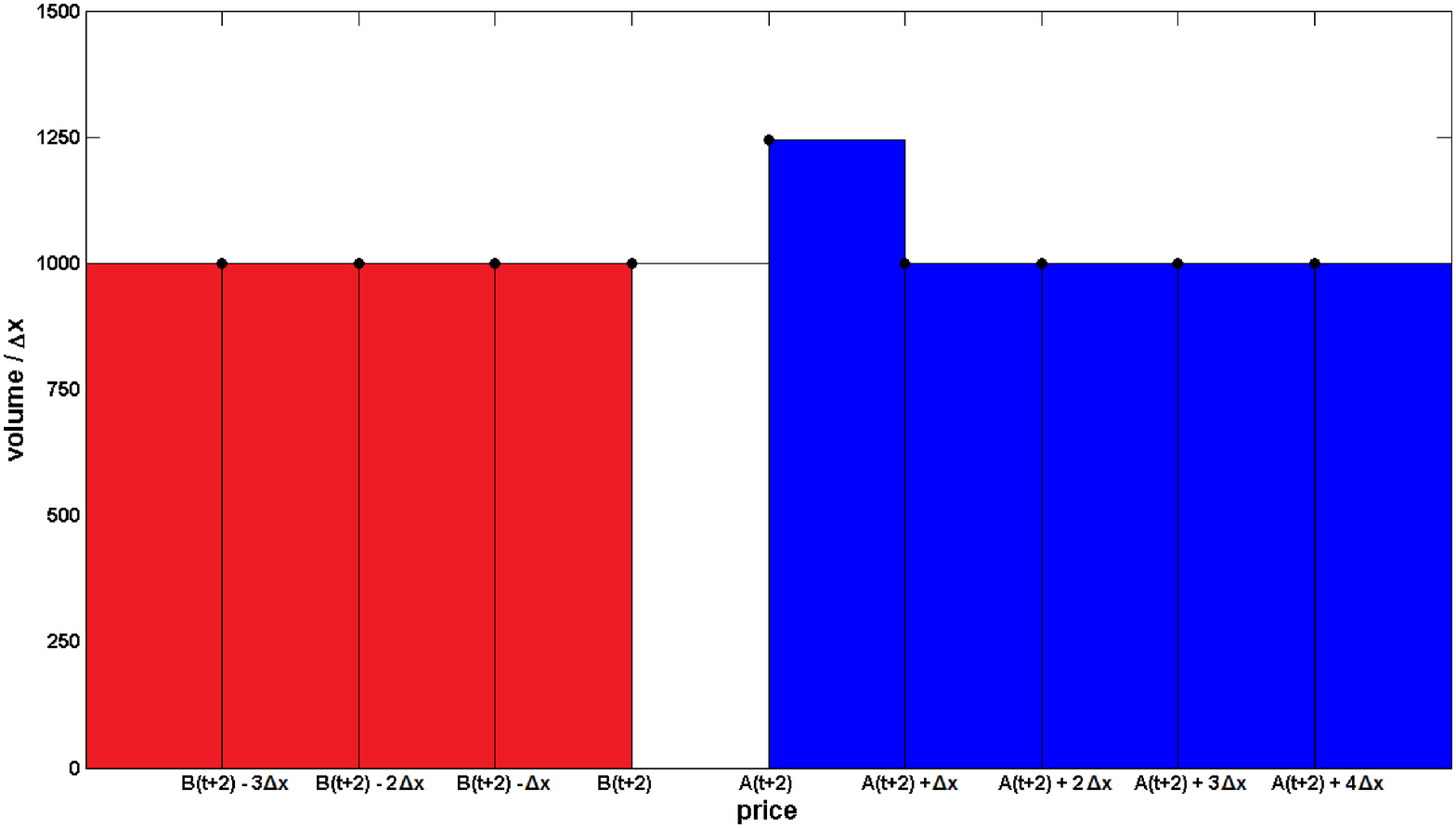}
%\caption{State at $t+2$ if a buy limit order in the spread arrives.}
\end{minipage}
\begin{minipage}[h]{0.48\textwidth}
\centering
	\includegraphics[width=\textwidth]{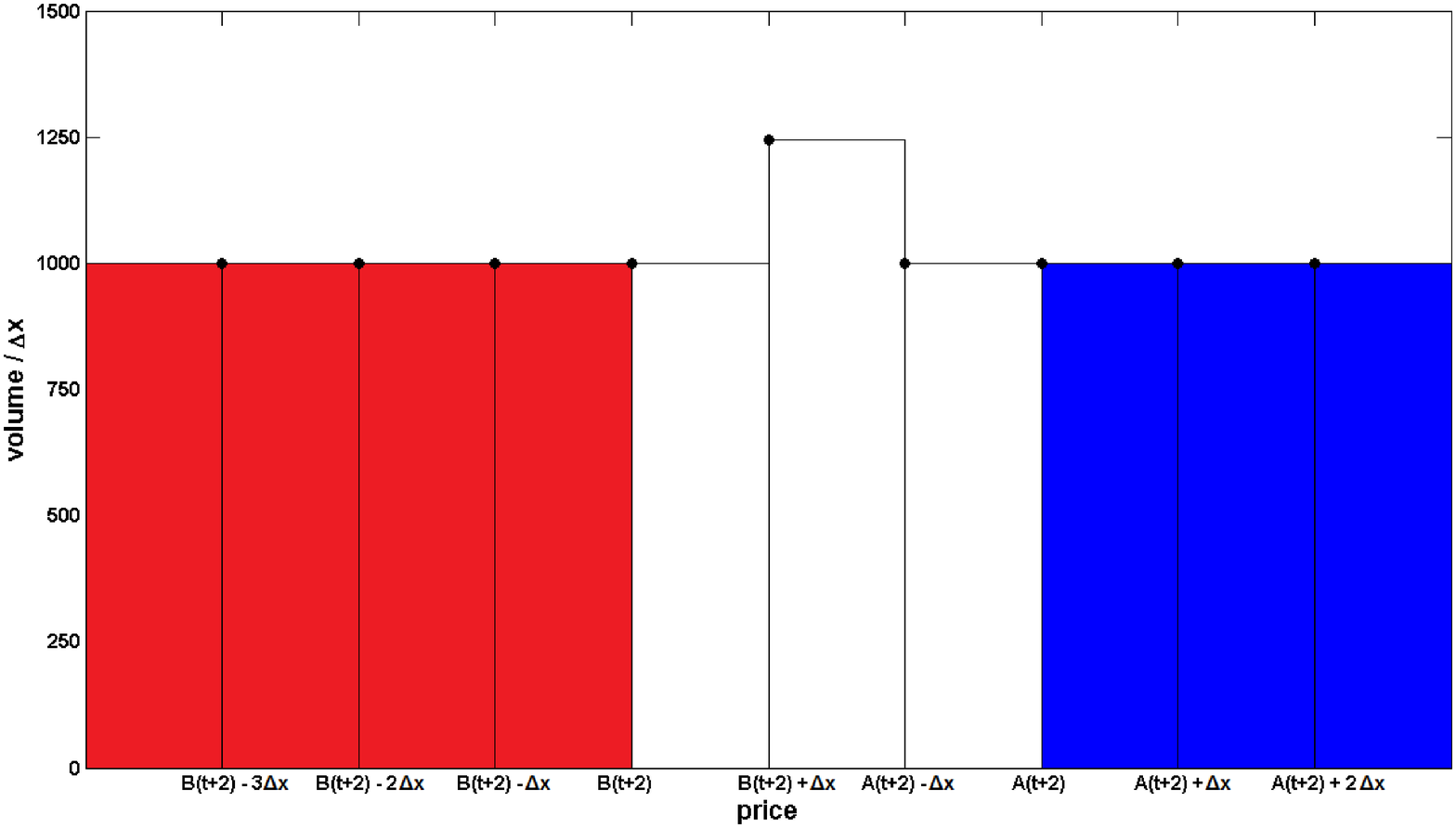}
%\caption{State at $t+2$ if (instead) a market buy order arrives}
\end{minipage}
\caption{Top left: Initial buy (red) and sell (blue) volume densities. The sell-side shadow book is uncolored. Top right: Second state after a sell-order has been placed in the shadow book at price level $A(t)-\Delta x^{(n)}$. Bottom left: Third state after a subsequent sell limit order has been placed in the spread. Bottom right: Third state if, instead of a subsequent sell limit order in the spread, a market buy order arrived.}
\label{fig:shadow-book2}
%\caption{Upper left: New state of the book after buy limit order placement in the spread at time $t+1$. Right: New state of the book at time $t+2$ if (instead) a market buy order arrived at time $t+2$.}
\end{figure}

{
\begin{assumption} \label{Assumption-3}
	For $I \in \{{\bf C},{\bf D}, {\bf G}, {\bf H}\}$ the sequences $\{\omega^I_k\}_{k \in {\mathbb N}}$ and $\{\pi^I_k\}_{k \in {\mathbb N}}$ are independent sequences of i.i.d. random variables. Moreover, the random variables $\pi^I_k$ have $C^2$-densities $f^I$ with compact support. %In particular, $f^I$ is bounded and belongs to $L^2$. 
\end{assumption}	

Lipschitz continuity of $f^I$ implies $\|f(\cdot \pm \Delta x^{(n)}) - f(\cdot)\|_\infty = \mathcal{O}(\Delta^{(n)} )$ as well as existence of a constant $K < \infty$ such that 
\begin{equation} \label{Lip-estimate}
\begin{split}
	\left| \mathbb{P}[\pi^I_k \in [x^{(n)}_j, x^{(n)}_{j+1})] -  \mathbb{P}[\pi^I_k \in [x^{(n)}_{j-1}, x^{(n)}_{j}) ] \right| 
	\leq & \, \int_{x^{(n)}_{j-1}}^{x^{(n)}_{j}} | f^I (y + \Delta x^{(n)}) - f^I (y) | dy \\
	\leq & \, K \left( \Delta x^{(n)}\right)^2. 
\end{split}
\end{equation}
Moreover, if we denote by  $f^{(n),I}$,  the  (scaled) expected changes in the heights of the density volume functions, due to cancelations and order placements, i.e.
\begin{eqnarray}
	f^{(n),I}(\cdot) & := & %\frac{1}{\Delta x^{(n)}} \mathbb{E}\left[ \omega^I_k \sum_{j=-\infty}^{\infty}\mathbf{1}_{\{\pi_{k}^{I}\in[x^{(n)}_j, x^{(n)}_{j+1})\}}(\cdot) \right] \\
	\sum_{j=-\infty}^{\infty} f^{(n),I}_j \mathbf{1}_{[x^{(n)}_j, x^{(n)}_{j+1})}(\cdot) \quad \mbox{with} \quad 
	f^{(n),I}_j := \frac{\mathbb{E}[\omega^I_1]}{\Delta x^{(n)}} \int_{x^{(n)}_j}^{x^{(n)}_{j+1}} f^I(x) \, \mathrm{d}x,
\end{eqnarray}
then 
\begin{equation}
	\|f^{n,I} - f^I \|_{\infty} = o(1) \quad \mbox{and} \quad
	\|T^{(n)}_\pm \circ f^{n,I} - f^{(n),I} \|_{\infty} = {\cal O}(\Delta x^{(n)}).
\label{eq:RegInitPriceDensIV}
\end{equation}
}

%%%%%%%%%%%%%%%%%%%%
%%%%%%%%%%%%%%%%%%%%
%%%%%%%%%%%%%%%%%%%%

\subsubsection{Event times}

The dynamics of event times is specified in terms of a sequence of inter-arrival times whose distributions may depend on prevailing best bid and ask prices.

\begin{assumption} \label{Assumption-4}
	Let $\{\varphi_k\}_{k \in {\mathbb N} }$ be a sequence of non-negative random variables that are  conditionally independent, given the current best bid and ask price:
\[
	\mathbb{P}[\varphi_k \leq t | S^{(n)}_k] = \mathbb{P}[\varphi_k \leq t | B^{(n)}_k, A^{(n)}_k].
\]
\end{assumption}

In the sequel we write $\varphi(A_k,B_k)$ for $\varphi_k$ to indicate the dependence of the distribution of $\varphi_k$ on the best bid and ask price. Similar notation will be applied to other random variables whenever convenient.

 In the $n$:th model, we scale the time by a factor $\Delta t^{(n)}$. More precisely, we assume that the dynamics of the event times in the $n$:th model is given by:
\begin{equation} \label{tau_n}
	\tau^{(n)}_{k+1} = \tau^{(n)}_k + {\cal C}^{(n)}_k(B^{(n)}_k,A^{(n)}_k), \quad \mbox{where} \quad {\cal C}^{(n)}_k(B^{(n)}_k,A^{(n)}_k) := \varphi(B^{(n)}_k, A^{(n)}_k) \cdot \Delta t^{(n)}.
\end{equation}

%%%%%%%%%%%%%%%%%%%%%%%%%%%%%
%%%%%%%%%%%%%%%%%%%%%%%%%%%%%
%%%%%%%%%%%%%%%%%%%%%%%%%%%%%

\subsubsection{Event types}

 Event types are described in terms of a sequence of random  \textit{event indicator variables} $\{\phi_k\}$ taking values in the set $\{{\bf A}, ..., {\bf H}\}$. We assume that the random variables
\[
	\phi_k = \phi(B^{(n)}_{k},A^{(n)}_{k}) \quad (k \in \mathbb{N}_{0})
\]
are conditionally independent, given the prevailing best bid and ask price and that their conditional probabilities
\[
	p^{(n)}_k=p^{(n),I}(B^{(n)}_{k}, A^{(n)}_{k}) := \mathbb{P}[ \phi_{k} = I |S^{(n)}_{k}]
\]
satisfy the following condition.

%\hspace{1mm}

\begin{assumption} \label{Assumption-5}
There are bounded continuous functions with bounded gradients $p^I: \mathbb{R}^2 \to [0,1]$ and a scaling parameter $\Delta p^{(n)} \to 0$ such that
\begin{eqnarray*}
	p^{(n),I}(\cdot,\cdot) &=&  \Delta p^{(n)} \cdot p^{I}(\cdot,\cdot) \quad ~~~~~~~ \mbox{for} \quad I={\bf A}, {\bf B} , {\bf E}, {\bf F} \\
	p^{(n),I}(\cdot,\cdot) &=& (1-\Delta p^{(n)}) \cdot p^{I}(\cdot,\cdot) \quad \mbox{for} \quad I={\bf C}, {\bf D}, {\bf G}, {\bf H} \\
	p^{A} + p^{B}  + p^{E}  + p^{F} & = & 1 \\
	p^{C} + p^{D}  + p^{G}  + p^{H} & = & 1
\end{eqnarray*}
\end{assumption}

%\hspace{1mm}

\begin{remark}
The preceding assumption implies that an event is an active order with probability $\Delta  p^{(n)}$ and a passive order with probability $1-\Delta  p^{(n)}$, independently of the state of the book. Conditioned on the order being active or passive, it is of type $I$ with a probability $p^I(\cdot,\cdot)$ that depends on the current best bid and ask prices. We allow the above probabilities to be zero in order to account for the fact that no price improvements can take place when $B^{(n)}_k = A^{(n)}_k$ and to avoid depletion of the order book.\footnote{For simplicity we assumed that the the initial volume density functions vanish outside a compact price interval. Hence there is a positive probability of depletion unless one assumes that no further buy/sell side price improvements take place if the distance of the current best bid/ask price from the initial state exceeds some threshold}.
\end{remark}

The expected impact of each active order on the state of the book will be of order $\Delta x^{(n)}$; that of a passive order of order $\Delta v^{(n)}$. Because active orders arrive at a rate that is $\Delta p^{(n)}$-times slower than that of passive orders, the relative average impact of active to passive orders on the state of the book will of the order $\frac{ \Delta  p^{(n)}  \Delta x^{(n)} }{\Delta v^{(n)} }.$ Our scaling limit requires to equilibrate the impact of active and passive orders. In order to guarantee that there will be no fluctuations in the standing volumes in the limit as $n \to \infty$ we also need a minimum relative frequency of passive order arrivals. This motivates the following assumption.

\begin{assumption} \label{assumption-scaling}
The scaling constants $\Delta p^{(n)}$, $\Delta x^{(n)}$, $\Delta v^{(n)}$ and $\Delta t^{(n)}$ are such that:
\[
	\lim_{n \to \infty} \frac{ \Delta x^{(n)} \cdot \Delta p^{(n)}}{ \Delta v^{(n)}} = c_0, \quad
	\lim_{n \to \infty} \frac{\Delta v^{(n)}}{\Delta t^{(n)}} = c_1, \quad \mbox{and} \quad
	\frac{\Delta p^{(n)}}{\Big( \Delta t^{(n)} \Big)^\alpha } = {\cal O}(1)
	%\quad \mbox{and} \quad \Delta p^{(n)} = {\cal O}(n^{-2-\epsilon})
\]
for some $\alpha \in \left( \frac{1}{2}, 1\right)$ and constants $c_0,c_1 >0$.\footnote{For the results that follow, we will assume that $\frac{ \Delta x^{(n)} \cdot \Delta p^{(n)}}{ \Delta v^{(n)}} = 1$ and $\frac{\Delta v^{(n)}}{\Delta t^{(n)}} = 1$ as $n \to \infty$. Any other constant would require further constants in the limiting dynamics.}
\end{assumption}

%%%%%%%%%%%%%%%%%%%%%%%%%%%%%
%%%%%%%%%%%%%%%%%%%%%%%%%%%%%
%%%%%%%%%%%%%%%%%%%%%%%%%%%%%

\subsubsection{Active order times}

 The previous two assumptions introduce two different time scales for order arrivals: a fast time scale for passive order arrivals, and a comparably slow time scale for active order arrivals. Inter-arrival times between passive orders are of the order $\Delta t^{(n)}$ while inter-arrival times between active orders are of the order $\Delta x^{(n)}$. In order to see this, let us denote by $\sigma^{(n)}_k$ the arrival time of the $k$:th active order. The \textit{number} $r^{(n)}_{k+1}$ of events one needs to wait until the $(k+1)$-st active order arrival can be viewed as the first success time in a series of Bernoulli experiments with success probability $\Delta p^{(n)}$ and expected value $\frac{1}{\Delta p^{(n)}}$.
%\[
%	q^{(n)}_k := p^{(n)} %\quad \mbox{where} \quad q^{*}_k := p^{*,A}_k + p^{*,B}_k + p^{*,E}_k + p^{*,F}_k;
%\]
%and expected value $\frac{1}{q^{(n)}_k}$.
The $(k+1)$-st active order arrives at time
\[
	\sigma^{(n)}_{k+1} = \sigma^{(n)}_{k} + \zeta^{(n)}_k \cdot \Delta x^{(n)}
\]
where
\[
	\zeta^{(n)}_k := \sum_{l=\sigma^{(n)}_k}^{r^{(n)}_{k+1}-1} \varphi_l \cdot \Delta p^{(n)}.
\]
Since the random variables $\varphi_{\sigma^{(n)}_k+1}, ..., \varphi_{r^{(n)}_k-1}$ are conditionally independent and identically distributed, $\{r^{(n)}_k\}$ and $\{\varphi_k\}$ are independent sequences. Because $\mathbb{E}[r^{(n)}_k] = \frac{1}{\Delta p^{(n)}}$, the conditional expected value $m(B^{(n)}, A^{(n)})$ of $\zeta^{(n)}_k$, given the prevailing bid and ask prices is independent of $n \in \mathbb{N}$. We assume that the mapping $m(\cdot,\cdot)$ is Lipschitz continuous.

\begin{assumption} \label{Assumption-7}
The conditional expected value $m(B,A)$ of $\zeta^{(n)}_k$ depends in a Lipschitz continuous manner on the prevailing pair of bid and ask prices $(B,A)$.
\end{assumption}

\subsubsection{State dynamics}

 We are now ready to describe the full dynamics of the state sequence. To this end, we put
\[
	S^{(n)}_k = \left( B^{(n)}_k, A^{(n)}_k, v^{(n)}_{b,k}, v^{(n)}_{s,k} \right)
\]
In terms of the indicator function $\mathbf{1}_{k}\left(S^{(n)}_{k}\right) := \left( \mathbf{1}_{{\bf A}}\left(\phi_{k}(B_{k}^{(n)},A_{k}^{(n)})\right), \ldots ,\mathbf{1}_{{\bf H}}\left(\phi_{k}(B_{k}^{(n)},A_{k}^{(n)}) \right)  \right)'$ the dynamics of the state sequence $\{S^{(n)}_k\}$ is of the form $$S^{(n)}_{k+1} = S^{(n)}_k + {\cal D}^{(n)}_k(S^{(n)}_k) $$
if we define the random operator ${\cal D}^{(n)}_k: E \rightarrow E$ by
\begin{equation}\label{D}	
	\quad {\cal D}^{(n)}_k(S^{(n)}_k) := \mathbb{M}^{(n)}_{k}(S^{(n)}_k)\cdot\mathbf{1}_{k}(S^{(n)}_{k})
\end{equation}
where the matrix $\mathbb{M}^{(n)}_{k}(S^{(n)}_k)$ equals
\begin{scriptsize}
\begin{alignat*}{2}
&
%\mathbb{M}_{k}(S_k):=
\left(
  \begin{array}{cccccccc}
    -\Delta x^{(n)} & \Delta x^{(n)} & 0 & 0 & 0 & 0 & 0 & 0 \\
    &&&&&&&\\
    0 & 0 & 0 & 0 & \Delta x^{(n)} & -\Delta x^{(n)} & 0 & 0 \\
    &&&&&&&\\
    M^{(n),A}_{k} & M^{(n),B}_{k} & - \frac{\Delta v^{(n)}}{{\Delta x^{(n)}}} M^{(n),C}_{k} \cdot v^{(n)}_{b,k} & \frac{\Delta v^{(n)}}{{\Delta x^{(n)}}} M_{k}^{(n),D} & 0 & 0 & 0 & 0 \\
    &&&&&&&\\
    0 & 0 & 0 & 0 & M^{(n),E}_{k} &  M^{(n),F}_{k} & - \frac{\Delta v^{(n)}}{{\Delta x^{(n)}}} M^{(n),G}_{k} \cdot v^{(n)}_{s,k}  & \frac{\Delta v^{(n)}}{{\Delta x^{(n)}}} M_{k}^{(n),H}
  \end{array}
\right).
\end{alignat*}
\end{scriptsize}
Here, the entries referring to shifts in the volume density functions, due to best bid and ask price changes, are given by
\begin{alignat}{2}
M_{k}^{(n),A} :=T^{(n)}_{+}\left(v^{(n)}_{b,k}\right)-v^{(n)}_{b,k},&\qquad M_{k}^{(n),E}:=T^{(n)}_{+}\left(v^{(n)}_{s,k}\right)-v^{(n)}_{s,k}\nonumber\\
M_{k}^{(n), B} :=T^{(n)}_{-}\left(v^{(n)}_{b,k}\right)-v^{(n)}_{b,k},&\qquad M_{k}^{(n),F} :=T^{(n)}_{-}\left(v^{(n)}_{s,k}\right)-v^{(n)}_{s,k}
\end{alignat}
%\end{sidewaystable}
and the entries referring the volume changes, due to placement and cancelation of volume, are given by
\begin{alignat}{2}
M^{(n),I}_{v,k}(x)&:=\omega_{k}^{I}\sum_{j=-\infty}^{\infty}   \mathbf{1}_{\{\pi_{k}^{I}\in[x^{(n)}_j,x^{(n)}_{j+1})\}}(x) & \text{ for events
I={\bf C}, {\bf D}, {\bf G}, {\bf H}}.
%M^{(n),I}_{v,k}(x)&:=\omega_{k}^{I}\sum_{j=-\infty}^{\infty}\mathbf{1}_{\{\pi_{k}^{I}\in[x^{(n)}_j,x^{(n)}_{j+1})\}}(x) & \text{ for events $I=D,H$}.
\label{def:OrigMnI}
\end{alignat}
Observing the dynamics in continuous time, we define
\begin{equation}
S^{(n)}(t):=S^{(n)}_{k}  \quad \mbox{and} \quad \tau^{(n)}(t):=\tau^{(n)}_{k}
\quad \text{for}\quad t\in[\tau^{(n)}_{k},\tau^{(n)}_{k+1}).
\label{def:OrigContState}
\end{equation}

%\subsubsection{Concluding comments}

\vspace{5mm}

\begin{remark} Overall, the state and time dynamics of our models are driven by the random sequences $\{\varphi_k\}$ (event times), $\{\phi_k\}$ (event types), $\{\pi^I_k\}$ (placement/cancelation price levels) and $\omega^I_k$ (placed/canceled orders). The joint dynamics of all models can be defined in terms of suitable independent families
\[
	\kappa_k := \left\{ (\varphi_k(B,A), \phi_k(B,A))_{(A,B) \in \mathbb{R}^2}, (\pi^I_k, \omega^I_k)_{I=A, ..., H} \right \}  \quad (k\in\mathbb{N}_{0})
\]
of independent random variables. In particular, the process $\{(S^{(n)}(t), \tau^{(n)}(t))\}_{t \in [0,T]}$ $(n \in \mathbb{N})$ is adapted to the filtration
\[
	{\cal F}^{(n)}_k := \sigma \left( \kappa_s :0 \leq s \leq \lfloor\frac{k}{\Delta t^{(n)}} \rfloor\right).
\]
\end{remark}

%%%%%%%%%%%%%%%%%%%%%%%%%%%%
%%%%%%%%%%%%%%%%%%%%%%%%%%%%
%%%%%%%%%%%%%%%%%%%%%%%%%%%%

\subsection{The main result}

Our main result is Theorem \ref{MainResult}. It states that - with our choice of scaling - the order book dynamics can be described by a coupled ODE:PDE system when $n \to \infty$: the dynamics of the best bid and ask prices will be given in terms of an ODE, while the relative buy and sell volume densities will be given by the respective unique classical solution of a first order linear hyperbolic PDE with variable coefficients.

\begin{theorem}[Law of Large Numbers for LOBs] Let $\{S^{(n)}\}_{n\geq 1}$ be the sequence of continuous time processes defined in (\ref{def:OrigContState}) and suppose that Assumptions \ref{Assumption-1} and \ref{Assumption-2}-\ref{Assumption-7} hold. Then, for all $T>0$ there exists a deterministic process $s:[0,T] \to E$ such that
\begin{equation*}
	\lim_{n \to \infty} \sup_{t\in[0,T]}\|S^{(n)}(t)-s(t)\|_{E} = 0 \quad \text{in probability.}
\end{equation*}
The process $s$ is of the form $s(t)=\left(\begin{array}{c}
                   \gamma(t) \\
                   v(\cdot,t)
                 \end{array}\right),
$
where
$\gamma(t)=\left(\begin{array}{c}
                   b(t) \\
                   a(t)
                 \end{array}\right)$
is the vector of the best bid and ask price  at time $t \in [0,T]$ and $v(x,t)=\left(\begin{array}{c}
                   v_{b}(x,t) \\
                   v_{s}(x,t)
                 \end{array}\right)$
denotes the vector of buy and sell volume densities at $t \in [0,T]$ relative to the best bid and ask price.
In terms of the matrices
\begin{alignat}{2}
A(\cdot)&:=%\frac{\Delta x}{\Delta t}
				\left(\begin{array}{cc}
                          p^{A}(\cdot)-p^{B}(\cdot) & 0 \\
                          0 & p^{E}(\cdot)-p^{F}(\cdot) \\
                        \end{array}
                      \right), ~~ %\label{def:ACoeff}, ~~
                      %\nonumber\\
      B(\cdot,x):= %-\frac{\Delta v}{\Delta t}
      			\left(\begin{array}{cc}
                          -p^{C}(\cdot)f^{C}(x) & 0 \\
                          0 & -p^{G}(\cdot)f^{G}(x) \\
                        \end{array}
                      \right), \label{def:BCoeff}
\end{alignat}
the vector
\begin{alignat}{2}
                      c(\cdot,x)&:=
                      		%\frac{\Delta v}{\Delta t}
				\left(\begin{array}{c}
                          p^{D}(\cdot)f^{D}(x)\\
                          \\
                           p^{H}(\cdot)f^{H}(x) \\
                        \end{array}
                      \right), \label{def:cCoeff}
\end{alignat}
and the function $m(\cdot,\cdot)$ that specifies the expected waiting time between two consecutive active order arrivals, the function $\gamma$ is the unique solution to the 2-dimensional ODE system
\begin{equation} \label{ODE}
\left\{
\begin{array}{ll}
%\begin{split}
	\frac{\mathrm{d} \gamma(t)}{\mathrm{d} t} &= \frac{A\left(\gamma(t)\right)}{m\left(\gamma(t)\right)}
	\left(\begin{array}{c} 1 \\ 1 \end{array}\right), \quad t \in [0,T]\\
	\gamma(0) &= \left(\begin{array}{c} B_0 \\ A_0 \end{array} \right)
%\end{split}
\end{array} \right.
\end{equation}
and $(v_{b},v_{s})$ is the unique non-negative bounded classical solution of the PDE
\begin{equation}
\left\{
\begin{array}{ll}
%\begin{split}
	v_{t}(t,x)&=\frac{1}{m^{}(\gamma(t))}\Big{(}A\left(\gamma(t)\right)
                      v_{x}(t,x)+B(\gamma(t),x)v(t,x) + c\left(\gamma(t),x\right)\Big{)},\quad (t,x)\in [0,T] \times \mathbb{R} \\
      v(0,x)&=v_{0}(0,x),\quad ~~~~~~~~~~~~~~~~~~~~~~~~~~~~~~~~~~~~~~~~~~~~~~~~~~~~~~~~~~~~~ x\in\mathbb{R}
%\end{split}
\end{array} \right. .
\label{eq:MainTheoPDE}
\end{equation}
\label{MainResult}
\end{theorem}

 The analysis of the limiting dynamics can be simplified by separating the randomness on the level of order arrival times from that of order types as shown in the following section. Subsequently, we give an explicit solution to the limiting PDE.

%%%%%%%%%%%%%%%%%%%%%%%%%%%%%%%%%%%%%%%%%
%%%%%%%%%%%%%%%%%%%%%%%%%%%%%%%%%%%%%%%%%
%%%%%%%%%%%%%%%%%%%%%%%%%%%%%%%%%%%%%%%%%

\subsubsection{State and time separation}

 For the continuous-time process $S^{(n)}$ we write
\[
	S^{(n)}(t) = \left( S^{(n)}_\gamma(t), S^{(n)}_v(t) \right)
\]
where $S^{(n)}_\gamma \in \mathbb{R}^2$ describes the dynamics of bid and ask prices and $S^{(n)}_v(t) \in L^2({\mathbb R}) \times L^2({\mathbb R})$ describes the dynamics of the buy and sell volume density functions. According to the following proposition the process can be expressed as the composition of a \textit{state process} $\eta^{(n)}$ and a \textit{time process} $\mu^{(n)}$. The proof follows from straightforward modifications of arguments given in Anisimov \cite[p.108]{Anisimov95}.

\begin{proposition}[State and time separation] \label{prop-state-and-time}
The process $S^{(n)}$ can be expressed as the composition of a random state process
\[
	\eta^{(n)}(t) = \left( \eta^{(n)}_\gamma(t), \eta^{(n)}_v(t) \right)
\]	
and a random time process $\mu^{(n)}$ as
\begin{equation*}
	S^{(n)}(t)=\eta^{(n)}\left(\mu^{(n)}(t)- \Delta t^{(n)} \right).
\end{equation*}
The state and time process is given by
\begin{equation} \label{state-process}
	\eta^{(n)}(t):=S_{k}^{(n)} \quad \mbox{for} \quad t \in\left[t^{(n)}_k, t^{(n)}_{k+1} \right)
\end{equation}
where $t^{(n)}_k := k \cdot \Delta t^{(n)}$ and
\begin{equation} \label{time-process}
	y^{(n)}(u) :=\tau_{k}^{(n)} \quad \mbox{for} \quad u \in\left[\tau^{(n)}_k, \tau^{(n)}_{k+1} \right),
\end{equation}
respectively. The time-change $\mu^{(n)}$ is then defined in terms of $y^{(n)}$ as
\begin{equation*}
	\mu^{(n)}(t):=\inf \{u > 0: y^{(n)}(u)>t\}.
    	\label{def:mu_n}
\end{equation*}
\end{proposition}

\mbox{ }

  The advantage of the state and time separation is that the processes $\eta^{(n)}$ and $\mu^{(n)}$ can be analyzed separately. In fact, we will show convergence in probability
\begin{equation*}
\lim_{n \to \infty} \sup_{t\in[0,T]}\left\|\eta^{(n)}(t)-\eta(t)\right\|_{E} = 0 ~~ \mbox{and}  ~~ \lim_{n \to \infty} \sup_{t\in[0,T]}\left|\mu^{(n)}(t)-\mu(t) \right| = 0
\end{equation*}
to limiting processes $\eta(t)=(\eta_\gamma(t), \eta_v(t))$ and $\mu(t)$. Since the state sequence takes values in the Hilbert space $E$, the time change theorem as proved, e.g. Billingsley \cite[p. 151]{BillingsleyConvProbMeas}, then implies that
\begin{equation*}
	\lim_{n \to \infty} \sup_{t \in [0,T]} \| S^{(n)}(t) - \eta(\mu(t)) \|_E =0 \quad\text{in probability.}
\end{equation*}
In our model, bid and ask prices are sufficient statistics for the evolution of the order book. In particular, the limiting behavior of the sequences $\eta^{(n)}_\gamma$ and $\mu^{(n)}$ can be analyzed without reference to volumes. In Section \ref{chapter-prices} we prove the following proposition.

\begin{proposition} \label{convergence-prices}
Let $\hat{\gamma}$ be the unique solution to the ODE
\begin{alignat}{2}
\begin{cases}
\frac{\mathrm{d} \hat{\gamma}(t)}{\mathrm{d} t}=A(\hat{\gamma}(t))
\left(
  \begin{array}{c}
   1 \\
   1 \\
  \end{array}
\right),\quad t\in(0,T]\\
\\
\hat{\gamma}(0)=\left(
  \begin{array}{c}
    B_0 \\
    A_0 \\
  \end{array}
\right).
\label{eq:autonStatePrice}
\end{cases}.
\end{alignat}
Then
\[
	\lim_{n \to \infty} \sup_{0 \leq t \leq T} \left| \eta^{(n)}_\gamma(t) - \hat{\gamma}(t) \right| = 0 \quad \mbox{in probability.}
\]
Moreover, the sequence of processes $\mu^{(n)}$ satisfies
\[
	\lim_{n \to \infty} \sup_{0 \leq t \leq T} \left| \mu^{(n)}(t) - \mu(t) \right| = 0 \quad \mbox{in probability.} \qquad \mbox{where} \qquad
	\mu^{-1}(t) = \int_0^t m(\hat{\gamma}_u) du.
\]
\end{proposition}

%\subsubsection{Volumes}

\mbox{ }

 Once the limiting time-change process $\mu$ has been identified, what remains to finish the proof of Theorem \ref{MainResult}, is to establish convergence of the volume processes $\eta^{(n)}_v$ to their deterministic limit. This will be achieved in Section \ref{chapter-volumes} where we prove the following result.
\begin{proposition} \label{convergence-volumes}
Let $\widehat{u}$ be the unique classical solution of the PDE
\begin{equation}
\left\{
\begin{array}{ll}
%\begin{split}
	\widehat{u}_{t}(x,t)&= %\Big{(}
	A\left(\widehat{\gamma}(t)\right)
                      \widehat{u}_{x}(x,t)+B(x,\widehat{\gamma}(t))\widehat{u}(x,t) + c\left(x,\widehat{\gamma}(t)\right),\quad (x,t)\in\mathbb{R}\times[0,T] \\
      \widehat{u}(x,0)&=v_{0}(x), ~~~~~~~~~~~~~~~~~~~~~~~~~~~~~~~~~~~~~~~~~~~~~~~~~~~~\quad x\in\mathbb{R}
%\end{split}
\end{array} \right. .
\label{PDE2}
\end{equation}
Then
\[
	\lim_{n \to \infty} \sup_{0 \leq t \leq T} \left \| \eta^{(n)}_v(t;\cdot) - \widehat{u}(t,\cdot) \right \|_{L^2} = 0 \quad \mbox{in probability.}
\]
\end{proposition}

%%%%%%%%%%%%%%%%%%%%%%%%%%%%%%%
%%%%%%%%%%%%%%%%%%%%%%%%%%%%%%%
%%%%%%%%%%%%%%%%%%%%%%%%%%%%%%%

\subsubsection{Explicit solution}

 The PDE system (\ref{PDE2}) is coupled only through the limiting price dynamics. In particular, the equations for the buy and sell side can be solved independently. For the buy side PDE, we can write
\begin{equation}
\begin{cases}
\frac{\partial u_{b}}{\partial t}                =A_{b}(t)\frac{\partial u_{b}}{\partial x}+B_{b}(t,x)u_{b}+c_{b}(t,x)\\
u_{b}(0,x)=v_{b,0}(x)
\end{cases}
\label{eq:buySideSolvPDE}
\end{equation}
where $A_{b}(t)= p^{A}(\hat{\gamma}(t))-p^{B}(\hat{\gamma}(t))$, $B_{b}(t,x):=-p^{C}(\hat{\gamma}(t))f^{C}(x)$, $c_{b}(t,x):=p^{D}(\hat{\gamma}(t))f^{D}(x)$.
Using the method of characteristic curves, the PDE reduces to a family of ODE:s; see \cite[Chapter 3]{Evans} for details. The characteristic equations for our buy side PDE read
\begin{equation}
\begin{cases}
\frac{\mathrm{d} x}{\mathrm{d} \tau}=-A_{b}(\tau)\\
x(0)=\xi
\end{cases}
\text{and}\quad
\begin{cases}
\frac{\mathrm{d} \bar{u}_{b}}{\mathrm{d} \tau}                =B_{b}(\tau, x(\tau))\bar{u}_{b}+c_{b}(\tau, x(\tau))\\
\bar{u}_{b}(0,\xi)=v_{b,0}(\xi).
\end{cases}
\label{eq:buySideSolvODE}
\end{equation}
The solution of this system of linear ODE:s as a function of the state $\xi \in \mathbb{R}$ can be given in closed form:
\begin{eqnarray*}
	x(t,\xi) & = & \xi - \int_0^t A_b(t) \mathrm{d}t \\
	\bar{u}_b(t,\xi) &=& \exp\Big( \int_0^t B_b(u,x(u,\xi)) du \Big)  \Big( v_{b,0}(\xi) + \int_0^t \exp\Big( - \int_0^s B_b(u,x(u,\xi)) du \Big)
	c_b(s,x(s,\xi)) ds \Big).
\end{eqnarray*}
It describes the surface $\{ (t,\xi) : u_b(t,x(t,\xi)) = \bar{u}_b(t,\xi) \mbox{ given } u_b(0,\xi) = v_{0,b}(\xi) \}$.
%\[
%	\{ (t,\xi) : u_b(t,x(t,\xi)) = \bar{u}_b(t,\xi); u_b(0,\xi) = v_{0,b}(\xi) \}.
%\]
The solution to the buy side PDE can be recovered from the solution to the ODE-system through
\[
	u_b(t,y) = \bar{u}_b \Big(t, y + \int_0^t A_b(s) ds \Big).
\]
Due to our smoothness assumptions on the volume placement and cancelation functions it is not hard to verify that the solution is uniformly bounded with uniformly bounded first and second derivatives with respect to the time and space variable. Moreover, since the function $v_{b,0}$ vanishes outside a compact interval (Assumption \ref{Assumption-1}) and no orders are placed or canceled beyond a distance $M$ from the best bid/ask price (Assumption \ref{Assumption-2}), the function $u_b(t,\cdot)$ vanishes outside some compact interval $I(T)$ for all $t \in [0,T]$. Altogether, one has the following result.

\begin{proposition} \label{prop-PDE-solution}
Under the assumptions of Theorem \ref{MainResult}, the PDE (\ref{eq:MainTheoPDE}) has a unique solution $u_b$. The solution is uniformly bounded, with uniformly bounded first and second derivatives with respect to both variables and there exists an interval $I$ such that $u_b(t,x) = 0$ for all $t \in [0,T]$ and $x \notin I$.
\end{proposition}

%%%%%%%%%%%%%%%%%%%%%%%%%%%
%%%%%%%%%%%%%%%%%%%%%%%%%%%
%%%%%%%%%%%%%%%%%%%%%%%%%%%

\subsection{Examples and applications}

In this section, we discus examples where the limiting dynamics can be given in closed form as well as possible applications of our order book model to portfolio liquidation problems. The situation is most transparent when orders arrive according to a Poisson dynamics. Let us therefore assume that in the benchmark model $n=1$  orders arrive according to independent Poisson processes with smooth rate functions $\lambda^{A}(\cdot,\cdot), ..., \lambda^{H}(\cdot,\cdot)$. Then $\{\varphi_k\}$ is simply the minimum of the inter-arrival times of the Poisson processes. Standard arguments yield:
\[
	m(\cdot,\cdot) = \frac{1}{ \lambda^A(\cdot,\cdot) + \cdots + \lambda^H(\cdot,\cdot)}.
\]
To simplify the analysis we normalize the arrival rates so that $m(\cdot,\cdot) \equiv 1$. In the $n$:th model we scale the arrival rates of passive orders by $\Delta t^{(n)}$ and those of active orders by $\Delta x^{(n)}$. The arrival rates in the $n$:th model thus satisfy:
\[	
	\lambda^{(n),I}(\cdot,\cdot) = \frac{ \lambda^I(\cdot,\cdot)}{\Delta x^{(n)}}  \quad \mbox{for} \quad I = {\bf A}, {\bf B}, {\bf E}, {\bf F}.
\]	
and
\[	
	\lambda^{(n),I}(\cdot,\cdot) =  \frac{ \lambda^I(\cdot,\cdot)}{\Delta t^{(n)}} \quad \mbox{for} \quad I = {\bf C}, {\bf D}, {\bf G}, {\bf H}.
\]
Then the probability of the next event being an active order is of the order $\Delta p^{(n)} = \frac{\Delta t^{(n)}}{\Delta x^{(n)}}$ and:
\begin{eqnarray*}
			p^{I}(\cdot,\cdot) & = & \frac{\lambda^I(\cdot,\cdot)}{ \lambda^A(\cdot,\cdot) + \lambda^B(\cdot,\cdot) + \lambda^E(\cdot,\cdot) +\lambda^F(\cdot,\cdot)} \quad \mbox{for} \quad I= {\bf A}, {\bf B}, {\bf E}, {\bf F} \\
	p^{I}(\cdot,\cdot) & = & \frac{\lambda^I(\cdot,\cdot)}{\lambda^C(\cdot,\cdot) + \lambda^D(\cdot,\cdot) + \lambda^G(\cdot,\cdot) +\lambda^H(\cdot,\cdot)} \quad \mbox{for} \quad I= {\bf C}, {\bf D}, {\bf G}, {\bf H}. \\
%
%	m(\cdot,\cdot) &=& \frac{1}{ \lambda^C(\cdot,\cdot) + \cdots + \lambda^H(\cdot,\cdot)}
\end{eqnarray*}

\subsubsection{Price dynamics}

Within the Poisson framework we can readily set up the ODE for the limiting dynamics of the best bid and ask price process. Of course, the \textit{spread} $\mathfrak{S}(t) := a(t) - b(t)$ should be non-negative at all times. This can easily be achieved if we require $p^A(\cdot)=p^E(\cdot)=0$ for $\mathfrak{S}(t)=0$. For most application it would in fact be sufficient to assume that the price dynamics depend on the prevailing best bid and ask prices only through the spread. If the stationarity condition
\[
	\lambda^A - \lambda^B = \lambda^E-\lambda^F = 0
\]
holds, then the spread is constant, $\mathfrak{S}(t) \equiv \mathfrak{S}(0)$, and always positive if the we require the limiting initial spread to be positive.
The following is a simple example where the spread is initially positive, then becomes zero, and eventually opens again.

%\hspace{1mm}

\begin{example} \label{ex}
Assume that the order arrival rates take the form
\begin{equation*}
\left\{
\begin{array}{ll}
%\begin{split}
\lambda^{A}(b(t),a(t))&=\lambda^{B}(b(t),a(t))\\
\\
\lambda^{E}(b(t),a(t))-\lambda^{F}(b(t),a(t))&=2\sqrt{a(t)-B_{0}}, \quad t \in [0,T]
%\end{split}
\end{array} \right. .
\end{equation*}
Then the bid price is constant, $b(t) \equiv B_{0}$, and $a(t)=B_{0}+\Big(t-\frac{T}{2}\Big)^{2}$ for $t \in [0,T]$; see Figure 4.
\end{example}

\begin{figure} \label{fig-ex}
\centering
\includegraphics[width=10cm]{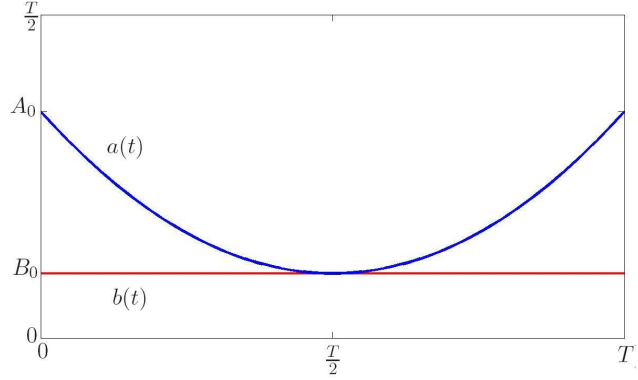}
\caption{Bid/ask price dynamics of Example \ref{ex}.}
\end{figure}

%\hspace{1mm}

The probability of a price increase, respectively, decrease can also be approximated in terms of the arrival rates. If the next event is an active order arrival on the buy side, then
\[
	\lim_{n \to \infty}\mathbb{P}[B^{(n)}(t) = B^n(t-) - \Delta x^{(n)}] = \frac{\lambda^A}{\lambda^A + \lambda^B}, \quad
	\lim_{n \to \infty}\mathbb{P}[B^{(n)}(t) = B^n(t-) + \Delta x^{(n)}] = \frac{\lambda^B}{\lambda^A + \lambda^B}.
\]
The limiting probabilities that the next price change is an increase/decrease of the best bid/ask price as well as unconditional probabilities of price changes can be computed analogously.  Clearly, these probabilities increase in the respective rates. A further possible application of our model includes estimations of the expected time-to-fill of a limit order in the original model. This question has been studied by, e.g. Lo et al \cite{Lo}. A simple approximation in our model framework is given by the first time the limiting price process hits the placement price level of the limit order.
%A refinement would be to consider the volume placed (over the original price tick interval) by calculating the time-to-fill as the time it takes until the placed limit volume plus the total limit volume ahead of the placed order becomes zero. The latter approximation would include price as well as the volume dynamics.

\subsubsection{Volume dynamics}

We now turn to the volume dynamics. When the limiting price dynamics is constant, then the explicit solution to the bid-side volume density PDE given in (\ref{eq:buySideSolvODE}) simplifies to
\begin{equation} \label{buy-side-constant}
	u_b(t,y) = e^{tB_b(y)} \left( v_{b,0}(y) + \frac{c_b(y)}{B_b(y)} \left[ 1 - e^{-t B_b(y)}\right]  \right)
\end{equation}
where we write $B_b(y)$ and $c_b(y)$ for $B_b(\gamma(0),y)$ and $c_b(\gamma(0),y)$, respectively. From this we see that the stationary solution is
\[
	u_b(t,y) = -\frac{c_b(y)}{B_b(y)} = \frac{p^D(\gamma(0))f^D(y)}{p^C(\gamma(0))f^C(y)} \equiv
	\frac{p^D}{p^C} \cdot \frac{f^D(y)}{f^C(y)}.
\]
In particular, the standing volume at the price level $y$ increases in the order arrival rates and the ``probabilities'' $f^D(y)$ with which an order is placed at that particular level when a placement occurs. Analogously, they decrease in $\lambda^C$ and $f^C(y)$. The following is a simple example where the spread is constant, but the prices are not. 
%
%If the the arrival dynamics depend only through the spread and if the spread is constant (but not necessarily the prices), then  (\ref{eq:buySideSolvODE}) gives us
%\begin{equation} \label{buy-side-constant}
%	u_b(t,y) = e^{tB_b(y)} \left( v_{b,0}(y) + \frac{c_b(y)}{B_b(y)} \left[ 1 - e^{-t B_b(y)}\right]  \right)
%\end{equation}
%
%

%\hspace{1mm}

\begin{example}
Let us assume that 
\begin{equation*}
\lambda^{A}(b,a)-\lambda^{B}(b,a)=\lambda^{E}(b,a)-\lambda^{F}(b,a)=1.
\end{equation*}
Then $b(t)=b_{0}+t$ and $a(t)=b(t)+\mathfrak{S}(0)$. Assume moreover that $b_0 = 0$ and that the buy side passive order arrival rates depend only on the spread: 
\begin{equation*}
\lambda^{C}(b,a)=\lambda^{D}(b,a)=a-b.
\end{equation*}
Thus, $p^{C,D}(\gamma(t)) \equiv p^{C,D}$. Let us further assume (ignoring the fact that our density functions need to be defined on compact intervals for simplicity) that
\[
	f^C \equiv \frac{1}{p^C} \quad \mbox{and} \quad f^D(x) = \frac{e^{-x}}{p^D}.
\]
Then, 
\[
	B_b(\gamma(t),y) \equiv -1 \quad \mbox{and} \quad c_b(\gamma(t),y) = e^{-y}.
\]
We compute from the general solution formula:
\[
	u_b(t,x) = e^{-t}u_b(0,x) + e^{-y}(1-e^{-t}). 
\]
\end{example}

\subsubsection{Portfolio liquidation}

We close this section with a brief discussion of how our LOB model could be used to obtain endogenous shape functions for portfolio liquidation models. In models of optimal portfolio liquidation under market impact the goal is to find trading strategies that unwind (or purchase) a large number $X > 0$ of shares within a pre-specified time window $[0,T]$ at minimal cost. It is typically assumed that prices are continuous (as in our limiting model), and that the expected distribution of the standing buy (or sell) side volume can be described in terms of a possibly time-dependent {\textsl shape function} 
\[
	f: [0,T] \times \mathbb{R} \to \mathbb{R}_+. 
\]	
The benchmark case studied in the seminal paper of Almgren and Chriss \cite{Almgren1999} corresponds to a block-shaped order book where $f(t,x) \equiv \delta$ for some $\delta > 0$; discrete-time liquidation problems with more general shape functions have been studied in, e.g. \cite{AlfonsiFruthSchied}. For a given shape function a sell order of size $E_t$ submitted at time $t \in [0,T]$ moves the best bid price by an amount $D_t$ defined through
\[
	E_t = \int_0^{D_t}f(t,x) \, dx.
\]
If we denote by $F^{-1}(t,\cdot)$ the inverse of the anti-derivative of the shape function and if we assume for simplicity that there is no permanent price impact and that the cost of trading is benchmarked against the mid quote, then the cost $c(t,E_t)$ of trading $E_t$ shares is half the spread plus the market impact costs (see, e.g. \cite{AlfonsiFruthSchied} and references therein for details):
\begin{equation} \label{cost}
	c(t,E_t) = E_t \left( \frac{1}{2} \mathfrak{S}(t) + F^{-1}(t,E_t) \right). 
\end{equation}
If orders can be submitted at discrete points in time $t_n$ $(n=1, ..., N)$, the resulting optimization problem is given by: 
\begin{equation}\label{discrete-opt} 
\begin{split}
	\min_{(E_{t_n})_{n=1}^N} \sum_{i=1}^N c(t_n,E_{t_n})
	\quad \mbox{s.t.} \quad \sum_{n=1}^N E_{t_n} = X.
\end{split}
\end{equation}
%In particular, an order of size $\eta_t \, dx$ moves the price by $\frac{\eta_t}{f(t,0)}\, dx$. The case studied in the seminal paper of Almgren and Chriss \cite{Almgren1999} corresponds to $f(t,x) \equiv \delta$ for some $\delta > 0$ that describes the ``height'' of the book. Liquidation problems in discrete time with general shape functions have been studied in, e.g. \cite{AlfonsiFruthSchied}. 
%\\[4pt]
%In the benchmark case without permanent price impact where a trader can submit block trades $\Delta \eta_0$ and $\Delta \eta_T$ at the initial and the terminal time and trades absolutely continuously at rates $\eta_t$ $(0<t<T)$ in between, the resulting optimization problem is of the form
%\begin{equation} \label{portfolio-liquidation}
%\begin{split}
%	& \min_{(\eta_t), \Delta \eta_0, \Delta \eta_T} \left\{ \int_0^T \frac{\eta^2_t}{f(t,0)} \, \mathrm{d}t + F^{-1}(0,\Delta \eta_0) + F^{-1}(T,\Delta \eta_T) \right\} \\
%	& \mbox{s.t.} \quad \int_0^T\eta_t \, \mathrm{d}t + \Delta \eta_0 + \Delta \eta_T= X.
%\end{split}
%\end{equation}
%\\[4pt]
Our model can be used to derive shape functions from order arrival dynamics. Within the framework of Poisson arrivals, the buy side dynamics in the $n$:th model is equivalent in distribution to modeling passive order arrivals and cancelations at the price levels $\{i\cdot \Delta x^{(n)}\}_{i \in \mathbb{Z}}$ in terms of independent Poisson processes with respective rates
\begin{eqnarray*}
	\eta^{(n),C}_i(\cdot) &=& \lambda^{(n),C}(\cdot) \int_{i\cdot \Delta x^{(n)}}^{(i+1) \cdot \Delta x^{(n)}} f^C(x) \, \mathrm{d}x \quad (i \in \mathbb{Z}) \\
	\eta^{(n),D}_i(\cdot) &=& \lambda^{(n),D}(\cdot) \int_{i\cdot \Delta x^{(n)}}^{(i+1) \cdot \Delta x^{(n)}} f^D(x) \, \mathrm{d}x \quad (i \in \mathbb{Z}).
\end{eqnarray*}
Placement and cancelation rates at different price levels (relative to the best bid/ask price) and average volumes placed/cancelled can readily be estimated from flow data; the same applies to aggregated (over all price levels) arrival rates $(\lambda^{(n),C/D})$ as well as market order arrival and spread placement rates $(\lambda^{(n),A/B})$ and volumes in the spread. We refer to Massey et al. \cite{MasseyParkerWhitt} for methods and techniques for estimating parameters of (non-homogeneous) Poisson processes and to Hautsch \cite{Hautsch11,Hautsch13} and references therein for applications of Poisson and general point processes to order dynamics and high-frequency trading.

%\hspace{1mm}

\begin{remark}
A possible difficulty would be to calibrate the volumes placed in the spread. In our model a spread placement is the result of many small events taking place in the shadow book while in practice this is one single event. On the other hand, one could use a low-dimensional parametric model for the initial state of the book (e.g. block-shaped) and the density functions $f^I(x)$ for $I \in \{\bf{C}, \bf{D}\}$ and $x<0$ and try to calibrate the resulting spread placements to empirical data. Another possibility would be to work with the stationary solution; in this case the density functions are only  required on the positive half line (and could hence be estimated from flow data). This is what is often done indirectly in portfolio liquidation models when one assumes that the (benchmark) price is a martingale and can hence be treated as a constant in the optimization problem.  
\end{remark}

%\hspace{1mm}

Once the placement, cancelation and active order arrival rates $\eta^{(n)}_i, \lambda^{(n),C/D}$ and $\lambda^{(n),A/B}$ have been estimated (possibly as functions of the best bid/ask price), functions $f^C$ and $f^D$ that satisfy the above equations can be constructed, e.g. by applying a method outlined in Schmidt et al.  \cite{SchmidtBastianMulansky} and the limiting model can be set up.

%\hspace{1mm}

\begin{example} \label{ex1}
Let us assume that the tick size is $\Delta x^{(n)}$ and that we are given empirical arrival rates which satisfy
\[
	\frac{\eta^{(n),C}_i}{\eta^{(n),C}_{i+1}} = e^{\kappa^C \Delta x^{(n)}} \quad \mbox{and} \quad
	\frac{\eta^{(n),D}_i}{\eta^{(n),D}_{i+1}} = e^{\kappa^D \Delta x^{(n)}}	
\]
for all $i \in \mathbb{Z}$ and some constants $\kappa^{C,D} > 0$. Then, the functions $f^{C,D}$ are recovered as:
\[
	f^{C,D}(x) = e^{-\kappa^{C,D} x}.
\]
Assuming that $\kappa^D > \kappa^C$, the stationary solution is of the form $u_b(t,x) = \kappa_1 e^{-\kappa_2 x}$ for $\kappa_1,\kappa_2 > 0$ and 
\[
	\int_0^\infty u_b(t,x) \,dx = \frac{\kappa_1}{\kappa_2} =: \kappa.
\]	
\end{example}

%\hspace{1mm}

Let us now assume that functions $f^{C,D}$ have been constructed along with the rates $\lambda^A, ..., \lambda^D$. Plugging the dynamics of the buy side volume density function into the portfolio liquidation problem, the resulting optimization problem reads
\[
	\min_{(E_{t_n})} \sum_{n=1}^N E_{t_n} \left( \frac{1}{2} \mathfrak{S}(t_n) + U^{-1}_b(t_n,E_{t_n}) \right)
	\quad \mbox{s.t.} \quad \sum_{n=1}^N E_{t_n} = X
\]
where $U_b^{-1}(t;\cdot)$ denotes the inverse of the anti-derivative of the volume density function $u_b(t;\cdot)$. In the framework of Example \ref{ex1} we arrive at the logarithmic impact function  $U^{-1}_b(t,y) = -\frac{1}{\kappa_2} \ln(1-\frac{y}{\kappa})$. 

%\hspace{1mm}

\begin{remark}
If one assumes that continuous trading is possible and restricts the set of admissible trading strategies to absolutely continuous ones, then the liquidation problem can be stated in terms of the volumes $u_b(t,0)$ $(0 \leq t \leq T)$ at the top  of the book. If both absolutely continuous and block trading is allowed as in \cite{HN}, then the dynamics of the full book is again required.   
\end{remark}

\section{Convergence of bid/ask prices} \label{chapter-prices}
According to Proposition \ref{prop-state-and-time}, the process $S^{(n)}$ can be represented in terms of a composition of a state process $\eta^{(n)}$ that jumps at deterministic times $\{t^{(n)}_k\}$ and a time process $\mu^{(n)}$ that accounts for the random event arrival times. Prices change less frequently at times $\{\sigma^{(n)}_k\}$. This suggests to introduce a second time scale - which will be referred to as \textsl{active order time} - defined by
\[
	s^{(n)}_k := k \cdot \Delta x^{(n)}
\]
along which to scale the price process. In order to make this more precise, let us denote by $\mathcal{D}_{\gamma,k}^{(n)}$ the restriction of the operator ${\mathcal D}^{(n)}_k$ to the price component of the state sequence and put
\[
	\widehat{\mathcal{D}}_{\gamma,k}^{(n)} := \sum_{l=\sigma^{(n)}_{k-1} }^{\sigma^{(n)}_k-1} \mathcal{D}^{(n)}_{\gamma,l} =
	{\mathcal{D}}_{\gamma,\sigma^{(n)}_k}^{(n)}
\]
where the second equality follows from the fact that prices do not change between times $\sigma^{(n)}_k$ and $\sigma^{(n)}_{k+1}-1$.
Furthermore, we introduce the family of continuous time stochastic processes $\widehat{\eta}_{\gamma}^{(n)}$ defined by
\[
	\widehat{\eta}^{(n)}(t) := \widehat{\eta}^{(n)}_k \quad \mbox{for } t \in  [s^{(n)}_k, s^{(n)}_{k+1})
\]
%for $t \in [k \cdot \Delta x^{(n)}, (k+1)\cdot \Delta x^{(n)})$
where
\begin{alignat}{2}
\begin{cases}
\widehat{\eta}_{\gamma,k+1}^{(n)}
&:=\widehat{\eta}_{\gamma, k}^{(n)} + \widehat{\mathcal{D}}_{\gamma,k}^{(n)}(\widehat{\eta}_{k}^{(n)})
\\
\label{def:akbkRec}\\
\widehat{\eta}_{0}^{(n)}&=\left(
  \begin{array}{c}
    B^{(n)}_{0} \\
    A^{(n)}_{0} \\
  \end{array}
\right).
\end{cases}
\end{alignat}
The quantity $\widehat{\eta}_{\gamma,k}^{(n)}$ describes the state of the price process after the $k$:th price change. The following lemma shows that the process $\eta^{(n)}_\gamma$, evolving on the level of event time, and the process $\widehat{\eta}^{(n)}_\gamma$, evolving on the level of active order time, are indistinguishable in the limit when $n \to \infty$.

\begin{lemma} \label{lemma-hatD}
For any $T > 0$ and $\epsilon > 0$, it holds that
\[
	\lim_{n \to \infty} \mathbb{P} \left[ \sup_{0 \leq t \leq T} \left| \sum_{k=0}^{t/\Delta x^{(n)}} \widehat{\mathcal{D}}^{(n)}_{\gamma,k} -
	\sum_{k=0}^{t/\Delta t^{(n)}} \mathcal{D}^{(n)}_{\gamma,k}  \right| > \epsilon \right] = 0.
\]
\end{lemma}
\begin{proof}
By construction, the two sums $\sum_{k=0}^{\lfloor t/\Delta x^{(n)} \rfloor} \widehat{\mathcal{D}}^{(n)}_{\gamma,k}$ and $\sum_{k=0}^{\lfloor t/\Delta t^{(n)} \rfloor} \mathcal{D}^{(n)}_{\gamma,k}$ have the same expected value for any $t \in [0,T]$:
\[
	\sum_{k=0}^{\lfloor t/\Delta x^{(n)} \rfloor} \mathbb{E} \widehat{\mathcal{D}}^{(n)}_{\gamma,k} =
	\sum_{k=0}^{\lfloor t/\Delta t^{(n)} \rfloor} \mathbb{E} {\mathcal{D}}^{(n)}_{\gamma,k}.
\]
As a result, it is enough to prove that
\[
	\lim_{n \to \infty} \mathbb{P} \left[ \sup_{0 \leq t \leq T} \left| \sum_{k=0}^{\lfloor t/\Delta x^{(n)} \rfloor}
	\left\{ \widehat{\mathcal{D}}^{(n)}_{\gamma,k}  - \mathbb{E} \widehat{\mathcal{D}}^{(n)}_{\gamma,k} \right\} \right|
	> \frac{\epsilon}{2} \right] =
	\lim_{n \to \infty} \mathbb{P} \left[ \sup_{0 \leq t \leq T} \left| \sum_{k=0}^{\lfloor t/\Delta t^{(n)} \rfloor}
	\left\{ \mathcal{D}^{(n)}_{\gamma,k}  - \mathbb{E}{\mathcal{D}}^{(n)}_{\gamma,k} \right\} \right|
	 > \frac{\epsilon}{2} \right] = 0.
\]
The random variables
\[
	\frac{1}{\Delta x^{(n)}}\left\{ \widehat{\mathcal{D}}^{(n)}_{\gamma,k} - \mathbb{E} \widehat{\mathcal{D}}^{(n)}_{\gamma,k} \right\}, \quad k=0, ..., \lfloor T/\Delta x^{(n)} \rfloor, ~~ n \in \mathbb{N}
\]	
and
\[
	\frac{1}{\Delta t^{(n)}}\left\{ \mathcal{D}^{(n)}_{\gamma,k} - \mathbb{E} \mathcal{D}^{(n)}_{\gamma,k} \right\}, \quad k=0, ..., \lfloor T/\Delta t^{(n)} \rfloor, ~~ n \in \mathbb{N}
\]
form triangular martingale difference arrays in the sense of Definition \ref{def-TMDA} with respect to the filtrations $\{{\cal F}_{\sigma^{(n)}_k}\}_{k \in \mathbb{N}}$ and $\{{\cal F}_{k \in \mathbb{N}}\}_k$, respectively. A direct computation shows that they are uniformly $L^2$-bounded. Thus, it follows from Theorem \ref{thm-appendix} that for all $\alpha > \frac{1}{2}$:
\[
	\lim_{n \to \infty} \mathbb{P} \left[ \sup_{0 \leq m \leq \lfloor \frac{T}{\Delta x^{(n)}} \rfloor } \frac{1}{\Delta x^{(n)}} \left| \sum_{k=0}^{m}
	\left\{ \widehat{\mathcal{D}}^{(n)}_{\gamma,k}  - \mathbb{E}\widehat{\mathcal{D}}^{(n)}_{\gamma,k} \right\} \right| \geq
	\frac{\epsilon}{2} \left( \frac{T}{\Delta x^{(n)}} \right)^ \alpha \right]  = 0
\]
as well as
\[
	\lim_{n \to \infty} \mathbb{P} \left[ \sup_{0 \leq m \leq \lfloor \frac{T}{\Delta t^{(n)}} \rfloor } \frac{1}{\Delta t^{(n)}} \left| \sum_{k=0}^{m}
	\left\{ {\mathcal{D}}^{(n)}_{\gamma,k}  - \mathbb{E}{\mathcal{D}}^{(n)}_{\gamma,k} \right\} \right| \geq
	\frac{\epsilon}{2} \left( \frac{T}{\Delta t^{(n)}} \right)^ \alpha \right]  = 0.
\]
Choosing $\alpha \in (\frac{1}{2},1)$ and multiplying the inequalities in the above probabilities by $\Delta x^{(n)}$ and $\Delta t^{(n)}$, respectively, proves the assertion.
%\Halmos
\end{proof}

%\hspace{1mm}

 Let $\widehat{\gamma}$ be the solution to the ODE (\ref{eq:autonStatePrice}) and consider the discretisation $\widehat{\gamma}^{(n)}_k := \widehat{\gamma}(s^{(n)}_k)$. The next lemma shows that the sequence of expected price processes $\widetilde{\gamma}^{(n)}$ defined by
\[
	\widetilde{\gamma}^{(n)}(t) := \widetilde{\gamma}^{(n)}_k	\quad \mbox{for } t \in  [s^{(n)}_k, s^{(n)}_{k+1})
\]
%for $t \in [k \cdot \Delta x^{(n)}, (k+1) \cdot \Delta x^{(n)})$
where
\begin{alignat}{2}
\begin{cases}
\widetilde{\gamma}_{k+1}^{(n)}
&:=\widetilde{\gamma}_{k}^{(n)} + \mathbb{E}\left[ \widehat{\mathcal{D}}_{\gamma,k}^{(n)}(\widehat{\gamma}_{k}^{(n)}) \right]
\\
\label{def:akbkRec}\\
\widetilde{\gamma}_{0}^{(n)}&=\left(
  \begin{array}{c}
    B^{(n)}_{0} \\
    A^{(n)}_{0} \\
  \end{array}
\right)
\end{cases}
\end{alignat}
converges uniformly to $\widehat{\gamma}$ on compact time intervals. The proof is standard; we give it merely for completeness.

\begin{lemma}
For any $T>0$
\begin{equation*}
\sup_{t\in[0,T]}|\widetilde{\gamma}^{(n)}(t)-\widehat{\gamma}(t)| =\mathcal{O}(\Delta x^{(n)}).
\end{equation*}
\begin{proof}
Active orders change prices by one tick. Moreover, by Assumption \ref{Assumption-5}, the conditional probabilities of an active order being a market buy/sell order or limit buy/sell order placement in the spread are independent of $n \in \mathbb{N}$. Hence
\begin{alignat}{2}
\mathbb{E}\left[ \widehat{\mathcal{D}}_{\gamma,k}^{(n)}(\widehat{\gamma}_{k}^{(n)}) \right]&= \Delta x^{(n)} \cdot \left(
  \begin{array}{c}
                          p^{A}(\widehat{\gamma}_{k}^{(n)}) - p^{B}(\widehat{\gamma}_{k}^{(n)}) \\
                           p^{E}(\widehat{\gamma}_{k}^{(n)}) - p^{F}(\widehat{\gamma}_{k}^{(n)})
                        \end{array}
                      \right)\nonumber\\
&=\Delta x^{(n)}\cdot\left\{ A(\widehat{\gamma}_{k}^{(n)})\left(
  \begin{array}{c}
   1 \\
   1 \\
  \end{array}
\right) \right\}
\label{eq:ExpectationsExplicit}
\end{alignat}
Thus, the sequence $\widehat{\gamma}^{(n)}$ defines a special case of the classical Euler scheme for the ODE (\ref{ODE})  and hence converges uniformly to its unique solution, see e.g. Hairer et al. \cite[Theorem 7.3]{HairerNorsettWannerI}, with rate $\Delta x^{(n)}$. %The $C^{2}$-property of the solutions follows by \cite[Theorem 20.10 on p.1104]{OlverAppliedMathematicsLectureNotes}.
%\Halmos
\end{proof}
\label{lem:ODEconv}
\end{lemma}

%\hspace{1mm}

 We are now ready to prove convergence in probability of the bid and ask prices.

\mbox{ }

\noindent \textsc{Proof of Proposition \ref{convergence-prices}}.

\begin{itemize}
\item[a)]
We first consider the convergence of the state process $\eta_{\gamma}^{(n)}$ and claim that
\begin{equation}
	\lim_{n \to \infty} \sup_{t\in[0,T]}|\eta_{\gamma}^{(n)}(t)-\widehat{\gamma}(t)| \rightarrow 0,\quad \text{in probability.}
\label{eq:etanToeta}
\end{equation}
In view of Lemma \ref{lemma-hatD}, we can write
%From (\ref{def:stateRecurr}) we have, by adding and subtracting the conditional expectation of the incremental price changes, that
\begin{alignat}{2}
\eta_{\gamma}^{(n)}(t) &=\widehat{\eta}_{\gamma}^{(n)} (t)\nonumber\\
&=\gamma^{(n)}(0)+\sum_{k=0}^{\lfloor t/\Delta x^{(n)} \rfloor} \widehat{\mathcal{D}}_{\gamma,k}^{(n)}\left(\widehat{\eta}_{\gamma,k}^{(n)}\right) \nonumber\\
&=\gamma^{(n)}(0)+\sum_{k=0}^{\lfloor t/\Delta x^{(n)}\rfloor}\mathbb{E}\left[ \widehat{\mathcal{D}}_{\gamma,k}^{(n)}(\widehat{\eta}_{\gamma,k}^{(n)}) \right]
+\sum_{k=0}^{\lfloor t/\Delta x^{(n)} \rfloor}\left( \widehat{\mathcal{D}}_{\gamma,k}^{(n)}(\widehat{\eta}_{\gamma,k}^{(n)})-\mathbb{E}\left[\widehat{\mathcal{D}}_{\gamma,k}^{(n)} (\widehat{\eta}_{\gamma,k}^{(n)}) \right]\right)
\nonumber
\label{eq:theoBegBuyArguments}
\end{alignat}
up to some random additive constant that vanishes almost surely uniformly in $t \in [0,T]$ as $n \to \infty$.
Adding and subtracting the sequence $\widetilde{\gamma}^{(n)}$ yields (again up to a vanishing additive constant):
\begin{eqnarray*}
	\left|\eta_{\gamma}^{(n)}(t)-\widehat{\gamma}(t)\right| & \leq & \left| \widetilde{\gamma}^{(n)}(t)-\widehat{\gamma}(t)\right| \\
	& & +  \left|\sum_{k=0}^{\lfloor t/\Delta x^{(n)} \rfloor}\mathbb{E}\left[\widehat{\mathcal{D}}_{\gamma,k}^{(n)}(\widehat{\eta}_{\gamma,k}^{(n)}) \right]-\widetilde{\gamma}^{(n)}(t)\right| \\
	& & +\left|\sum_{k=0}^{\lfloor t/\Delta x^{(n)} \rfloor}\left(\widehat{\mathcal{D}}_{\gamma,k}^{(n)}(\widehat{\eta}_{\gamma,k}^{(n)})-\mathbb{E}\left[\widehat{\mathcal{D}}_{\gamma,k}^{(n)}(\widehat{\eta}_{\gamma,k}^{(n)}) \right]\right)\right|.
%&\qquad+\left| \widehat{\gamma}^{(n)}(t)-\widehat{\gamma}(t)\right|_{2}. \nonumber
%\label{eq:tripleEtaSigma}
\end{eqnarray*}
For the first term, we deduce from Lemma \ref{lem:ODEconv} that
\begin{equation}
	\sup_{t\in[0,T]} |\widetilde{\gamma}^{(n)}(t)-\widehat{\gamma}(t)| =\mathcal{O}(\Delta t^{(n)}).
\label{eq:gammahatnConv}
\end{equation}
For the second term we use the Lipschitz continuity of the event probabilities $p^I(\cdot,\cdot)$ in order to establish the existence of a constant $L_\gamma > 0$ such that:
\begin{alignat}{2}
\left|\sum_{k=0}^{\lfloor t/\Delta x^{(n)} \rfloor}\mathbb{E}\left[\widehat{\mathcal{D}}_{\gamma,k}^{(n)}(\widetilde{\eta}_{\gamma,k}^{(n)}) \right]-\widetilde{\gamma}^{(n)}(t)\right|
&=\left|\sum_{k=0}^{\lfloor t/\Delta x^{(n)} \rfloor}\mathbb{E}\left[\widehat{\mathcal{D}}_{\gamma,k}^{(n)}(\widehat{\eta}_{\gamma,k}^{(n)}) \right]-\mathbb{E}\left[\widehat{\mathcal{D}}_{\gamma,k}^{(n)}(\widehat{\gamma}_{k}^{(n)}) \right]\right| \nonumber\\
&\leq\sum_{k=0}^{\lfloor t/\Delta x^{(n)} \rfloor}\left|\mathbb{E}\left[\widehat{\mathcal{D}}_{\gamma,k}^{(n)}(\widehat{\eta}_{\gamma,k}^{(n)}) \right]-\mathbb{E}\left[\widehat{\mathcal{D}}_{\gamma,k}^{(n)}(\widehat{\gamma}_{k}^{(n)})  \right]\right| \nonumber\\
&\leq \Delta x^{(n)}\cdot L_{\gamma} \cdot \sum_{k=0}^{\lfloor t/\Delta x^{(n)} \rfloor}\left|\widehat{\eta}_{\gamma,k}^{(n)}-\widehat{\gamma}_{k}^{(n)} \right|. \nonumber
%\label{ineq:priceConvI}
\end{alignat}
The third term corresponds to the noise-term of the price process. For each $n \in {\mathbb N}$, the sequence
\[
	y^{n}_k := \widehat{\mathcal{D}}_{\gamma,k}^{(n)}(\widehat{\eta}_{\gamma,k}^{(n)})-\mathbb{E}\left[\widehat{\mathcal{D}}_{\gamma,k}^{(n)}(\widehat{\eta}_{\gamma,k}^{(n)})\right],
	\quad k=0, ..., \lfloor T/\Delta x^{(n)} \rfloor
\]
is a martingale difference sequence. A direct computation shows that
\[
	\sup_{n,k} \mathbb{E}|y^n_k|^2 \leq C \cdot \left( \Delta x^{(n)} \right)^2.
\]
Hence, the law of large numbers for triangular martingale difference arrays (Theorem \ref{thm-appendix} and Corollary \ref{corollary-appendix}) implies
\[
	\lim_{n \to \infty} \mathbb{P}\left [ \sup_{0 \leq m \leq T/\Delta x^{(n)}} | \sum_{k=0}^m y^n_k | > 0 \right] = 0,
\]
just as in the proof of Lemma \ref{lemma-hatD}. Thus, using Lemma \ref{lemma-hatD} again, we see that
\begin{alignat}{2}
\left|{\eta}_{\gamma}^{(n)}(t)-\widehat{\gamma}(t)\right| =
\left|\widehat{\eta}_{\gamma}^{(n)}(t)-\widehat{\gamma}(t)\right| \leq \Delta x^{(n)}\cdot L_{\gamma} \cdot \sum_{k=0}^{\lfloor t/\Delta x^{(n)} \rfloor}\left|\widehat{\eta}_{\gamma,k}^{(n)}-\widehat{\gamma}_{k}^{(n)} \right| + o(1)\quad \text{in probability} \nonumber
\end{alignat}
for some additive term of order $o(1)$ uniform in $t \in [0,T]$. As a result, (\ref{eq:etanToeta}) follows from an application of Gronwall's lemma along with Lemma \ref{lemma-hatD}.
%and (\ref{eq:etanToeta}) follows from an application of Gronwall's lemma. %(Lemma  \ref{lem:discLem}).
%
%analogously as in the abstract setting (see (\ref{eq:TII}) of the Appendix) to show the uniform convergence of $|\eta_{\gamma}^{(n)}(t)-\widehat{\gamma}(t)|_{2}\rightarrow 0$ $a.s.$ over $[0,T]$ and thus we have proved (\ref{eq:etanToeta}).\\
%
\item[b)] Let us now consider the cumulative ``active order time process''
\begin{alignat}{2}
y^{(n)}(t)&=\sum_{k=0}^{\lfloor t/\Delta x^{(n)} \rfloor}\zeta^{(n)}\left(\widehat{\eta}_{\gamma,k}^{(n)}\right) \cdot \Delta x^{(n)}\nonumber\\
&=\Delta x^{(n)} \cdot \left\{ \sum_{k=0}^{\lfloor t/\Delta x^{(n)} \rfloor}\mathbb{E}\left[\zeta^{(n)} \left(\widehat{\eta}_{\gamma,k}^{(n)}\right) \right]+
\sum_{k=0}^{\lfloor t/\Delta x^{(n)} \rfloor}\left(\zeta^{(n)}\left(\widehat{\eta}_{\gamma,k}^{(n)}\right)-\mathbb{E}\left[\zeta^{(n)}\left(\widehat{\eta}_{\gamma,k}^{(n)}\right) \right]\right) \right\} \nonumber\\
&=\Delta x^{(n)} \sum_{k=0}^{\lfloor t/\Delta x^{(n)} \rfloor}m^{}\left(\widehat{\eta}_{\gamma,k}^{(n)}\right)+
\Delta x^{(n)} \cdot \sum_{k=0}^{\lfloor t/\Delta x^{(n)} \rfloor}\left(\zeta^{(n)}\left(\widehat{\eta}_{\gamma,k}^{(n)}\right)-\mathbb{E}\left[\zeta^{(n)}\left(\widehat{\eta}_{\gamma,k}^{(n)}\right)\right]\right).\nonumber\\
\label{eq:timeEVI}
\end{alignat}
By the above established uniform convergence of $\widehat{\eta}_{\gamma}^{(n)}$ to $\widehat{\gamma}$ in probability and because the function $m$ is Lipschitz continuous, the first sum converges to the function
\begin{equation}
	y(t) = \int_{0}^{t}m(\widehat{\gamma}(u))\mathrm{d} u.
\label{eq:timeEVII}
\end{equation}
Applying the same arguments as above to the martingale difference sequences
\[
	\zeta^{(n)}\left(\widehat{\eta}_{\gamma,k}^{(n)}\right)-\mathbb{E}\left[\zeta_{k}^{(n)}\left(\widehat{\eta}_{\gamma,k}^{(n)}\right) \right], \quad k=0, ..., \lfloor \frac{t}{\Delta x^{(n)}} \rfloor
\]
we see that the second term vanishes uniformly in $t \in [0,T]$ in probability. Thus,
\[
	\lim_{n \to \infty} \sup_{t \in [0,T]} |y^{(n)}(t) - y(t)| = 0 \quad \text{in probability.}
\]
%The second term vanishes as $n \to \infty$ by the law of large numbers.
%
%Applying the classical Kolmogorov Strong Law of Large Numbers to the second sum once more, we conclude that
%\begin{equation*}
%|y^{(n)}(t)-y(t)|\rightarrow0\quad \text{a.s. as $n\rightarrow\infty$}.
%\end{equation*}
Since $y^{(n)}$ and $y$ are increasing functions, their inverses $\mu^{(n)}$ and $\mu$ exist. By continuity
\begin{equation*}
\sup_{t\in[0,T]}|\mu^{(n)}(t)-\mu(t)|\rightarrow0\quad \text{in probability as $n\rightarrow\infty$}
\end{equation*}
%where $\mu(t)=y^{-1}(t)$ and thus we have
and
\begin{equation}
	\mu'(t)=\left(y^{-1}(t)\right)'=\frac{1}{y'(y^{-1}(t))}=\frac{1}{m(\mu(t))}=\frac{1}{m \big(\widehat{\gamma}(\mu(t))\big)}.
\label{eq:muDiff}
\end{equation}
\end{itemize}
%By the time change theorem
%\noindent The best bid and ask prices in the models $\Gamma_{k}^{(n)}$ may be written as a composition (\ref{eq:PriceCompTime}) and applying the
%
Since both the state and the time process converge, we conclude from the time change theorem  that
\begin{equation*}
\sup_{t\in[0,T]}\left|\eta_\gamma^{(n)}(t)-\gamma(t)\right| \rightarrow 0\quad\text{in probability $n\rightarrow\infty$},
\end{equation*}
where $\gamma(t)=\widehat{\gamma}(\mu(t))$ and %by the chain rule, (\ref{eq:autonStatePrice}), (\ref{eq:timeEVII}) and (\ref{eq:muDiff}) we have that
\begin{alignat}{2}
\gamma'(t)=\widehat{\gamma}'(\mu(t))\cdot \mu'(t)=\frac{A(\widehat{\gamma}(\mu(t)))}{m(\widehat{\gamma}(\mu(t)))}\left(\begin{array}{c}
                                                                 1 \\
                                                                 1
                                                               \end{array}\right)
=\frac{A(\gamma(t))}{m(\gamma(t))}\left(\begin{array}{c}
                                                                 1 \\
                                                                 1
                                                               \end{array}\right).
                                                               \nonumber
\end{alignat}
\hfill $\Box$

%%%%%%%%%%%%%%%%%%%%%%%%%%%
%%%%%%%%%%%%%%%%%%%%%%%%%%%
%%%%%%%%%%%%%%%%%%%%%%%%%%%
%%%%%%%%%%%%%%%%%%%%%%%%%%%
%%%%%%%%%%%%%%%%%%%%%%%%%%%
%%%%%%%%%%%%%%%%%%%%%%%%%%%

\section{Convergence of volume densities} \label{chapter-volumes}

In this section we prove Proposition \ref{convergence-volumes}. To this end, we denote by ${\cal D}^{(n)}_{v,k}(\cdot,\cdot)$ the restriction of the operator ${\cal D}^{(n)}_k$ to $L^2 \times L^2$, i.e the restriction of ${\cal D}^{(n)}_{k}$ to the volume components of the state process. We need to show that the sequence $\{\eta^{(n)}_v\}_{n \in \mathbb{N}}$ of $L^2 \times L^2$-valued step-functions defined recursively by
\begin{equation}
\eta_{v}^{(n)}(t, \cdot):=\eta_{v,k}^{(n)}\quad \text{for $t\in[t_{k}^{(n)},t_{k+1}^{(n)})$}
\label{def:etavt}
\end{equation}
where
\begin{alignat}{2}
\begin{cases}
\eta_{v,k+1}^{(n)}&:=\eta^{(n)}_{k}+\mathcal{D}_{v,k}^{(n)}\left(\eta_{\gamma,k}^{(n)},\eta_{v,k}^{(n)}\right)\\
\label{def:etavk}\\
\eta_{v,0}^{(n)}&:=v_{0}^{(n)}
\end{cases}
\end{alignat}
converges in probability in $L^2$ to the unique solution of the PDE (\ref{PDE2}).
We will show convergence in several steps. In a first step we find a convergent discretization scheme of the PDE which is coherent with the order book dynamics. Subsequently, we link this scheme to the expected dynamics of the volume densities.

%%%%%%%%%%%%%%%%%%%%%%%%%%%%%%%%%%%%%%%%%%%%%%%%%%%
%%%%%%%%%%%%%%%%%%%%%%%%%%%%%%%%%%%%%%%%%%%%%%%%%%%
%%%%%%%%%%%%%%%%%%%%%%%%%%%%%%%%%%%%%%%%%%%%%%%%%%%

\subsection{A numerical scheme for the limiting PDE}

 For any $n \in {\mathbb N}$, the scaling parameters $\Delta x^{(n)}$ and $\Delta t^{(n)}$ define a grid $\{(t^{(n)}_k, x^{(n)}_k)\}$ on $[0,T] \times \mathbb{R}$ through $t_{k}^{(n)}=k \cdot \Delta t^{(n)}$ $(k \in \mathbb{N}_0)$ and $x_{j}^{(n)} =j \cdot \Delta x^{(n)}$ $(j \in \mathbb{Z})$.
In a first step, we approximate the unique solution $\widehat{u}:[0,T] \times \mathbb{R} \to \mathbb{R}^2$ to (\ref{PDE2}) through a sequence of grid-point functions $\widehat{u}^{(n)}: [0,T] \times \mathbb{R} \to \mathbb{R}^2$. To this end, we put
\[
	A_R(t):=\left(\begin{array}{cc}
                          p^{A}(\hat{\gamma}(t)) & 0 \\
                          0 & p^{E}(\hat{\gamma}(t)) \\
                        \end{array}
                      \right), \quad
      A_L(t):=\left(\begin{array}{cc}
                          p^{B}(\hat{\gamma}(t)) & 0 \\
                          0 & p^{F}(\hat{\gamma}(t)) \\
                        \end{array}
                      \right)
\]
and
\[
	F(t,x) := B(\hat{\gamma}(t),x), \quad g(t,x) := c(\hat{\gamma}(t),x).
\]
Furthermore, we introduce operators $\mathcal{H}^{(n)}_t$ that act on $v \in L^2$ according to
\begin{equation} \label{def:operatorH}
\begin{split}
	\mathcal{H}^{(n)}_t(v)  := &
	%& \quad +
	%~ v + \mathbb{E}\left[\mathcal{D}_{v,k}^{(n)}\left(\hat{\gamma}(t), v \right) \right]  \nonumber \\
	~ v + \Delta p^{(n)} \cdot A_{R}(t) \left[ v(\cdot + \Delta x^{(n)})-v(\cdot) \right] \nonumber\\
	& ~~~ + \Delta p^{(n)} \cdot A_{L}(t) \left[ v(\cdot-\Delta x^{(n)})-v(\cdot) \right]  \nonumber\\
	& ~~~ + \Delta v^{(n)} \cdot (1 - \Delta p^{(n)}) \cdot \left[ F(t,\cdot) \cdot v(\cdot)+ g(t,\cdot) \right].
\end{split}
\end{equation}
The sequence of grid-point approximations is then defined recursively by
\begin{equation}
	\widehat{u}^{(n)}(t,\cdot):=\widehat{u}_{k}^{(n)}  \quad \mbox{for} \quad t \in [t^{(n}_{k}, t^{(n)}_{k+1})
\end{equation}
where
\begin{alignat}{2}
\begin{cases}
	\widehat{u}^{(n)}_{k+1} &=  \mathcal{H}^{(n)}_{t^{(n)}_k}(\widehat{u}^{(n)}_k)  \\
	\\
	%\label{def:UnjkRec}\\
	\widehat{u}^{(n)}_{0}&=v_{0}^{(n)}.
\end{cases}
\end{alignat}
The sequence of step-functions $\{\widehat{u}^{(n)}\}$ essentially describes a discretized limiting volume dynamics of the order book. We benchmark this dynamics against the expected pre-limit volume dynamics when prices are replaced by the limiting dynamics. More precisely, we introduce another sequence of step functions $u^{(n)}:[0,T] \times \mathbb{R} \to \mathbb{R}$ by
\begin{equation}
{u}^{(n)}(t,\cdot):=u_{k}^{(n)}  \quad \mbox{for} \quad t \in [t^{(n}_{k}, t^{(n)}_{k+1})  %\quad \text{and} \quad \hat{u}^{(n)}(\cdot,t):=\hat{u}_{k}^{(n)}
\label{def:Contwidetilde}
\end{equation}
where
\begin{alignat}{2}
\begin{cases}
	u_{k+1}^{(n)}& = u_{k}^{(n)}+\mathbb{E}\left[\mathcal{D}_{v,k}^{(n)}\left(\hat{\gamma}(t^{(n)}_k), u_{k}^{(n)} \right) \right]\\
\label{def:UTildenjkRec}\\
u_{0}^{(n)}&:=v_{0}^{(n)}.
\end{cases}
\end{alignat}
In a first step we are now going to show that the grid-point functions $\widehat{u}^{(n)}$ approximate the solution $\widehat{u}$ of our PDE. Subsequently we show that the PDE can as well be approximated by the functions $u^{(n)}$.

\begin{proposition}[Convergence of the numerical scheme] \label{prop-PDE-scheme}
Assume that the assumptions of Theorem \ref{MainResult} hold. Then, the processes $\widehat{u}^{(n)}$ define a convergent finite difference scheme of the PDE (\ref{PDE2}), i.e.
\begin{equation}
\sup_{t\in[0,T]}\|\widehat{u}^{(n)}(t, \cdot)-\widehat{u}(t, \cdot)\|_{L^2}\rightarrow 0\quad \text{as $n\rightarrow\infty$}.
\label{eq:convNumScheme}
\end{equation}
\begin{proof}
For $t \in [0,T]$ the local truncation error associated with the grid-point approximations $\widehat{u}^{(n)}$ is defined as
\begin{equation}
\mathcal{L}^{(n)}(t,x):=\frac{1}{\Delta t^{(n)}}\left( \widehat{u}(t+\Delta t^{(n)},x) - {\cal H}^{(n)}_t( \widehat{u}(t,\cdot) )(x) \right).
\label{def:locTrunc}
\end{equation}
Smoothness of the solution $\widehat{u}$ (bounded with uniformly bounded first and second derivatives) along with Assumption \ref{assumption-scaling} implies that the following estimate holds uniformly in $t \in [0,T]$ and $x \in \mathbb{R}$:
\begin{alignat}{2}
\mathcal{L}^{(n)}(t,x)&=\frac{1}{\Delta t^{(n)}}\Big{(} \widehat{u}(t+\Delta t^{(n)},x)-\widehat{u}(t,x)- \Delta p^{(n)} \cdot A_R(t) ( \widehat{u}(t, x+\Delta x^{(n)})-\widehat{u}(t,x)) \nonumber\\
&\qquad\qquad\qquad\qquad\qquad\qquad\qquad\quad- \Delta p^{(n)} \cdot A_L(t) (\widehat{u}(t,x-\Delta x^{(n)})-\widehat{u}(t,x))\nonumber\\
&\qquad\qquad\qquad\qquad\qquad\qquad\qquad\quad-\Delta v^{(n)} \cdot (1 - \Delta p^{(n)}) \cdot \left[ F(t,x)\widehat{u}(t,x)- g(t,x) \right]\Big{)}\nonumber\\
&=\frac{1}{\Delta t^{(n)}}\Big{(}\widehat{u}_{t}(t,x)\Delta t^{(n)}+o(\Delta t^{(n)})- \Delta p^{(n)} \cdot A_R (\widehat{u}_{x}(t,x)\Delta x^{(n)}+o(\Delta x^{(n)}))\nonumber\\
&\qquad\qquad\qquad\qquad\qquad\qquad\qquad\quad- \Delta p^{(n)} \cdot A_L(-\widehat{u}_{x}(t,x))\Delta x^{(n)}+o(\Delta x^{(n)}))\nonumber\\
&\qquad\qquad\qquad\qquad\qquad\qquad\qquad\quad-\Delta v^{(n)} \cdot (1 - \Delta p^{(n)}) \cdot \left[ F(t,x)\widehat{u}(t,x)- g(t,x) \right]\Big{)}\nonumber\\
&=\widehat{u}_{t}(x,t)-\left(A_{L}(t)-A_{R}(t)\right)\widehat{u}_{x}(x,t)-F(x,t)\widehat{u}(x,t)-g(x,t)+o(1)\qquad\qquad\qquad\nonumber\\
&=\widehat{u}_{t}(t,x)-A(t)\widehat{u}_{x}(t,x)-B(t,x)\widehat{u}(t,x)-c(t,x)+o(1)\nonumber\\
&=o(1) %\quad \mbox{uniformly in $t \in [0,T]$ and $x \in \mathbb{R}$.}
\label{eq:localConv}
\end{alignat}
%uniformly in $t \in [0,T]$ and $x \in \mathbb{R}$.
%\\[4pt]
The solution $\widehat{u}(t;\cdot)$ the PDE vanishes outside a compact interval for all $t \in [0,T]$ (Proposition \ref{prop-PDE-solution}) and hence so does $\mathcal{L}^{(n)}(t,\cdot)$. As a result, we also have
\begin{equation} \label{L2-L}
	\lim_{n \to \infty} \sup_{t \in [0,T]} \| \mathcal{L}^{(n)}(t,\cdot) \|_{L^2} = 0.
\end{equation}

 From (\ref{def:locTrunc}), one has that $\widehat{u}(t+\Delta t^{(n)}, \cdot)=\mathcal{H}^{(n)}_t(\widehat{u}(t,\cdot))+\Delta t^{(n)}\mathcal{L}^{(n)}(t, \cdot).$
%\[
%	\widehat{u}(t+\Delta t^{(n)}, \cdot)=\mathcal{H}^{(n)}_t(\widehat{u}(t,\cdot))+\Delta t^{(n)}\mathcal{L}^{(n)}(t, \cdot).
%\]
In terms of the error-function
\begin{equation}
\delta \widehat{u}^{(n)}(t,x):=\widehat{u}^{(n)}(t,x) - \widehat{u}(t,x)
\label{def:ErrFunc}
\end{equation}
this yields
\begin{eqnarray*}
	\delta \widehat{u}^{(n)} (t^{(n)}_{k+1}; \cdot) &=& \widehat{u}^{(n)}(t^{(n)}_{k+1}; \cdot) - \widehat{u}(t^{(n)}_k + \Delta t^{(n)}; \cdot) \\
	& = & \mathcal{H}^{(n)}_k(\widehat{u}^{(n)}(t^{(n)}_k, \cdot))-\mathcal{H}^{(n)}_k(\widehat{u}(t^{(n)}_k, \cdot))-\Delta t^{(n)}\mathcal{L}^{(n)}(t^{(n)}_k, \cdot) \\
	&=& \mathcal{H}^{(n)}_k(\delta \widehat{u}^{(n)}(t^{(n)}_k, \cdot))-\Delta t^{(n)}\mathcal{L}^{(n)}(t^{(n)}_k, \cdot),
\end{eqnarray*}
due to the linearity of the operator $\mathcal{H}^{(n)}_k$. Using this property iteratively, and putting $(\mathcal{H}^{(n)})^k := \mathcal{H}^{(n)}_1 \circ \cdots \circ \mathcal{H}^{(n)}_k$ one finds
\begin{equation}
	\delta \widehat{u}(t_{k+1}^{(n)}, \cdot)=(\mathcal{H}^{(n)})^k(\delta \widehat{u}^{(n)}(0,\cdot))-\Delta t^{(n)}\sum_{i=0}^{k}(\mathcal{H}^{(n)})^{k-i}(\mathcal{L}^{(n)}(t^{(n)}_{i}), \cdot).
\label{eq:ErrorIter}
\end{equation}
From the definition of $\mathcal{H}^{(n)}_t$ together with the fact that the functions $F$ and $g$ are uniformly bounded by assumption, one finds (for large enough $n \in \mathbb{N}$):
\begin{eqnarray*}
	\| \mathcal{H}^{(n)}_t (v) \|_{L^2} & \leq & \|v\|_{L^2} \left( 1- \Delta p^{(n)} (A_R(t) + A_S(t)) \right) + \|T^{(n)}_+(v) \|_{L^2} \cdot \left( \Delta p^{(n)} \cdot A_R(t) \right) \\
	& & + \|T^{(n)}_-(v) \|_{L^2} \cdot \left( \Delta p^{(n)} \cdot A_L(t) \right) + \Delta v^{(n)} \left( \|F\|_\infty \|v\|_{L^2} + \|g\|_\infty \right).
\end{eqnarray*}
Thus, by the isometry property $\| T^{(n)}_\pm(v) \|_{L^2} = \| v \|_{L^2}$ of the translation operators, there exists a constant $C > 0$ that is independent of $t \in [0,T]$ such that
\[
	\sup_{\| v\|_{L^2} = 1} \| \mathcal{H}^{(n)}_t (v) \|_{L^2} \leq 1 + C \Delta t^{(n)}.
\]
In particular, since $k \leq \lfloor \frac{T}{\Delta t^{(n)}} \rfloor$
\begin{eqnarray*}
	\sup_ k \left \| (\mathcal{H}^{(n)})^k(\delta \widehat{u}^{(n)}(0,\cdot)) \right\|_{L^2} & \leq & (1+C \Delta t^{(n)})^{\lfloor T/ \Delta t^{(n)} \rfloor}
	\| \delta \widehat{u}^{(n)}(0,\cdot) \|_{L^2} \\
	& \leq & e^{CT} \| \delta \widehat{u}^{(n)}(0,\cdot) \|_{L^2}  \\
	& = & o(1),
\end{eqnarray*}
where the last equality follows from Assumption \ref{Assumption-1}. Similarly, from (\ref{L2-L}):
\[
	\sup_{k,i,n} \left \| (\mathcal{H}^{(n)})^{k-i}(\mathcal{L}^{(n)}(t^{(n)}_{i-1}), \cdot) \right\|_{L^2} = o(1).
\]
Using the same arguments as before, one concludes that the sum in (\ref{eq:ErrorIter}) vanishes uniformly in time. Hence,
\begin{alignat}{2}
	\lim_{n \to \infty} \sup_{t \in [0,T]} \| \widehat{u}^{(n)}(t,\cdot)- \widehat{u}(t,\cdot) \|_{L^2}&=0.
	\nonumber
\end{alignat}
%\Halmos
\end{proof}
\label{prop:numSchemConv}
\end{proposition}
\vspace{3mm}
 Next, we show that the functions $u^{(n)}$ also approximate the PDE. More precisely, the following holds.

\begin{proposition} \label{prop-PDE-scheme2}
Under the assumptions of Theorem \ref{MainResult}
\[
	\lim_{n \to \infty} \sup_{t \in [0,T]} \| \widehat{u}^{(n)}(t;\cdot) - {u}^{(n)}(t;\cdot) \|_{L^2} = 0.
\]
\end{proposition}
\begin{proof}
The proof is similar to that of Proposition \ref{prop-PDE-scheme}. By analogy to the operator
$\mathcal{H}^{(n)}_t$ we introduce an operator $\widehat{\mathcal{H}}^{(n)}_t$ on $L^2$ by
\begin{equation} 
\begin{split} \label{def:operatorH1}
	\widehat{\mathcal{H}}^{(n)}_t(v)  := &
	%& \quad +
	%~ v + \mathbb{E}\left[\mathcal{D}_{v,k}^{(n)}\left(\hat{\gamma}(t), v \right) \right]  \nonumber \\
	~ v + \Delta p^{(n)} \cdot A_{R}(t) \left[ v(\cdot + \Delta x^{(n)})-v(\cdot) \right] \nonumber\\
	& ~~~ + \Delta p^{(n)} \cdot A_{L}(t) \left[ v(\cdot-\Delta x^{(n)})-v(\cdot) \right]  \nonumber\\
	& ~~~ + \Delta v^{(n)} \cdot (1-\Delta p^{(n)}) \left[ F^{(n)}(t,\cdot) \cdot v(\cdot) + g^{(n)}(t,\cdot) \right].
\end{split}
\end{equation}
where for $t \in [t^{(n)}_k, t^{(n)}_{k+1})$:
\begin{eqnarray*}
	F^{(n)}(t,\cdot) & := & \left( \begin{array}{cc}
	-f^{(n),C}(\cdot) \cdot p^C(\widehat{\gamma}(t^{(n)}_k)) & 0 \\ 0 & -f^{(n),G}(\cdot) \cdot p^G(\widehat{\gamma}(t^{(n)}_k))
	\end{array} \right) \\
	g^{(n)}(t,\cdot) & := & \left( \begin{array}{c}
	f^{(n),D}(\cdot) \cdot p^D(\widehat{\gamma}(t^{(n)}_k)) \\ f^{(n),H}(\cdot) \cdot p^H(\widehat{\gamma}(t^{(n)}_k))  \end{array} \right).
\end{eqnarray*}
For the error function $\delta \widehat{u}^{(n)}_k(\cdot) := \widehat{u}^{(n)}_k(\cdot) - {u}^{(n)}_k(\cdot)$ we then obtain
\begin{eqnarray*}
	\delta \widehat{u}^{(n)}_{k+1}(\cdot) & = & \widehat{\mathcal{H}}^{(n)}_t(\delta \widehat{v}^{(n)}_{k}		(\cdot)) + \Delta v^{(n)} \cdot (1 - \Delta p^{(n)}) \cdot \left( \delta F^{(n)}(\cdot) \cdot \widehat{u}^{(n)}_k(\cdot) + \delta g^{(n)}
	(\cdot) \right) \\
	& =: & \widehat{\mathcal{H}}^{(n)}_t(\delta \widehat{u}^{(n)}_{k}(\cdot)) + \Delta v^{(n)} \cdot (1-\Delta p^{(n)}) \cdot
	\widehat{\mathcal{L}}(t^{(n)}_k;\cdot)
\end{eqnarray*}
where
\begin{eqnarray*}
	\delta F^{(n)}(t,\cdot) &:=& \left( \begin{array}{cc}
	-( f^{(n),C}(\cdot) - f^C(\cdot) ) \cdot p^C(\widehat{\gamma}(t^{(n)}_k))  & 0 \\ 0 & -( f^{(n),G}(\cdot) -  f^G(\cdot)) \cdot p^G(\widehat{\gamma}(t^{(n)}_k))
	\end{array} \right), \\
	\delta g^{(n)}(t,\cdot) &:=& \left( \begin{array}{c}
	\left( f^{(n),D}(\cdot) -  f^D(\cdot) \right) \cdot p^D(\widehat{\gamma}(t^{(n)}_k))  \\
	\left( f^{(n),H}(\cdot) -  f^H(\cdot) \right) \cdot p^H(\widehat{\gamma}(t^{(n)}_k))
	\end{array} \right).
\end{eqnarray*}
By construction, the grid-point functions $\widehat{u}^{(n)}$ are uniformly bounded. As it result, it follows from Assumption \ref{Assumption-2} that
\[
	\lim_{n \to \infty} \sup_{k=0, ..., \lfloor T/\Delta t^{(n)} \rfloor} \| \widehat{\mathcal{L}}(t^{(n)}_k;\cdot) \|_{L^2}  = 0.
\]	
One can now proceed as in the proof of Proposition \ref{prop-PDE-scheme}, to conclude that
\[
	\lim_{n \to \infty} \sup_{t \in [0,T]} \| \widehat{u}^{(n)}(t,\cdot) - u^{(n)}(t;\cdot) \|_{L^2} = 0.
\]
%\Halmos
\end{proof}

%%%%%%%%%%%%%%%%%%%%%%%%%%%%%%%%%%%%%%
%%%%%%%%%%%%%%%%%%%%%%%%%%%%%%%%%%%%%%
%%%%%%%%%%%%%%%%%%%%%%%%%%%%%%%%%%%%%%

\subsection{Expected volume dynamics and discretized PDEs}

 To show the convergence of volume density functions we compare the random states $\eta_{v}^{(n)}$ with the deterministic  approximations of the limiting PDE obtained in the previous subsection. For this, we introduce the deterministic step function valued processes $\widetilde{u}^{(n)}$:
\begin{equation}
\widetilde{u}^{(n)}(t,\cdot):=\widetilde{u}_{k}^{(n)}  \quad \mbox{for } t \in [t^{(n}_{k}, t^{(n)}_{k+1}) %\quad \text{and} \quad \hat{u}^{(n)}(\cdot,t):=\hat{u}_{k}^{(n)}
\label{def:Contwidetilde}
\end{equation}
where
\begin{alignat}{2}
\begin{cases}
	\widetilde{u}_{k+1}^{(n)}&:=\widetilde{u}_{k}^{(n)}+\mathbb{E}\left[\mathcal{D}_{v,k}^{(n)}\left( \eta^{(n)}_{\gamma,k}, \widetilde{u}^{(n)}_k \right) \right]\\
\label{def:UTildenjkRec}\\
\widetilde{u}_{0}^{(n)}&:=v_{0}^{(n)}.
\end{cases}.
\end{alignat}

 The process $\widetilde{u}^{(n)}$ describes the expected dynamics of the volume density functions for the actual price process; in particular, $\widetilde{u}^{(n)}$ is a stochastic process. By contrast, the process $u^{(n)}$ describes the dynamics of the expected volume density functions when the random evolution of bid and ask prices is replaced by its deterministic limit. We have:
\begin{eqnarray*}
	\|\eta_{v}^{(n)}(t,\cdot)-\widehat{u}(t,\cdot)\|_{L^2} & \leq &
       \|\eta_{v}^{(n)}(t,\cdot)-\widetilde{u}^{(n)}(t,\cdot)\|_{L^2} + \|\widetilde{u}^{(n)}(t,\cdot)-u^{(n)}(t,\cdot)\|_{L^2} \\
       && + \| u^{(n)}(t,\cdot)-\widehat{u}^{(n)}(t,\cdot) \| +\|\widehat{u}^{(n)}(t,\cdot)-\widehat{u}(t,\cdot)\|_{L^2}.
\label{eq:Theo:DensConv}
\end{eqnarray*}
The last two terms are deterministic and converges uniformly to zero by Propositions \ref{prop:numSchemConv} and \ref{prop-PDE-scheme2}. It remains to show convergence of the first two (random) terms. This will be achieved in the following two subsections.

%%%%%%%%%%%%%%%%%%%%%%%%%%%%%%%%%%%%%%
%%%%%%%%%%%%%%%%%%%%%%%%%%%%%%%%%%%%%%
%%%%%%%%%%%%%%%%%%%%%%%%%%%%%%%%%%%%%%

\subsubsection{Estimating the price impact of expected volume dynamics}

 The term $\|\widetilde{u}^{(n)}(t,\cdot)- u^{(n)}(t,\cdot)\|_{L^2}$ measures the impact of the noise in the  price process on the expected standing volume. The following proposition shows that it converges to zero almost surely (i.e. for almost all price processes), uniformly over compact time intervals.

\begin{proposition} \label{prop-tilde-u}
Under the assumptions of Theorem \ref{MainResult} it holds that
\begin{equation*}
\sup_{t\in[0,T]}\|\widetilde{u}^{(n)}(t,\cdot)- u^{(n)}(t,\cdot)\|_{L^2}\rightarrow 0\quad\text{in probability as $n\rightarrow\infty$.}
\end{equation*}
\end{proposition}
{
\begin{proof}
We argue again as in the proof of Proposition \ref{prop-PDE-scheme}. Analogously to the operator $\widehat{\mathcal{H}}^{(n)}$ defined in the proof of Proposition \ref{prop-PDE-scheme2} we define for $t \in [t^{(n)}_k, t^{(n)}_{k+1})$ the operator  
\begin{equation} \label{def:operatorH2}
\begin{split}
	\widetilde{\mathcal{H}}^{(n)}_k(v)  := &
	%& \quad +
	%~ v + \mathbb{E}\left[\mathcal{D}_{v,k}^{(n)}\left(\hat{\gamma}(t), v \right) \right]  \nonumber \\
	~ v + \Delta p^{(n)} \cdot \widetilde{A}^{(n)}_{R}(t) \left[ v(\cdot + \Delta x^{(n)})-v(\cdot) \right] \nonumber\\
	& ~~~ + \Delta p^{(n)} \cdot \widetilde{A}_{L}(t) \left[ v(\cdot-\Delta x^{(n)})-v(\cdot) \right]  \nonumber\\
	& ~~~ + \Delta v^{(n)} \cdot (1-\Delta p^{(n)}) \left[ \widetilde{F}^{(n)}(t,\cdot) \cdot v(\cdot) + \widetilde{g}^{(n)}(t,\cdot) \right].
\end{split}
\end{equation}
where for $t \in [t^{(n)}_k, t^{(n)}_{k+1})$:
\begin{eqnarray*}
	\widetilde{A}^{(n)}_{R}(t) & := & \left( \begin{array}{cc}
	p^A(\eta^{(n)}_{\gamma,k}) & 0 \\ 0 & p^E(\eta^{(n)}_{\gamma,k}) 
	\end{array} \right) 
\end{eqnarray*}
and $\widetilde{A}^{(n)}_{L}(t)$, $\widetilde{F}^{(n)}(t,\cdot)$ and $\widetilde{g}^{(n)}(t,\cdot)$ are defined analogously. Let us further put $\delta p^{(n),I}_k := p^{I}(\hat{\gamma}(t^{(n)}_k)) - p^{I}(\eta^{(n)}_{\gamma,k})$ for $I={\bf A},{\bf B},{\bf E},{\bf F}$ and
\[
	\delta A^{(n)}_{R,k}:=\left(\begin{array}{cc}
                          \delta p^{(n),A}_k  & 0 \\
                          0 & \delta p^{(n),E}_k
                        \end{array}
                      \right), \quad
      \delta A^{(n)}_{L.k} :=\left(\begin{array}{cc}
                           \delta p^{(n),B}_k  & 0 \\
                          0 & \delta p^{(n),F}_k
                        \end{array}
                      \right)
\]
and denote by $\delta f^{(n)}_k$ and $\delta g^{(n)}_k$ the corresponding quantities for cancelations and limit order placements. 
Then the error function $\delta \widetilde{u}^{(n)}_k(\cdot) :=  \widetilde{u}^{(n)}_k(\cdot)- u^{(n)}_k(\cdot)$
%\[
%	\delta u^{(n)}_k(\cdot) :=  \widetilde{u}^{(n)}_k(\cdot)- u^{(n)}_k(\cdot)
%\]	
satisfies $\delta \widetilde{u}^{(n)}_0 = 0$ and can be represented in terms of the operator $\widetilde{\mathcal{H}}^{(n)}_k$ as follows: 
\[
	\delta \widetilde{u}^{(n)}_{k+1}(\cdot) = \widetilde{\mathcal{H}}^{(n)}_k(\delta \widetilde{u}^{(n)}_{k}(\cdot) ) + \Delta t^{(n)} \cdot \widetilde{\mathcal{L}}^{(n)}_k(t^{(n)}_{k},\cdot )
\]
where 
\begin{equation}
\begin{split}
	\widetilde{\mathcal{L}}^{(n)}_k(t^{(n)}_k, \cdot) & := \frac{\Delta p^{(n)}}{\Delta t^{(n)}}  \cdot \delta {A}^{(n)}_{R,k} \left[ u^{(n)}_k(\cdot + \Delta x^{(n)})-u^{(n)}_k(\cdot) \right] \nonumber\\
	& ~~~ + \frac{\Delta p^{(n)}}{\Delta t^{(n)}} \cdot \delta {A}^{(n)}_{L,k} \left[ u^{(n)}_k(\cdot - \Delta x^{(n)})-u^{(n)}_k(\cdot) \right]  \nonumber\\
	& ~~~ + (1-\Delta p^{(n)}) \cdot \left[ \delta {f}^{(n)} \cdot u^{(n)}_k(\cdot) + \delta {g}^{(n)}_k \right].
\end{split}
\end{equation}
Corollary \ref{un-lip} establishes
\[
	\| u^{(n)}_k\|_{L^2} \leq L \quad \mbox{and} \quad \| T^{(n)}_\pm u^{(n)}_{k} - u^{(n)}_{k} \|_{L^2} \leq L \cdot \Delta x^{(n)}
\]
for some constant $L < \infty$ that is independent of $(n,k)$. Using our assumptions on the placements and cancelations functions and event probability functions along with the fact that 
\begin{equation*}
	\lim_{n \to \infty} \sup_{k=0, ..., \lfloor T/\Delta t^{(n)} \rfloor}  |\eta^{(n)}_{\gamma,k}-\widehat{\gamma}(t^{(n)}_k)| \rightarrow 0\quad \text{in probability},
\end{equation*}
this implies 
\begin{equation*} 
	\lim_{n \to \infty} \sup_{k=0, ..., \lfloor T/\Delta t^{(n)} \rfloor}  \| \widetilde{\mathcal{L}}^{(n)}_k(t^{(n)}_k,\cdot) \|_{L^2} = 0.
\end{equation*}
We can now argue as in the proof of Proposition \ref{prop-PDE-scheme} to conclude. 
\end{proof}}

%%%%%%%%%%%%%%%%%%%%%%%%%%%%%%%%%%%%%%%%%
%%%%%%%%%%%%%%%%%%%%%%%%%%%%%%%%%%%%%%%%%
%%%%%%%%%%%%%%%%%%%%%%%%%%%%%%%%%%%%%%%%%

\subsubsection{Convergence of volumes to their expected values}

In this subsection we apply a law of large number for Hilbert space-valued triangular martingale difference arrays (TMDAs) in order to establish the missing convergence to zero of the distance between $\eta^{(n)}_v$ and $\widetilde{u}^{(n)}$. More precisely, our goal is to prove the following result.

\begin{proposition}Suppose the assumption of Theorem \ref{MainResult} hold. Then,
\begin{equation*}
	\sup_{t\in[0,T]}\|\eta_{v}^{(n)}(t,\cdot)-\widetilde{u}^{(n)}(t, \cdot)\|_{L^2} \to 0 \quad\text{in probability} \quad \mbox{as } n \to \infty.
\end{equation*}
\label{prop:etaNuTildeConv}
\end{proposition}
\begin{proof}
Using the definition of $\eta_{v,k}^{(n)}$ in (\ref{def:etavk}) and $\widetilde{u}_{k}^{(n)}$  in (\ref{def:Contwidetilde}) we see that
\[
	\widetilde{u}^{(n)}_k = \mathbb{E} \eta_{v,k}^{(n)},
\]
conditioned on the price process. As a result,
\begin{alignat*}{2} \label{eq:TheoEtaVbConv}
\|\eta_{v}^{(n)}(t,\cdot)-\widetilde{u}^{(n)}(t,\cdot)\|_{L^2}
= \|\sum_{k=0}^{\lfloor t/\Delta t^{(n)} \rfloor}\left(\mathcal{D}_{v,k}^{(n)}(\cdot,\cdot)-\mathbb{E}[\mathcal{D}_{v,k}^{(n)}(\eta_{\gamma,k}^{(n)},\eta_{v,k}^{(n)}) ]\right) \|_{L^2}.
\end{alignat*}
In order to establish convergence of the sum to zero uniformly in time we introduce the $L^2$-valued triangular martingale-difference-array
\begin{equation}\label{def-Y}
	Y^n_k := \mathcal{D}_{v,k}^{(n)}(\cdot,\cdot)-\mathbb{E}[\mathcal{D}_{v,k}^{(n)}(\eta_{\gamma,k}^{(n)},\eta_{v,k}^{(n)})].
\end{equation}
If we can show that there exists $\beta > \frac{1}{2}$ such that
\begin{equation} \label{bound-expected-value}
	\sup_{n,k} \left( \frac{1}{\Delta t^{(n)}} \right)^{2\beta} \mathbb{E}[\|Y^n_k\|^2_{L^2}] < \infty,
\end{equation}
then Theorem \ref{thm-appendix} and Corollary \ref{corollary-appendix} would guarantee that
\begin{equation} \label{limit}
	\lim_{n \to \infty} \mathbb{P} \left[ \sup_{0 \leq m\leq \lfloor T/\Delta t^{(n)} \rfloor} \| \sum_{k=0}^m  Y^{n}_k \|_{L^2} > \epsilon \right] = 0
\end{equation}
and the proposition would be proved. {To establish (\ref{bound-expected-value}), we need to bound the following terms: 
	\begin{eqnarray*}
		\sup_{n,k} \mathbb{E} \left[ \left\| {\bf 1}_{D_k,H_k} \frac{\Delta v^{(n)}}{\Delta x^{(n)}} M^{(n),D,H}_{v,k} - 
		 \mathbb{E}\left[ {\bf 1}_{D_k,H_k} \frac{\Delta v^{(n)}}{\Delta x^{(n)}} M^{(n),D,H}_{v,k} \right] \right\|^2_{L^2} 
		 \right]  \\
		\sup_{n,k} \mathbb{E} \left[ \left\| {\bf 1}_{C_k,G_k} \frac{\Delta v^{(n)}}{\Delta x^{(n)}} M^{(n),C,G}_{v,k} \eta^{(n)}_{v,k} - 
		 \mathbb{E}\left[ {\bf 1}_{C_k,G_k} \frac{\Delta v^{(n)}}{\Delta x^{(n)}} M^{(n),C,G}_{v,k} \eta^{(n)}_{v,k} \right] \right\|^2_{L^2} 
		 \right] \\
		 \sup_{n,k} \mathbb{E} \left[ \left\| {\bf 1}_{A_k,B_k}  \left(T_{\pm}^{(n)}\left({\eta}_{v_b,k}^{(n)}\right)-{\eta}_{v_b,k}^{(n)} \right) - 
		 \mathbb{E}\left[ {\bf 1}_{A_k,B_k} \left(T_{\pm}^{(n)}\left({\eta}_{v_b,k}^{(n)}\right)-{\eta}_{v_b,k}^{(n)} \right) \right] \right\|^2_{L^2} 
		 \right] \\
		  \sup_{n,k} \mathbb{E} \left[ \left\| {\bf 1}_{E_k,F_k}  \left(T_{\pm}^{(n)}\left({\eta}_{v_s,k}^{(n)}\right)-{\eta}_{v_s,k}^{(n)} \right) - 
		 \mathbb{E}\left[ {\bf 1}_{E_k,F_k} \left(T_{\pm}^{(n)}\left({\eta}_{v_s,k}^{(n)}\right)-{\eta}_{v_s,k}^{(n)} \right) \right] \right\|^2_{L^2} 
		 \right].  
	\end{eqnarray*}
This is done in Lemma \ref{last-lemma} in the appendix. In particular, this lemma shows that the placement/cancellation and shift terms of of the order 
\[
	\mathcal{O}\left( (\Delta v^{(n)})^{2} \right) \quad \mbox{and} \quad  \mathcal{O}\left( (\Delta v^{(n)})^{2 \alpha} + (\Delta v^{(n)})^{2 - \alpha}\right),
\]
respectively. Since $\alpha \in (1/2,1)$  the assertion follows for $\beta := \min\{ \alpha, 1 - \alpha/2 \} $. }

\end{proof}

%%%%%%%%%%%%%%%%%%%%%%%%%%%
%%%%%%%%%%%%%%%%%%%%%%%%%%%
%%%%%%%%%%%%%%%%%%%%%%%%%%%

\section{Conclusion}
In this work a law of large numbers for limit order books was established. Starting from order arrival and cancelation rates for all price levels, we showed that the LOB dynamics can be described by a coupled PDE:ODE system when tick and order sizes tend to zero while arrival rates tend to infinity in a particular way. A key insight is that the scaling limit requires two time scales: a fast time scale for passive order arrivals and a comparably slow time scale for active order arrivals. The proof of convergence of volume densities was carried out in three steps: We first showed that the expected LOB dynamics resembles a numerical scheme of hyperbolic PDEs plus noise, provided the random price dynamics is replaced by its deterministic limit. Subsequently, we showed that the impact of the noise in the price process on the volume dynamics vanishes in the limit. Finally, we used a law of large numbers for triangular martingale difference arrays to prove that the LOB model converges to its expected value. Our model allows for approximation of key order book statistics such as expected price increments, expected standing volumes at future times and expected time to fill.

Several questions remain open. First, it would be interesting to establish a CLT or, more generally, a diffusion approximation for LOBs. Based on the idea of having different time scales for active and passive order arrivals, Bayer, Horst and Qiu \cite{BayerHorstQiu} have recently established a first SPDE scaling limit for order books. However, they assume that cancellations are subject to additive (rather than multiplicative) noise so volumes may become negative. Second, it would be interesting to solve models of optimal portfolio liquidation based on our limiting model. For this end, the model should ideally be calibrated to market data.

\begin{appendix}
\section{A Law of Large Numbers for Banach-Space-Valued TMDAs}

Here we prove a law of large numbers for triangular martingale difference arrays taking values in  real separable p-uniformly smooth Banach spaces. A Banach space $E$ is called $p$-uniformly smooth, where $p\in(1,2]$, if
 \begin{equation*}
        \rho_{E}(\tau)=\sup \left\{\|\frac{x+y}{2}\|+\|\frac{x-y}{2}\|-1:\|x\|=1,\|y\|=\tau \right\} = {\cal O}(\tau^p).
\end{equation*}
All Hilbert spaces are 2-uniformly smooth by the parallelogram identity. The spaces $C$, $l^{1}$ and $L^{1}$ are not uniformly smooth.

%\hspace{1mm}

\begin{definition} \label{def-TMDA}
A family of random variables $y^n_k$, $k = 1, ..., n$, $n=1,2, ...$ defined on some probability space $(\Omega,{\cal F}, \mathbb{P})$ is called a triangular martingale difference array (TMDA) with respect to a family $\{ {\cal F}^n \}_{n =1,2, ...}$  of filtrations, ${\cal F}^n = \{{\cal F}^n_k\}_{k=0}^n,$ if for all $n=1,2, ...$ the sequence $y^n_1, ..., y^n_n$ is an ${\cal F}^n$-martingale difference sequence (MDS), i.e.
\[
	\mathbb{E}\left[y^n_k | {\cal F}_{k-1}^n\right] = 0.
\]
\end{definition}

%\hspace{1mm}

 If $\{y^n_k\}$ is a TMDA, then for all $n=1,2, ...$ one has
\[
	\mathbb{E}[\sum_{j=1}^k y^n_k | {\cal F}_{k-1}^n] = y^n_k,
\]
that is, partial sums are martingale. For such martingales, Pisier \cite{Pisier} proved the following moment estimate.

\begin{lemma}\label{lemma-Pisier}
Let $E$ be a real separable $p$-uniformly smooth Banach space $(1 \leq p \leq 2)$. Then, for all $r \geq 1$ there exists a constant $C>0$ such that for all martingales
\[
	\left\{ \left( \sum_{i=1}^n X_i, {\cal G}_n \right) \right\}_{n \geq 1}
\]
with values in $E$, we have
\[
	\mathbb{E} \left[ \sup_{n \geq 1} \left| {  X_n} 
	%\sum_{i=1}^n X_i 
	\right|\right ]^r  \leq C \mathbb{E} \left( \sum_{n=1}^\infty |{X_i - X_{i-1}}|^p \right)^{r/p}.
\]
\end{lemma}

 This lemma allows us to prove the following law of large numbers for TMDAs.

\begin{theorem} \label{thm-appendix}
Let $y^n_k$, $k = 1, ..., n$, $n=1,2, ...$ be a TMDA taking values in a real separable $p$-uniformly smooth Banach space $E$ for $1 \leq p \leq 2$ such that
\[
	\sup_{n,k} \mathbb{E} |y^n_k|^p < \infty.
\]
Then, for all $\beta > 0$ such that $\beta \cdot p > 1$ one has for all $\epsilon > 0$  that
\[
	\lim_{n \to \infty} \mathbb{P}\left[ \sup_{1 \leq m \leq n} | \sum_{k=1}^m y^n_k | \geq \epsilon \cdot n^\beta  \right] = 0.
\]
\end{theorem}
\begin{proof}
By Markov's inequality
\[
	\mathbb{P}\left[ \sup_{1 \leq m \leq n} | \sum_{k=1}^m y^n_k | \geq \epsilon \cdot n^\beta  \right] \leq
	\frac{1}{\epsilon n^{q \cdot \beta}} \mathbb{E}\left[ \sup_{1 \leq m \leq n} | \sum_{k=1}^m y^n_k | \right]^q
\]
for all $q \geq 1$. Thus, it follows from Lemma \ref{lemma-Pisier} that
\begin{eqnarray*}
	\mathbb{P}\left[ \sup_{1 \leq m \leq n} | \sum_{k=1}^m y^n_k | \geq \epsilon \cdot n^\beta  \right] & \leq & C n^{-\beta \cdot q} \mathbb{E} \left[ \sum_{k=1}^n |y^n_k|^p \right]^{q/p}  \\
	& \leq & C n^{-\beta \cdot q + \frac{q}{p}}
\end{eqnarray*}
for a generic constant $C > 0$ since the random variables $y^n_k$ have a uniformly bounded $p$:th moment. Hence, the assertion follows as soon as $-\beta \cdot q + \frac{q}{p} < 0$, which holds for all $q > 0$ as $\beta \cdot p > 1$.
%\Halmos
\end{proof}

%\hspace{1mm}

 As an immediate corollary from the preceding theorem one obtains the following law of large numbers for TMDAs.

\begin{corollary}\label{corollary-appendix}
{Let $y^n_k$, $k = 1, ..., n$, $n=1,2, ...$ be a TMDA taking values in a real separable $2$-uniformly smooth Banach space $E$ such that
\[
	\sup_{n,k} \Big( n^{2 \beta} \mathbb{E} |y^n_k|^2 \Big) < \infty
\]
for some $\beta > \frac{1}{2}$. Then
\[
	\lim_{n \to \infty}  \sup_{1 \leq m \leq n} | \sum_{k=1}^m y^n_k | = 0 \quad \mbox{in probability.}
\]
}
%If the stronger condition
%\begin{equation} \label{stronger-condition}
%	\sup_k \Big( n \cdot \mathbb{E} |y^n_k|^2 \Big) \leq \frac{C}{n^{1+\epsilon}}
%\end{equation}
%for some $0 < C < \infty$ and $\epsilon > 0$, then the limit above hold almost surely.
\end{corollary}
\begin{proof}
We apply Theorem \ref{thm-appendix} to the TMDA
\[
	\hat{y}^n_k := n^\beta y^n_k.
\]
Then
\[
	\lim_{n \to \infty} \mathbb{P}\left[ \sup_{1 \leq m \leq n} | \sum_{k=1}^m \hat{y}^n_k | \geq \epsilon \cdot n^\beta  \right] = 0.
\]
{ Hence the assertion follows from:}
\[
	\mathbb{P}\left[ \sup_{1 \leq m \leq n} | \sum_{k=1}^m \hat{y}^n_k | \geq \epsilon \cdot n^\beta  \right] =
	\mathbb{P}\left[ \sup_{1 \leq m \leq n} | \sum_{k=1}^m y^n_k | \geq \epsilon \right] = 0.
\]
%\Halmos
\end{proof}

%%%%%%%%%%%%%%%%%%%%%%%%%%%%%%
%%%%%%%%%%%%%%%%%%%%%%%%%%%%%%
%%%%%%%%%%%%%%%%%%%%%%%%%%%%%%

\section{Properties of volume density functions}

In this appendix, we prove some properties of the volume density functions. In particular, we show that the sequences $\{\eta^{(n)}_{v,k}\}$ take values in $L^2$ almost surely. We first use an induction argument to establish a useful representation of the volume density function.

\begin{lemma} \label{lemma-B1}
The buy side volume density function $\eta^{(n)}_{v_b}$ satisfies:
\begin{small}
\begin{equation} \label{eq:indFormulaEtavbk}
\begin{split}
%\begin{alignat}{2}
%{\eta}_{v_{b},k}^{(n)}\nonumber\\
	 {\eta}_{v_{b},k}^{(n)}=&\left(\left(T_{+}^{(n)}\right)^{\sum_{i=0}^{k-1}\mathbf{1}_{i}^{(n),A}}\circ\left(T_{-}^{(n)}\right)^{\sum_{i=0}^{k-1}\mathbf{1}_{i}^{(n),B}} \right)\left(v_{b,0}^{(n)}\right)\nonumber\\
&+ \frac{\Delta v^{(n)}}{{\Delta x^{(n)}}} \cdot \sum_{i=0}^{k-1}\left\{\left(\left(T_{+}^{(n)}\right)^{\sum_{j=i+1}^{k-1}\mathbf{1}_{j}^{(n),A}}\circ\left(T_{-}^{(n)}\right)^{\sum_{j=i+1}^{k-1}\mathbf{1}_{j}^{(n),B}} \right)\left(-M_{v,i}^{(n),C}{\eta}_{v_b,i}^{(n)}\mathbf{1}_{i}^{(n),C}+M_{v,i}^{(n),D}\mathbf{1}_{i}^{(n),D}\right)\right\}.
\end{split}
\end{equation}
%\end{alignat}
\end{small}
\end{lemma}
\begin{proof}
For $k=0$, the equation holds by definition. Let us therefore assume it holds for all $k \leq p$. For $k=p+1$ one then obtains
\begin{small}
\begin{alignat}{2}
{\eta}_{v_b,p+1}^{(n)}&={\eta}_{v_b,p}^{(n)}+\left(T_{+}^{(n)}({\eta}_{v_b,p}^{(n)})-{\eta}_{v_b,p}^{(n)}\right)\mathbf{1}_{p}^{(n),A}+\left(T_{-}^{(n)}({\eta}_{v_b,p}^{(n)})-{\eta}_{v_b,p}^{(n)}\right)\mathbf{1}_{p}^{(n),B}\nonumber\\
&\qquad\quad -\frac{\Delta v^{(n)}}{{\Delta x^{(n)}}} \cdot M_{v,p}^{(n),C}{\eta}_{v_b,p}^{(n)}\mathbf{1}_{p}^{(n),C}+ \frac{\Delta v^{(n)}}{{\Delta x^{(n)}}} \cdot M_{v,p}^{(n),D}\mathbf{1}_{p}^{(n),D}\nonumber\\
&=\left(\left(T_{+}^{(n)}\right)^{\mathbf{1}_{p}^{(n),A}}\circ \left(T_{-}^{(n)}\right)^{\mathbf{1}_{p}^{B}}\right)\left({\eta}_{v_b,p}^{(n)}\right)-M_{p}^{(n),C}{\eta}_{v_b,p}^{(n)}\mathbf{1}_{p}^{(n),C}+M_{p}^{(n),D}\mathbf{1}_{p}^{(n),D}\nonumber\\
&=\left(\left(T_{+}^{(n)}\right)^{\sum_{i=0}^{p}\mathbf{1}_{i}^{(n),A}}\circ\left(T_{-}^{(n)}\right)^{\sum_{i=0}^{p}\mathbf{1}_{i}^{(n),B}} \right)\left(v_{b,0}^{(n)}\right)\nonumber\\
&+\frac{\Delta v^{(n)}}{{\Delta x^{(n)}}} \cdot \sum_{i=0}^{p}\left\{\left(\left(T_{+}^{(n)}\right)^{\sum_{j=i+1}^{p}\mathbf{1}_{j}^{(n),A}}\circ\left(T_{-}^{(n)}\right)^{\sum_{j=i+1}^{p}\mathbf{1}_{j}^{(n),B}} \right)\left(-M_{v,i}^{(n),C}{\eta}_{v_b,i}^{(n)}\mathbf{1}_{i}^{(n),C}+M_{v,i}^{(n),D}\mathbf{1}_{i}^{(n),D}\right)\right\}.
\nonumber
\end{alignat}
\end{small}
%\Halmos
\end{proof}

\subsection{Boundedness of volume densities}

Using the isometry property of the translation operator we deduce that the $L^2$ norm of the volume density function can be estimated from above by considering a model with only passive order placements. 
%In that case
%\begin{eqnarray*}
%	\eta_{v_{b},k+1}^{(n)} &=& \eta_{v_{b},0}^{(n)} + \frac{\Delta v^{(n)} }{\Delta x^{(n)} } \sum_{l=0}^k M^{(n),D}_{l} 
%%	\\
%%	&=& \eta_{v_{b},0}^{(n)} + \frac{\Delta v^{(n)} }{\Delta x^{(n)} } \sum_{l=0}^k \left(M^{(n),D}_{l} - \mathbb{E}[M^{(n),D}_{l}] \right) +
%%	\frac{\Delta v^{(n)} }{\Delta x^{(n)} } \sum_{l=0}^k \mathbb{E}[M^{(n),D}_{l}]
%\end{eqnarray*}
%and so uniformly in $k=0, ..., \frac{T}{\Delta t^{(n)}}$ (up to some additive constant that converges to zero in the $L^2$-norm):
%\begin{equation} \label{eq1}
%	\| \eta_{v_{b},k+1}^{(n)} \|_{L^2} \leq \| \eta_{v_{b},0}^{(n)} \|_{L^2} + \Delta v^{(n)} \left\| \sum_{l=0}^k \left( \frac{1}{\Delta x^{(n)}} M^{(n),D}_{l} - f^{(n),D}_{l} \right)  \right\|_{L^2} + T \|f^D\|_{L^2}.
%\end{equation}
%With this, we are now able to prove the following $L^2$-boundedness of the volume density functions.
{ In a similar way we can estimate the expected order book hight at any given price tick. More precisely, we have the following result.} 

\vspace{2mm}

\begin{lemma} \label{lemma-bounded}
The expected $L^2$-norm of the volume density function is uniformly bounded:
\begin{equation} \label{expected-norm}
	%\mathbb{E} \left\| \Delta v^{(n)} \cdot \left\{ \frac{1}{\Delta x^{(n)}} M^{(n),D}_{l} - f^{(n),D}_{l}] \right\}\right \|^2_{L^2}  \leq C \frac{(\Delta v^{(n)})^2}{\Delta x^{(n)}}.
	\sup_{n \in \mathbb{N}, \, k=0, ... ,\lfloor T/\Delta t^{(n)} \rfloor} \mathbb{E} \|\eta^{(n)}_{v_b,k} \|^2_{L^2} \leq C
\end{equation}
for some constant $C<\infty$. { Likewise, the expected order book hight is uniformly bounded, i.e. if we put $\eta^{(n)}_{v_b,k} = (\eta^{(n),j}_{v_b,k})_{j \in \mathbb{Z}}$, then :
\begin{equation} \label{expected-norm2}
	\sup_{j \in \mathbb{Z}, \, n \in \mathbb{N}, \, k=0, ... ,\lfloor T/\Delta t^{(n)} \rfloor}  \mathbb{E} |\eta^{(n),j}_{v_b,k} |^2 \leq C.
\end{equation}}
%
%
%
%Moreover, if
%\[
%	\Delta p^{(n)} = {\cal O}(n^{-1-\epsilon})
%\]
%for some $\epsilon > 0$, then the sequence of volume density functions is quasi bounded in $L^2$ in the sense that
%\[
%	\sup_{k = 0, ..., \lfloor T/\Delta t^{(n)} \rfloor} \| \eta^{(n)}_{v_b,k}\|_{L^2} \leq C + \kappa_n
%\]
%for some constant $C > 0$ and a sequence $\{\kappa_n\}$ of random variable that converges almost surely to zero.
\end{lemma}
\begin{proof}
{It is enough to consider a model with only order placements where $\omega^{D}_k = 1$ a.s. W.l.o.g. we may also assume that $|\pi^{D}_k| \leq 1$ a.s. and $\eta^{(n)}_{v_b,0} \equiv 0$. Furthermore we may as well use a representation of the volume densities in {\sl absolute} rather than {\sl relative} coordinates. In such a model, $\mathbb{E} \|\eta^{(n)}_{v_b,k} \|^2_{L^2}$ is of the form
\[
	\mathbb{E} \|\eta^{(n)}_{v_b,k} \|^2_{L^2} = \left( \frac{\Delta v^{(n)}}{\Delta x^{(n)}} \right)^2 \mathbb{E} 
	%\sum_{j = \lfloor -1/\Delta x^{(n)} \rfloor}^{\lfloor 1/\Delta x^{(n)} \rfloor} 
	\sum_{j \in \mathbb{Z}} \left( \sum_{i=1}^k a^{(n)}_{i,j} \right)^2 \cdot \Delta x^{(n)} 
\] 
where $a^{(n)}_{i,j} := {{\bf 1}}_{\{\pi^D_{i} \in [x^{(n)}_j, x^{(n)}_{j+1})\}}$ and the distribution of $\pi^D_i$ properly adjusted to account for the representation in absolute coordinates.  
%\omega^I_k \sum_{j=-\infty}^{\infty}\mathbf{1}_{\{\pi_{k}^{I}\in[x^{(n)}_j, x^{(n)}_{j+1})\}}(\cdot) \right], 
Since the random variables $\pi^D_i$ have compact support, only finitely many summands are non-zero and we may rearrange terms. Using conditional independence of the placement variables though time, this yields:
\begin{eqnarray*}
	\mathbb{E} \|\eta^{(n)}_{v_b,k} \|^2_{L^2} 
% &=&  \frac{\left( \Delta v^{(n)}\right)^2}{\Delta x^{(n)}}  
%	%\sum_{j = \lfloor -1/\Delta x^{(n)} \rfloor}^{\lfloor 1/\Delta x^{(n)} \rfloor} 
%	\sum_{j \in \mathbb{Z}} \sum_{i=1}^k \mathbb{E}  \left( a^{(n)}_{i,j} \right)^2  
%	+  \frac{\left( \Delta v^{(n)}\right)^2}{\Delta x^{(n)}}  
%	%\sum_{j = \lfloor -1/\Delta x^{(n)} \rfloor}^{\lfloor 1/\Delta x^{(n)} \rfloor} 
%	\sum_{j \in \mathbb{Z}} \sum_{i,i'=1, i \neq i'}^k \mathbb{E}  a^{(n)}_{i,j} \mathbb{E}  a^{(n)}_{i',j} \\
	= \frac{\left( \Delta v^{(n)}\right)^2}{\Delta x^{(n)}}  
	\sum_{i=1}^k \sum_{j \in \mathbb{Z}}  \mathbb{E}  \left( a^{(n)}_{i,j} \right)^2  
	+  \frac{\left( \Delta v^{(n)}\right)^2}{\Delta x^{(n)}}  
	 \sum_{i,i'=1, i \neq i'}^k \sum_{j \in \mathbb{Z}} \mathbb{E}  a^{(n)}_{i,j} \mathbb{E}  a^{(n)}_{i',j}. 
\end{eqnarray*}
Using the fact that no placements take place at price levels with a distance of more than $1$ from the prevailing best bid/ask price: 
\begin{eqnarray*}
	\mathbb{E}  \left( a^{(n)}_{i,j} \right)^2 & \leq & \|f^D\|_{\infty} \Delta x^{(n)} {\bf 1}_{\{|j-\eta^{(n)}_{\gamma,i}| \leq 1\}}  \\
	\left| \mathbb{E}  a^{(n)}_{i,j} \mathbb{E}  a^{(n)}_{i',j}  \right| & \leq & \|f^D\|^2_{\infty} \left( \Delta x^{(n)} \right)^2 {\bf 1}_{\{ |j-\eta^{(n)}_{\gamma,i}| \leq 1\}} {\bf 1}_{\{|j-\eta^{(n)}_{\gamma,i'}| \leq 1 \}}. 
\end{eqnarray*}
In particular, the inner sums extend over at most $\frac{2}{\Delta x^{(n)}} + 1$ terms. As a result  
our scaling assumptions guarantee that
\[
	\sup_{n,k=1, ... ,\lfloor T/\Delta t^{(n)} \rfloor } \mathbb{E} \|\eta^{(n)}_{v_b,k} \|^2_{L^2} < \infty.
\]
The second assertion follows analogously as 
\[
	\mathbb{E} |\eta^{(n),j}_{v_b,k} |^2 = \frac{\left(  \Delta v^{(n)} \right)^2}{\Delta x^{(n)}}  \mathbb{E} 
	%\sum_{j = \lfloor -1/\Delta x^{(n)} \rfloor}^{\lfloor 1/\Delta x^{(n)} \rfloor} 
	\left( \sum_{i=1}^k a^{(n)}_{i,j} \right)^2.
\] 
 }
%
%Since $f^{(n),D}$ is bonded and belongs to $L^2$ a direct computation shows that
%\begin{equation} \label{expected-norm2}
%	\left\| \Delta v^{(n)} \cdot \left\{ \frac{1}{\Delta x^{(n)}} M^{(n),D}_{l} - f^{(n),D}_{l} \right\}\right \|_{L^2}  \leq C \frac{\Delta v^{(n)}}{\sqrt{\Delta x^{(n)}}}
%\end{equation}
%for some constant $C > 0$ and all sufficiently large $n \in {\mathbb N}$. As for the second assertion, notice that (\ref{eq1}) along with Assumption \ref{assumption-scaling} yields
%\begin{eqnarray*}
%	\frac{1}{\Delta t^{(n)}} \mathbb{E} \left\| \Delta v^{(n)} \cdot \left\{ \frac{1}{\Delta x^{(n)}} M^{(n),D}_{l} - f^{(n),D}_{l}] \right\}\right \|^2_{L^2} & \leq &
%	C \frac{(\Delta v^{(n)})^2}{\Delta x^{(n)} \Delta t^{(n)}} \\
%	&=& \frac{\Delta v^{(n)}}{\Delta x^{(n)}} \\
%	&=& \Delta p^{(n)} \\
%	&=& {\cal O}\left(\frac{1}{n^{1+\epsilon}} \right).
%\end{eqnarray*}
%Hence the result follows from Corollary \ref{corollary-appendix}.
%\Halmos
\end{proof}

%%%%%%%%%%%%%%%%%%%%%%%%%%%%%%%%%%
%%%%%%%%%%%%%%%%%%%%%%%%%%%%%%%%%%
%%%%%%%%%%%%%%%%%%%%%%%%%%%%%%%%%%

\subsection{Norm estimates}

 The next result will be used to prove a Lipschitz continuity property of the grid-point approximation of the limiting PDE.

\begin{lemma} \label{lemma-shift}
There exists a constant $C< 0$ such that for all $n \in \mathbb{N}$ and $k=0, ..., \lfloor T/\Delta t^{(n)} \rfloor$:
\[
	\left\|  \mathbb{E} \left[ T^{(n)}_\pm (\eta^{(n)}_{v_b,k}) - \eta^{(n)}_{v_b,k} \right] \right\|_{L^2} \leq C \cdot \Delta x^{(n)}.
\]
\end{lemma}
\begin{proof}
Using Lemma \ref{lemma-B1} and the linearity of the translation operator $T_{+}^{(n)}$ it follows that a.s.
\begin{alignat}{2}
&T_{+}^{(n)}\left({\eta}_{v_b,k}^{(n)}\right)-{\eta}_{v_b,k}^{(n)}\nonumber\\
&=\left(\left(T_{+}^{(n)}\right)^{\sum_{i=0}^{k-1}\mathbf{1}_{i}^{A}}\circ\left(T_{-}^{(n)}\right)^{\sum_{i=0}^{k-1}\mathbf{1}_{i}^{B}} \right)\left(T_{+}^{(n)}\left(v_{b,0}^{(n)}\right)-v_{b,0}^{(n)}\right) \label{eq:DeltaEtaInd3} \\
&\qquad+ \frac{\Delta v^{(n)}}{{\Delta x^{(n)}}} \cdot \sum_{i=0}^{k-1}\left(\left(T_{+}^{(n)}\right)^{\sum_{j=i+1}^{k-1}\mathbf{1}_{j}^{A}}\circ\left(T_{-}^{(n)}\right)^{\sum_{j=i+1}^{k-1}\mathbf{1}_{j}^{B}} \right)\left(\left[T_{+}^{(n)}\left(M_{i}^{(n),D}\right)-M_{i}^{(n),D}\right]\mathbf{1}_{i}^{D}\right) \label{eq:DeltaEtaInd2} \\
&\qquad- \frac{\Delta v^{(n)}}{{\Delta x^{(n)}}} \cdot \sum_{i=0}^{k-1}\left(\left(T_{+}^{(n)}\right)^{\sum_{j=i+1}^{k-1}\mathbf{1}_{j}^{A}}\circ\left(T_{-}^{(n)}\right)^{\sum_{j=i+1}^{k-1}\mathbf{1}_{j}^{B}} \right)\left(\left[T_{+}^{(n)}\left(M_{i}^{(n),C}{\eta}_{v_b,i}^{(n)}\right)-M_{i}^{(n),C}{\eta}_{v_b,i}^{(n)}\right] \mathbf{1}_{i}^{C}\right)
\label{eq:DeltaEtaInd}
\end{alignat}
Taking the expected value and norms in (\ref{eq:DeltaEtaInd}) we find:
\begin{alignat}{2}
&\left\|\mathbb{E}\left[T_{+}^{(n)}\left({\eta}_{v_b,k}^{(n)}\right)-{\eta}_{v_b,k}^{(n)}\right]\right\|_{L^{2}}\nonumber\\
&\leq\left\|T_{+}^{(n)}\left(v_{b,0}^{(n)}\right)-v_{b,0}^{(n)}\right\|_{L^{2}}+ { 
\frac{\Delta v^{(n)}}{\Delta x^{(n)}}} \sum_{i=0}^{k-1} \left\|T_{+}^{(n)}\left(\mathbb{E} \left[ M^{(n),D}_i \right] \right)-\mathbb{E}\left[ M^{(n),D}_ i \right]\right\|_{L^{2}}. \nonumber \\
&\qquad\qquad\qquad\qquad\qquad\quad + { \frac{\Delta v^{(n)}}{\Delta x^{(n)}}} \sum_{i=0}^{k-1}\left\| T_{+}^{(n)}\left(\mathbb{E}\left[M_{v,i}^{(n),C}{\eta}_{v_b,i}^{(n)}\right]\right)- \mathbb{E}\left[M_{v,i}^{(n),C}{\eta}_{v_b,i}^{(n)}\right]\right\|_{L^{2}}.
\end{alignat}
%&\leq \left\|T_{+}^{(n)}\left(v_{b,0}^{(n)}\right)-v_{b,0}^{(n)}\right\|_{L^{2}}+\Delta v^{(n)}\sum_{i=0}^{k-1}\left\| \mathbb{E}\left[ T_{+}^{(n)}\left({\eta}_{v_b,i}^{(n)}\right)- {\eta}_{v_b,i}^{(n)}\right]\right\|_{L^{2}}\nonumber\\
%&\qquad\qquad\qquad\qquad\qquad\quad+\Delta v^{(n)}\sum_{i=0}^{k-1} \left\|T_{+}^{(n)}\left(f^{(n),D}\right)-f^{(n),D}\right\|_{L^{2}}
%\label{ineq:lnqIntro}
%\end{alignat}
By Assumptions \ref{Assumption-1} and \ref{Assumption-3} there exists a constant $K<\infty$ such that $\left\|T_{+}^{(n)}\left(v_{b,0}^{(n)}\right)-v_{b,0}^{(n)}\right\|_{L^{2}} \leq K \Delta x^{(n)} $
%\begin{eqnarray*}
%	\left\|T_{+}^{(n)}\left(v_{b,0}^{(n)}\right)-v_{b,0}^{(n)}\right\|_{L^{2}} & \leq & K \Delta x^{(n)} 
%\end{eqnarray*}
and
\begin{eqnarray*}
	\frac{1}{\Delta x^{(n)}} \left\|T_{+}^{(n)}\left(\mathbb{E} \left[ M^{(n),D}_i \right] \right)-
	\mathbb{E}\left[ M^{(n),D}_ i \right]		\right\|_{L^{2}} 
	  &=& \left\|T_{+}^{(n)}\left( f^{(n),D}  \right)- f^{(n),D} 	\right\|_{L^{2}}  \\
	 & \leq & K \Delta x^{(n)}.
\end{eqnarray*}
As for the cancellation terms, independence of the event dynamics from the standing volumes yield:
\begin{eqnarray*}
	\frac{1}{\Delta x^{(n)}} T_{+}^{(n)}\left( \mathbb{E} \left[ M^{(n),C}_i {\eta}_{v_b,i}^{(n)} \right] \right) & = &
	T_{+}^{(n)} \left( f^{(n),C} \mathbb{E}  {\eta}_{v_b,i}^{(n)} \right),\\ 
	\frac{1}{\Delta x^{(n)}}  \mathbb{E} \left[ M^{(n),C}_i {\eta}_{v_b,i}^{(n)} \right] & = &  
	f^{(n),C} \mathbb{E}  {\eta}_{v_b,i}^{(n)}. 
\end{eqnarray*}
In view of the second assertion of Lemma \ref{lemma-bounded} and using the fact that $f^{(n),C}$ is bounded along with Assumption  \ref{Assumption-3} we find a constant $K< \infty$ such that:
\begin{eqnarray*}
	& & \frac{1}{\Delta x^{(n)}} \left\|T_{+}^{(n)}\left(\mathbb{E} \left[ M^{(n),C}_i {\eta}_{v_b,i}^{(n)} \right] \right)-
	\mathbb{E}\left[ M^{(n),C}_ i {\eta}_{v_b,i}^{(n)} \right]	\right\|_{L^{2}} \\
	& \leq & \left \| T^{(n)}_+\left( f^{(n),C} \right) \mathbb{E} \left[ T_{+}^{(n)} \left( {\eta}_{v_b,i}^{(n)} \right) - {\eta}_{v_b,i}^{(n)} \right] \right \|_{L^2} \\
	& & \, +  \left\|  \left( T^{(n)}\left( f^{(n),C} \right) - f^{(n),C} \right) \mathbb{E} \eta^{(n)}_{v_b,i} \right\|_{L^2} \\
	& \leq & K \left( \mathbb{E} \| T_{+}^{(n)}\left({\eta}_{v_b,i}^{(n)}\right)-{\eta}_{v_b,i}^{(n)} \|_{L^2} + \Delta x^{(n)}  \right).
\end{eqnarray*}
Altogether, we arrive at the following estimate:
%Since the $L^2$-norms of the volume density functions are uniformly bounded, we use Assumptions \ref{Assumption-1} and \ref{Assumption-2} to deduce that for some generic constant $K > 0$:
\begin{alignat}{2}
\left\|\mathbb{E}\left[T_{+}^{(n)}\left({\eta}_{v_b,k}^{(n)}\right)-{\eta}_{v_b,k}^{(n)}\right]\right\|_{L^{2}}
%&\leq K \Delta x^{(n)}+ \Delta v^{(n)} \sum_{i=0}^{k-1}\left\|\mathbb{E}\left[T_{+}^{(n)}\left({\eta}_{v_b,i}^{(n)}\right)-{\eta}_{v_b,i}^{(n)}\right]\right\|_{L^{2}} + K \Delta x^{(n)}\nonumber\\
\leq K \Delta x^{(n)} +K \Delta v^{(n)} \sum_{i=0}^{\lfloor T / \Delta t^{(n)}\rfloor}\left\|\mathbb{E}\left[T_{+}^{(n)}\left({\eta}_{v_b,i}^{(n)}\right)-{\eta}_{v_b,i}^{(n)}\right]\right\|_{L^{2}}.
\label{ineq:ExpDeltaEta}
\end{alignat}
Hence, it follows from Gronwall's lemma that
\[
	\sup_{k = 0, ..., T / \Delta t^{(n)} }  \left\|\mathbb{E}\left[T_{+}^{(n)}\left({\eta}_{v_b,k}^{(n)}\right)-{\eta}_{v_b,k}^{(n)}\right]\right\|_{L^{2}} = \mathcal{O}\left( \Delta x^{(n)} \right).
\]
%\end{itemize}
%\Halmos
\end{proof}

%\hspace{1mm}

\begin{corollary} \label{un-lip}
There exists a constant $C > 0$ such that
\[
	\| u^{(n)}(t, \cdot + \Delta x^{(n)} ) - u^{(n)}(t, \cdot ) \|_{L^2} \leq C \cdot \Delta x^{(n)}.
\]
Moreover,
\[
	\sup_{n \in \mathbb{N}, t \in [0,T]} \| u^{(n)}(t;\cdot) \|_{L^2} < \infty.
\]
\end{corollary}
\begin{proof}
In order to establish the first assertion we represent the functions $u^{(n)}$ as
\begin{equation} \label{eqn2}
	u^{(n)}_k = \mathbb{E} {\zeta}^{(n)}_k
\end{equation}	
where
\begin{alignat}{2}
\begin{cases}
	{\zeta}_{k+1}^{(n)}&:=\zeta_{k}^{(n)}+ \mathcal{D}_{v,k}^{(n)}\left( \gamma(t^{(n)}_k), {\zeta}^{(n)}_k \right) \\
\label{def:UTildenjkRec}\\
{\zeta}_{0}^{(n),j}&:=v_{0}^{(n),j}
\end{cases}.
\end{alignat}
For $t \in [k \cdot \Delta t^{(n)}, (k+1) \cdot \Delta^{(n)} )$ the preceding lemma then implies:
\[
	\| u^{(n)}(t, \cdot \pm \Delta x^{(n)} ) - u^{(n)}(t, \cdot ) \|_{L^2} =
	\left \| \mathbb{E} \left[ T^{(n)}_\pm \left( \zeta^{(n)}_k \right) - \zeta^{(n)}_k \right] \right \|_{L^2} \leq C \Delta x^{(n)}.
\]
The second assertion follows from (\ref{eqn2}) together with Lemma \ref{lemma-bounded}:
\begin{eqnarray*}
	\| u^{(n)}(t^{(n)}_k;\cdot) \|_{L^2} = \| \mathbb{E} {\zeta}^{(n)}_k \|_{L^2} \leq \mathbb{E} \| {\zeta}^{(n)}_k \|_{L^2}  \leq C.
\end{eqnarray*}
%\Halmos
\end{proof}

Using Lipschitz continuity of $f^I$ along with the point wise shift estimate of the initial volume densities of Assumption \ref{Assumption-1} the following result can be established by analogy to Corollary \ref{lemma-shift}.

\begin{corollary} \label{lemma-shift2}
%We have that
\begin{equation} \label{expected-norm3}
	\sup_{n \in \mathbb{N}, \, j \in \mathbb{Z}, \, k=0, ... ,\lfloor T/\Delta t^{(n)} \rfloor} \left| \mathbb{E} \left[ \eta^{(n),j+1}_{v_b,k} - \eta^{(n),j}_{v_b,k} \right]  \right | = \mathcal{O}(\Delta x^{(n)}). 
\end{equation}
\end{corollary}

%\subsection{Some useful norm-estimates}

We close this appendix with norm estimates which are key to the proof of Proposition \ref{prop:etaNuTildeConv}.

\begin{lemma} \label{last-lemma}
	The following norm estimates hold:
	\begin{eqnarray*}
		\sup_{n,k} \mathbb{E} \left[ \left\| {\bf 1}_{D_k,H_k} \frac{\Delta v^{(n)}}{\Delta x^{(n)}} M^{(n),D,H}_{v,k} - 
		 \mathbb{E}\left[ {\bf 1}_{D_k,H_k} \frac{\Delta v^{(n)}}{\Delta x^{(n)}} M^{(n),D,H}_{v,k} \right] \right\|^2_{L^2} 
		 \right] &=& \mathcal{O}\left( (\Delta v^{(n)})^2 \right) \\
		\sup_{n,k} \mathbb{E} \left[ \left\| {\bf 1}_{C_k,G_k} \frac{\Delta v^{(n)}}{\Delta x^{(n)}} M^{(n),C,G}_{v,k} \eta^{(n)}_{v,k} - 
		 \mathbb{E}\left[ {\bf 1}_{C_k,G_k} \frac{\Delta v^{(n)}}{\Delta x^{(n)}} M^{(n),C,G}_{v,k} \eta^{(n)}_{v,k} \right] \right\|^2_{L^2} 
		 \right] &=& \mathcal{O}\left( (\Delta v^{(n)})^2 \right) \\
%	\end{eqnarray*}
%	and
%	\begin{eqnarray*}
		 \sup_{n,k} \mathbb{E} \left[ \left\| {\bf 1}_{A_k,B_k}  \left(T_{\pm}^{(n)}\left({\eta}_{v_b,k}^{(n)}\right)-{\eta}_{v_b,k}^{(n)} \right) - 
		 \mathbb{E}\left[ {\bf 1}_{A_k,B_k} \left(T_{\pm}^{(n)}\left({\eta}_{v_b,k}^{(n)}\right)-{\eta}_{v_b,k}^{(n)} \right) \right] \right\|^2_{L^2} 
		 \right] &=& \mathcal{O}\left( (\Delta v^{(n)})^{2 \alpha} + (\Delta v^{(n)})^{2 - \alpha}\right)  \\
		  \sup_{n,k} \mathbb{E} \left[ \left\| {\bf 1}_{E_k,F_k}  \left(T_{\pm}^{(n)}\left({\eta}_{v_s,k}^{(n)}\right)-{\eta}_{v_s,k}^{(n)} \right) - 
		 \mathbb{E}\left[ {\bf 1}_{E_k,F_k} \left(T_{\pm}^{(n)}\left({\eta}_{v_s,k}^{(n)}\right)-{\eta}_{v_s,k}^{(n)} \right) \right] \right\|^2_{L^2} 
		 \right] &=& \mathcal{O}\left( (\Delta v^{(n)})^{2 \alpha} + (\Delta v^{(n)})^{2 - \alpha}\right). 
	\end{eqnarray*}
\end{lemma}
\begin{proof}
The first two estimates follow from boundedness of the density functions $f^I$ along with independence of the event dynamics from volumes and (\ref{expected-norm}). In order to establish the third and fourth estimate we have to prove that 
\[
	\sup_{n,k} \left\{  \Delta p^{(n)} \cdot \mathbb{E} \left[ \| T_{+}^{(n)}\left({\eta}_{v_b,k}^{(n)}\right)-{\eta}_{v_b,k}^{(n)} \|_{L^2}^2  \right] \right\}
%		 \mathbb{E}\left[ {\bf 1}_{A_k} \left(T_{+}^{(n)}\left({\eta}_{v_b,k}^{(n)}\right)-{\eta}_{v_b,k}^{(n)} \right) \right] \right\|^2_{L^2} 
%		 \right] 
		 = \mathcal{O}\left( (\Delta v^{(n)})^{2 \alpha} + (\Delta v^{(n)})^{2 - \alpha}\right).
\]
To this end we use a representation of  $T^{(n)}_+ \eta^{(n)}_{v_b,k} - \eta^{(n)}_{v_b,k}$ as in Lemma \ref{lemma-shift} but in {\sl absolute} rather than {\sl relative} coordinates. This means that the shift terms drop out of the representation but the probabilities of placements and cancellations need to be properly adjusted. 

Assumption \ref{Assumption-1} allows us to bound the impact of the initial condition (\ref{eq:DeltaEtaInd3}) by a term of the order $(\Delta x^{(n)})^2.$
%\[
%	 \|T^{(n)}_+ \eta^{(n)}_{v_b,0} - \eta^{(n)}_{v_b,0}\|^2_{L^2} \leq K (\Delta x^{(n)})^2.
%\] 
To compute the norm of the sum in (\ref{eq:DeltaEtaInd2}) we need to compute a term of the form
\[
	\left( \frac{\Delta v^{(n)}}{\Delta x^{(n)}} \right)^2 \mathbb{E}  
	\sum_{j \in \mathbb{Z}} \left( \sum_{i=1}^k a^{(n)}_{i,j} \right)^2 \cdot \Delta x^{(n)} 
\] 
where 
\[
	a^{(n)}_{i,j} =
	\left\{ 
	\begin{array}{ll} 1 & \mbox{if } \pi^D_i \in [(j-1)\cdot \Delta x^{(n)}, j \cdot \Delta x^{(n)}) \\
	-1 & \mbox{if } \pi^D_i \in [j \cdot \Delta x^{(n)}, (j+1) \cdot \Delta x^{(n)}) \\
	0 & \mbox{else}
	\end{array}
	\right. .
\]
In particular there exists a constant $K< \infty$ such that
\begin{eqnarray*}
	\mathbb{E}  \left( a^{(n)}_{i,j} \right)^2 & \leq & K \Delta x^{(n)} {\bf 1}_{\{|j-\eta^{(n)}_{\gamma,i}| \leq 1\}}  \\
	\left| \mathbb{E}  a^{(n)}_{i,j}  \right| 
%	& \leq & \left| \mathbb{P}[\pi_i \in [(j-1)\cdot \Delta x^{(n)}, j \cdot \Delta x^{(n)}) ] -  		
%	\mathbb{P}[\pi_i \in [(j-1)\cdot \Delta x^{(n)}, j \cdot \Delta x^{(n)}) ]\right| {\bf 1}_{\{ |j-\eta^{(n)}_{\gamma,i}| \leq 1\}} {\bf 1}_{\{|j-\eta^{(n)}_{\gamma,i'}| \leq 1 \}} \\
	& \leq & K \left( \Delta x^{(n)} \right)^2 {\bf 1}_{\{|j-\eta^{(n)}_{\gamma,i}| \leq 1 \}}
%	
%	\left( \Delta x^{(n)} \right)^2 |\mathbb{E}[\eta] {\bf 1}_{\{ |j-\eta^{(n)}_{\gamma,i}| \leq 1\}} {\bf 1}_{\{|j-\eta^{(n)}_{\gamma,i'}| \leq 1 \}}. 
\end{eqnarray*}
where the second inequality follows from (\ref{Lip-estimate}); the indicator functions account for the representation of volumes in {\sl absolute} coordinates. 
Using the fact that the random variables $\pi^D_k$ have compact support and that events are conditionally independent through time, we can now argue as in the proof of Lemma \ref{lemma-bounded} to deduce that:
\begin{eqnarray*}
	\Delta p^{(n)} \left( \frac{\Delta v^{(n)}}{\Delta x^{(n)}} \right)^2 \mathbb{E}  
	\sum_{j \in \mathbb{Z}} \left( \sum_{i=1}^k a^{(n)}_{i,j} \right)^2 \cdot \Delta x^{(n)} 
	& \leq & K \Delta p^{(n)} \frac{(\Delta v^{(n)})^2}{\Delta x^{(n)}} \left( \frac{1}{\Delta v^{(n)}} + \frac{1}{(\Delta v^{(n)})^2} \frac{1}{\Delta x^{(n)} } (\Delta x^{(n)})^4 \right) \\
	& = & K  \left((\Delta v^{(n)})^{2\alpha} + (\Delta v^{(n)})^{2- \alpha} \right).
\end{eqnarray*}

To compute the norm of the sum in (\ref{eq:DeltaEtaInd}) we need to compute a similar term, but with $a^{(n)}_{i,j}$ replaced by 
\[
	b^{(n)}_{i,j} =
	\left\{ 
	\begin{array}{ll} \eta^{(n),j}_{v_b,k} & \mbox{if } \pi^D_i \in [(j-1)\cdot \Delta x^{(n)}, j \cdot \Delta x^{(n)}) \\
	-\eta^{(n),j-1}_{v_b,k} & \mbox{if } \pi^D_i \in [ j \cdot \Delta x^{(n)}, (j+1) \cdot \Delta x^{(n)}) \\
	0 & \mbox{else}
	\end{array}
	\right. .
\]
Using (\ref{expected-norm2}) we have again that $\mathbb{E}  \left( a^{(n)}_{i,j} \right)^2 \leq K \Delta x^{(n)} {\bf 1}_{\{|j-\eta^{(n)}_{\gamma,i}| \leq 1\}}$. Using independence of the event dynamics from volumes along with (\ref{expected-norm3}) and Lipschitz continuity of $f^C$ we also obtain a constant $K < \infty$ such that:
\begin{eqnarray*}
	\left| \mathbb{E}  a^{(n)}_{i,j}  \right| & \leq & 
	\left| 
	\mathbb{P}[\pi^D_i \in [j\cdot \Delta x^{(n)}, (j+1) \cdot \Delta x^{(n)})] \cdot \mathbb{E} \eta^{(n),j}_{v_b,i} -
	\mathbb{P}[\pi^D_i \in [(j-1) \cdot \Delta x^{(n)}, j \cdot \Delta x^{(n)})] \cdot \mathbb{E} \eta^{(n),j-1}_{v_b,i}
	\right| \\ 
	& \leq & K \left\{ \left( \Delta x^{(n)} \right)^2 + |\mathbb{E} \eta^{(n),j}_{v_b,i} - \mathbb{E} \eta^{(n),j-1}_{v_b,i}| 
	\cdot \mathbb{P}[\pi^D_i \in [(j-1) \cdot \Delta x^{(n)}, j \cdot \Delta x^{(n)})]  \right\}.
	% K \left( \Delta x^{(n)} \right)^2 {\bf 1}_{\{|j-\eta^{(n)}_{\gamma,i'}| \leq 1 \}}
\end{eqnarray*}
Hence it follows from Corollary \ref{lemma-shift2} that
\[
	\left| \mathbb{E}  a^{(n)}_{i,j}  \right|  \leq K \left( \Delta x^{(n)} \right)^2 
\]
and the assertion follows as in the case of placements. 
\end{proof}
\end{appendix}

% Acknowledgments here
\section*{Acknowledgments.}
Most of this work was done while the first author was visiting CEREMADE; grateful acknowledgment is made for hospitality and financial support. Both authors acknowledge financial support through the SFB 649 ``Economic Risk'' and Deutsche Bank AG. We thank D\"orte Kreher for pointing out a gap in an earlier version and seminar participants at various institutions for valuable comments and suggestions.

% References here (outcomment the appropriate case)

% CASE 1: BiBTeX used to constantly update the references
%   (while the paper is being written).
%\bibliographystyle{ormsv080} % outcomment this and next line in Case 1
%\bibliography{OBHPDE} % if more than one, comma separated
% CASE 2: BiBTeX used to generate mypaper.bbl (to be further fine tuned)
%\input{mypaper.bbl} % outcomment this line in Case 2

\bibliographystyle{plain}
\bibliography{OBHPDE}

\end{document}